\def\draft{0} 
\documentclass[11pt]{article}

\usepackage{amsfonts,amsmath,amssymb,amsthm,boxedminipage,color,url,fullpage}
\usepackage[pdfstartview=FitH,colorlinks,linkcolor=blue,filecolor=blue,citecolor=blue,urlcolor=blue]{hyperref} 
\usepackage{enumerate,paralist}
\usepackage[numbers]{natbib} 
\usepackage[labelfont=bf]{caption}
\usepackage{aliascnt,cleveref}
\usepackage{xspace}
\usepackage{graphicx}
\usepackage{tikz}

\usepackage{multicol}
\newcommand{\remove}[1]{}
\newcommand{\Ignore}{\remove}

\newcommand{\Draft}[1]{\ifnum\draft=1\texttt{ #1} \fi}

\ifnum\draft=1
 \newcommand{\authnote}[2]{{\bf [{\color{red} #1's Note:} {\color{blue} #2}]}}
\else
 \newcommand{\authnote}[2]{}
\fi

\newcommand{\sdotfill}{\textcolor[rgb]{0.8,0.8,0.8}{\dotfill}} 

\newenvironment{protocol}{\begin{proto}}{\vspace{-\topsep}\sdotfill\end{proto}}
\newenvironment{algorithm}{\begin{algo}}{\vspace{-\topsep}\sdotfill\end{algo}}

\newcommand{\resp}{resp.,\ }
\newcommand{\ie} {i.e.,\ }
\newcommand{\eg} {e.g.,\ }
\newcommand{\wrt} {with respect to\ }
\newcommand{\wlg} {without loss of generality\xspace}
\newcommand{\cf}{{cf.,\ }}

\newcommand{\ceil}[1]{\left\lceil #1 \right\rceil}

\newcommand{\exec}[1]{\left \langle #1 \right \rangle}

\newcommand{\set}[1]{\ens{#1}}
\newcommand{\setb}[1]{\bigl\{#1\bigl\}}

\newcommand{\floor}[1]{\left \lfloor#1 \right \rfloor}
\newcommand{\vect}[1]{\left( #1 \right)}
\newcommand{\eqdef}{:=}

\newcommand{\N}{{\mathbb{N}}}

\newcommand{\zo}{{\{0,1\}}}
\newcommand{\zn}{{\zo^n}}
\newcommand{\zm}{{\zo^m}}

\newcommand{\deffont}\emph{}

\newcommand{\eps}{\varepsilon}

\newcommand{\la}{\gets}

\newcommand{\poly}{\operatorname{poly}}

\newcommand{\Exp}{\operatorname*{E}}

\newcommand{\Hmin}{\operatorname{H_{\infty}}}

\newcommand{\negl}{\operatorname{neg}}

\newcommand{\Supp}{\operatorname{Supp}}

\newcommand{\MathAlg}[1]{\mathsf{#1}}

\renewcommand{\cref}{\Cref}

\newtheorem{theorem}{Theorem}

\newaliascnt{lemma}{theorem}
\newtheorem{lemma}[lemma]{Lemma}
\aliascntresetthe{lemma}
\crefname{lemma}{Lemma}{Lemmas}

\newaliascnt{problem}{theorem}

\aliascntresetthe{problem}
\crefname{problem}{Open Problem}{Open Problems}

\newaliascnt{claim}{theorem}
\newtheorem{claim}[claim]{Claim}
\aliascntresetthe{claim}
\crefname{claim}{Claim}{Claims}

\newaliascnt{corollary}{theorem}
\newtheorem{corollary}[corollary]{Corollary}
\aliascntresetthe{corollary}
\crefname{corollary}{Corollary}{Corollaries}

\newaliascnt{construction}{theorem}

\aliascntresetthe{construction}
\crefname{construction}{Construction}{Constructions}

\newaliascnt{fact}{theorem}
\newtheorem{fact}[fact]{Fact}
\aliascntresetthe{fact}
\crefname{fact}{Fact}{Facts}

\newaliascnt{proposition}{theorem}
\newtheorem{proposition}[proposition]{Proposition}
\aliascntresetthe{proposition}
\crefname{proposition}{Proposition}{Propositions}

\newaliascnt{conjecture}{theorem}

\aliascntresetthe{conjecture}
\crefname{conjecture}{Conjecture}{Conjectures}

\newaliascnt{notation}{theorem}

\aliascntresetthe{notation}
\crefname{notation}{Notation}{Notation}

\newaliascnt{definition}{theorem}
\newtheorem{definition}[definition]{Definition}
\aliascntresetthe{definition}
\crefname{definition}{Definition}{Definitions}

\newaliascnt{remark}{theorem}
\newtheorem{remark}[remark]{Remark}
\aliascntresetthe{remark}
\crefname{remark}{Remark}{Remarks}

\newaliascnt{example}{theorem}

\aliascntresetthe{example}
\crefname{example}{Example}{Examples}

\newaliascnt{proto}{theorem}

\newtheorem{proto}[proto]{Protocol}

\aliascntresetthe{proto}
\crefname{proto}{protocol}{protocols}

\newaliascnt{algo}{theorem}
\newtheorem{algo}[algo]{Algorithm}
\aliascntresetthe{algo}
\crefname{algo}{algorithm}{algorithms}

\newaliascnt{expr}{theorem}

\aliascntresetthe{expr}
\crefname{experiment}{experiment}{experiments}

\crefname{descriptiona}{description}{descriptions}

\def\FullBox{$\Box$}
\def\qed{\ifmmode\qquad\FullBox\else{\unskip\nobreak\hfil
\penalty50\hskip1em\null\nobreak\hfil\FullBox
\parfillskip=0pt\finalhyphendemerits=0\endgraf}\fi}

\def\qedsketch{\ifmmode\Box\else{\unskip\nobreak\hfil
\penalty50\hskip1em\null\nobreak\hfil$\Box$
\parfillskip=0pt\finalhyphendemerits=0\endgraf}\fi}

\newcommand{\Tau}{\mathrm{T}} 

\newcommand{\ex}[2]{\Exp_{#1}\left[#2\right]}
\newcommand{\exx}[1]{\Exp\left[#1\right]}

\newcommand{\Ex}{\Exp}
\newcommand{\pr}[1]{\Pr\left[#1\right]}
\newcommand{\prb}[1]{\Pr\bigl[#1\bigl]}
\newcommand{\ppr}[2]{\Pr_{#1}\left[#2\right]}

\newcommand{\Sc}{\ensuremath{ \MathAlg{S}}\xspace}
\newcommand{\Rc}{\ensuremath{ \MathAlg{R}}\xspace}

\newcommand{\Server}{\ensuremath{ \MathAlg{Server}}\xspace}
\newcommand{\User}{\ensuremath{ \MathAlg{User}}\xspace}
\newcommand{\Servers}{\ensuremath{{\widetilde{\Server}}}\xspace}

\newcommand{\Ss}{\ensuremath{{\widetilde{\Sc}}}\xspace}
 \newcommand{\Rs}{\ensuremath{{\widetilde{\Rc}}}\xspace}

\newcommand{\Ac}{\ensuremath{ \MathAlg{A}}\xspace}
\newcommand{\As}{\ensuremath{{\widetilde{\MathAlg{A}}}}\xspace}
\newcommand{\Bc}{\ensuremath{ \MathAlg{B}}\xspace}
\newcommand{\Ec}{\ensuremath{ \MathAlg{E}}\xspace}

\newcommand{\Mc}{\ensuremath{ \MathAlg{M}}\xspace}
\newcommand{\Dc}{\ensuremath{ \MathAlg{D}}\xspace}
\newcommand{\Cc}{\ensuremath{ \MathAlg{C}}\xspace}
\newcommand{\Inv}{\ensuremath{ \MathAlg{Inv}}\xspace}
\newcommand{\tInv}{\ensuremath{ \widetilde{\MathAlg{Inv}}}\xspace}

\newcommand{\Invd}{\ensuremath{ \MathAlg{Inv}_D}\xspace}
\newcommand{\Invf}{\ensuremath{ \MathAlg{Inv}_f}\xspace}

\newcommand{\cI}{\mathcal{I}}
\newcommand{\cY}{\mathcal{Y}}

\newcommand{\cX}{\mathcal{X}}

\newcommand{\cZ}{\mathcal{Z}}

\newcommand{\cS}{\mathcal{S}}

\newcommand{\calO}{{\mathcal{O}}}

\newcommand{\G}{G}

\newcommand{\ens}[1]{\left\{#1\right\}}
\newcommand{\size}[1]{\left|#1\right|}

\newcommand{\out}{\operatorname{out}}

\newcommand{\trans}{{\sf trans}}

\newcommand{\view}{\operatorname{view}}

\newcommand{\Uni}{{\mathord{\mathcal{U}}}}

\newcommand{\oracle}{\calO}

\newcommand{\SD}{\proba{SD}}

\newcommand{\ppt}{{\sc{pptm}}\xspace}

\newcommand{\co}{\calO}

\newcommand{\cQ}{{\cal{Q}}}

\newcommand{\Trans}{\mathrm{Trans}}
\newcommand{\Transs}{\widetilde{\mathrm{Trans}}}

\newcommand{\calH}{\mathcal{H}}

\newcommand{\Vc}{{\mathsf{V}}}

\newcommand{\prob}[2]{\Pr_{#1} \left[ #2 \right]}
\newcommand{\proba}[1]{\mathsf{\textsc{#1}}}

\newcommand{\MyAtop}[2]{\genfrac{}{}{0pt}{}{#1}{#2}}
\newcommand{\lex}{{\mathsf{lex}}}

\newcommand{\Sam}{\ensuremath{\MathAlg{Sam}}\xspace}
\newcommand{\Cf}{\ensuremath{\MathAlg{ColFinder}}\xspace}
\newcommand{\augQ}{augQueries\xspace}
\newcommand{\parent}{{\rm p}}
\newcommand{\extension}{{\mathsf{ext}}}
\newcommand{\Inp}{m}
\newcommand{\Out}{\ell}
\renewcommand{\next}{{\rm next}}
\newcommand{\ext}{{\rm Ext}}
\newcommand{\ans}{{\rm Ans}}
\newcommand{\dec}{{\rm Dec}}

\newcommand{\Jump}{\alpha\mbox{-}{\rm Jump}}
\newcommand{\Gap}{\alpha\beta\mbox{-}{\rm Gap}}
\newcommand{\GapFirst}{{\rm GapFirst}}

\newcommand{\PJump}{\alpha\mbox{-}{\rm PJump}}

\newcommand{\hit}{{\rm hit}}

\newcommand{\Balanced}{{\mathsf{Balanced}}}

\newcommand{\Hit}{{\sf Hit}}

\newcommand{\Decoder}{\ensuremath{\mathsf{Decoder}}\xspace}
\newcommand{\aux}{\ensuremath{\mathsf{aux}}\xspace}
\newcommand{\q}{q}
\newcommand{\Q}{Q}
\newcommand{\h}{h}
\renewcommand{\H}{H}
\newcommand{\s}{s}

\newcommand{\aq}{t}
\newcommand{\qv}{\overline{\q} }
\newcommand{\NoBreak}{\ensuremath{\mathsf{NoBreak}}\xspace}
\newcommand{\TwoDecom}{\ensuremath{\mathsf{TwoOpenings}}\xspace}
\newcommand{\Fail}{\ensuremath{\mathsf{Fail}}\xspace}

\newcommand{\com}{\ensuremath{\mathsf{com}}\xspace}
\newcommand{\Com}{\ensuremath{\mathsf{Com}}\xspace}
\newcommand{\decom}{\ensuremath{\mathsf{decom}}\xspace}

\begin{document}
\sloppy

\title{Finding Collisions in Interactive Protocols --- \\ Tight Lower Bounds
on the Round and Communication Complexities of Statistically Hiding Commitments\thanks{This is the final draft of this paper. The full version was published in the SIAM Journal on Computing 
	\cite{HaitnerHRS15}.  Extended abstracts of this work appeared in the Annual Symposium on Foundations of Computer Science (FOCS) 2007 \cite{HaitnerHRS07} and in  the Theory of Cryptography Conference (TCC) 2013 \cite{HaitnerHS08}.}
\Draft{\\{\small \sc Working Draft: Please Do Not Distribute}}
}
\author{Iftach Haitner\thanks{School of Computer Science, Tel Aviv University. E-mail:
 \texttt{iftachh@cs.tau.ac.il}. Research supported by ISF grant 1076/11, the Israeli Centers of Research Excellence (I-CORE) program (Center  No. 4/11), US-Israel BSF grant 2010196.} \thanks{Part of this research was conducted while at the Weizmann Institute of Science.} \and Jonathan J.\ Hoch\thanks{ImageSat Israel. E-mail:
 \texttt{hoch@imagesatisrael.com}.} \footnotemark[3] \and Omer Reingold\thanks{Stanford University and Weizmann Institute of Science. E-mail: \texttt{omer.reingold@gmail.com}. Research partially supported by the DARPA PROCEED program} \and Gil Segev\thanks{School of Computer Science and Engineering, Hebrew University of Jerusalem, Jerusalem 91904, Israel. Email:  \texttt{segev@cs.huji.ac.il}. Research supported by the European Union's Seventh Framework Programme (FP7) via a Marie Curie Career Integration Grant, by the Israel Science Foundation (Grant No.\ 483/13), and by the Israeli Centers of Research Excellence (I-CORE) Program (Center  No.\ 4/11).} \footnotemark[3]}

\begin{titlepage}
\maketitle

\begin{abstract}
We study the round and communication complexities of various cryptographic protocols. We give tight lower
bounds on the round and communication complexities of any fully black-box reduction of a statistically hiding
commitment scheme from one-way permutations, and from trapdoor permutations. As a corollary, we derive similar tight lower bounds for several
other cryptographic protocols, such as single-server private information retrieval, interactive hashing, and oblivious transfer that guarantees statistical security for one of the parties.

Our techniques extend the collision-finding oracle due to \citeauthor{Simon98} (EUROCRYPT '98) to the setting of
interactive protocols and the reconstruction paradigm of \citeauthor{GennaroT00} (FOCS '00).
\end{abstract}

\noindent\textbf{Keywords:} statistically hiding commitments; private information retrieval; one-way functions; black-box impossibility results.

\thispagestyle{empty}
\pagenumbering{gobble}
\clearpage
\tableofcontents

\thispagestyle{empty}
\clearpage
\pagenumbering{arabic}
\end{titlepage}

\section{Introduction}\label{section:Introduction}

Research in the foundations of cryptography is concerned with the construction of provably secure
cryptographic tools. The security of such constructions relies on a growing number of computational
assumptions, and in the last few decades much research has been devoted to demonstrating the
feasibility of particular cryptographic tasks based on the weakest possible assumptions. For
example, the existence of one-way functions has been shown to be equivalent to the existence of
pseudorandom functions and permutations \cite{GoldreichGM86, LubyR88}, pseudorandom generators
\cite{BlumM84, HastadILL99}, universal one-way hash functions and signature schemes \cite{NaorY89,Rompel90}, different types of commitment schemes \cite{HaitnerNgOnReVa09,HaitnerReVaWe09,HastadILL99, Naor91}, private-key encryption \cite{GoldreichGM84} and other primitives.

Many constructions based on minimal assumptions, however, result in only a theoretical impact due
to their inefficiency, and in practice more efficient constructions based on seemingly stronger
assumptions are being used. Thus, identifying tradeoffs between the \emph{efficiency} of
cryptographic constructions and the strength of the computational assumptions on which they rely is
essential in order to obtain a better understanding of the relationship between cryptographic tasks
and computational assumptions.

In this paper we follow this line of research, and study the tradeoffs between the \emph{round} and \emph{communication} complexities of cryptographic protocols on one hand, and the strength of their underlying computational
assumptions on the other. We provide lower bounds on the round and communication complexities of black-box reduction of
statistically hiding and computationally binding commitment schemes (for short, statistically hiding commitments) from one-way permutations and from families of trapdoor permutations. Our
lower bound matches known upper bounds resulting from \cite{NaorOVY98}. As a corollary of our main
result, we derive similar tight lower bounds for several other cryptographic protocols, such as
{single-server private information retrieval}, {interactive hashing}, and {oblivious
transfer} that guarantees statistical security for one of the parties.

In the following paragraphs we discuss the notion of statistically hiding commitment schemes and describe the setting in which our lower bounds are proved.

\paragraph{Statistically hiding commitments.} A commitment scheme defines a two-stage
interactive protocol between a sender \Sc and a receiver \Rc; informally, after the \emph{commit
stage}, \Sc is bound to (at most) one value, which stays hidden from \Rc, and in the \emph{reveal
stage} \Rc learns this value. The two security properties hinted at in this informal description
are known as \emph{binding} (\Sc is bound to at most one value after the commit stage) and
\emph{hiding} (\Rc does not learn the value to which \Sc commits before the reveal stage). In a
\emph{statistically hiding} commitment scheme, the hiding property holds even against \emph{all-powerful receivers}
(\ie the hiding holds information-theoretically), while the binding property is required to hold
only for polynomially bounded senders.

Statistically hiding commitments can be used as a building block in constructions of statistical
zero-knowledge arguments \cite{BrassardCC88, NaorOVY98} and of certain coin-tossing
protocols~\cite{Lindell03}. When used within protocols in which certain commitments are never
revealed, statistically hiding commitments have the following advantage over computationally hiding
commitment schemes: in such a scenario, it should be infeasible to violate the binding property
\emph{only during the execution of the protocol}, whereas the committed values will remain hidden
\emph{forever} (\ie regardless of how much time the receiver invests after the completion of the
protocol).

Statistically hiding commitments with a constant number of rounds were shown to exist based on
specific number-theoretic assumptions \cite{BoyarKK90, BrassardCC88} (or, more generally, based on
any collection of claw-free permutations \cite{GoldwasserMR88} with an efficiently recognizable
index set \cite{GoldreichK96}), and collision-resistant hash functions \cite{DamgardPP97, NaorY89}.
Protocols with higher round complexity were shown to exist based on different types of one way
functions.  (The communication complexity of the aforementioned protocols varies according to the specific hardness assumption assumed). Protocols with $\Theta( n / \log n)$ rounds and $\Theta(n)$ communication complexity (where $n$ is the input
length of the underlying function) were based on one-way permutations \cite{NaorOVY98} and (known-)
regular one-way functions \cite{HaitnerHKKMS05}.\footnote{The original presentations of the above
protocols have $\Theta(n)$ rounds. By a natural extension, however, the number of rounds in these
protocols can be reduced to $\Theta\left( n / \log n \right)$, see \cite{HaitnerR12,KoshibaS06}.} Finally, protocols with a polynomial number of rounds (and thus, polynomial communication complexity) were based on any one-way
function \cite{HaitnerNgOnReVa09,HaitnerReVaWe09}.\footnote{When provided with a non-uniform advice, the round complexity of \cite{HaitnerReVaWe09} reduces to $\Theta( n / \log n)$.}

\paragraph{Black-box reductions.} As mentioned above, the focus of this paper is proving lower bounds on the round and communication complexity of various cryptographic constructions. In particular,
showing that any construction of statistically hiding commitments based on trapdoor permutations requires a fairly large number of rounds. However, under standard assumptions (\eg the existence of
collision-resistant hash functions), \emph{constant-round statistically hiding commitments do exist}. So if these assumptions hold, the existence of trapdoor permutations implies the existence of
constant-round statistically hiding commitments in a \emph{trivial logical sense}. Faced with similar
difficulties, \citet{ImpagliazzoR89} presented a paradigm for proving
impossibility results under a restricted, yet important, subclass of reductions called
\textit{black-box reductions}. Their method was extended to showing lower bounds on the
\emph{efficiency} of reductions by \citet*{KimST99}.

Intuitively a black-box reduction of a primitive $P$ to a primitive $Q$, is a construction of $P$
out of $Q$ that ignores the internal structure of the implementation of $Q$ and just uses it as a
``subroutine" (\ie as a black-box). In the case of \emph{fully} black-box reductions, the
proof of security (showing that an adversary that breaks the implementation of $P$ implies an
adversary that breaks the implementation of $Q$) is also black-box (\ie the internal structure
of the adversary that breaks the implementation of $P$ is ignored as well). For a more exact
treatment of black-box reductions see \cref{subsection:reductions}.

\subsection{Our Results}\label{subsection:OurResults}
We study the class of fully black-box reductions of statistically hiding commitment schemes from
families of trapdoor permutations, and prove lower bounds on the round and communication complexities of such constructions. Our lower bounds hold also for \emph{enhanced} families of trapdoor permutations: one can efficiently sample a uniformly distributed public key and an element in the permutation's domain, so that inverting the element is hard, even when the random coins used for the above sampling are given as an auxiliary input. Therefore, the bounds stated below imply similar bounds for reduction from one-way permutations\footnote{In general, a black-box impossibility result that is proved w.r.t\ trapdoor permutations does not necessarily hold w.r.t.\ one-way permutation as the additional functionality (of having a trapdoor) may also be used by an attacker. Our result, however, holds even w.r.t.\ {\em enhanced} trapdoor permutations, and these can be used to simulate a one-way permutation by obliviously sampling a public key (the assumed obliviousness of the sampling algorithm enables to transform any attack against such a one-way permutation  into an attack against the underlying enhanced trapdoor permutation).}. Informally, the round complexity lower bound is as follows:
\begin{theorem}[The round complexity lower bound, informal]\label{thm:Intro:RoundComplexityLB}
Any fully black-box reduction of a statistically hiding commitment scheme from a family of
trapdoor permutations over $\zn$ has $\Omega (n / \log n)$ communication
rounds.\footnote{The result holds even if the hiding is only guaranteed to hold against \emph{honest receivers} --- receivers that follow the prescribed protocol.}
\end{theorem}
The above lower bound matches the upper bound due to
\cite{HaitnerR06a, KoshibaS06} (the scheme of \cite{NaorOVY98} has $\Theta(n)$ rounds), who give a fully black-box construction of an $n / (c \cdot \log n)$-round statistically hiding commitment scheme from one-way permutations over $\zn$, for any $c>0$.\footnote{Their proof of security reduction runs in time $\poly(n) \cdot 2^{c \log n}$, and thus efficient only for constant $c$.} In addition, we note that our result and its underlying proof technique, in particular rule out fully black-box reductions of collision-resistant hash functions from one-way function. This provides an alternative and somewhat ``cleaner'' proof than that given by Simon \cite{Simon98} (although our proof applies to fully black-box reductions and Simon's proof applies even to semi black-box ones).

The separation oracle introduced for proving \cref{thm:Intro:RoundComplexityLB}, yields the following lower bound on the communication complexity of statistically hiding commitments:
\begin{theorem}[The communication complexity lower bound, informal]\label{thm:Intro:CommComplexityLB}
In any fully black-box reduction of a statistically hiding commitment scheme from family of
trapdoor permutations over $\zn$, the sender communicates $\Omega(n)$ bits.\footnote{The result holds even if the hiding is only guaranteed to hold against \emph{honest receivers}, and the binding is only guaranteed to hold against \emph{honest senders} --- senders that follow the prescribed protocol in the commit stage.}
\end{theorem}

The above lower bound matches (up to a constant factor) the upper bound due to \cite{NaorOVY98, HaitnerR06a, KoshibaS06}, who give a fully black-box reduction from a statistically hiding commitment scheme from a family of trapdoor permutations over $\zn$, where the sender sends $n-1$ bits. We remark, however, that the above bound says nothing about the number of bits sent by the \emph{receiver}, a number which in the case of \cite{NaorOVY98, HaitnerR06a, KoshibaS06} is $\Theta(n^2)$, and thus dominates the overall communication complexity of the protocol. We also note that the above bound does not grow when the number of committed bits grows, and as such it only matches the bound of \cite{NaorOVY98, HaitnerR06a, KoshibaS06} when the number of bits committed is constant (when committing to $k$ bits, the number of bits sent by the sender in \cite{NaorOVY98, HaitnerR06a, KoshibaS06} is $\Theta(n k)$ but this can be easily reduced to $\Theta( n k/\log n)$).

\subsubsection{Any Hardness Reductions}
We also consider a more general notion of hardness for trapdoor permutations that extends the
standard polynomial hardness requirement; a trapdoor permutation $\tau$
over $\zn$ is $\s$-hard, if any probabilistic algorithm running in time $\s(n)$
inverts $\tau$ on a uniformly chosen image in $\zn$ with probability at most $1 / \s(n)$. We show that any fully black-box reduction of a statistically hiding commitment scheme from a family of $\s$-hard trapdoor permutations requires $\Omega(n / \log \s(n))$ communication rounds. This bound matches the any hardness reduction given in \cite{HaitnerR06a}. Interestingly, the communication complexity lower bound does not change when considering stronger trapdoor permutations.

\subsubsection{Taking the Security of the Reduction into Account}
The informal statements above consider constructions that invoke only trapdoor permutations over $n$ bits.
We would like to extend the result to consider constructions which may invoke the trapdoor
permutations over more than a single domain. In this case, however, better upper bounds are known.
In particular, given security parameter $1^n$ it is possible to apply the scheme of \cite{NaorOVY98}
using a one-way permutation over $n^{\epsilon}$ bits. This implies statistically hiding commitments of $\Theta(n^{\epsilon})$ rounds, where the sender communicates $\Theta(n^{\epsilon})$ bits. This subtle issue is not unique to our setting, and in fact arises in
any study of the efficiency of cryptographic reductions (see, in particular, \cite{GennaroT00,
Wee07}). The common approach for addressing this issue is by restricting the class of constructions
(as in the informal statement of our main theorem above). We
follow a less restrictive approach and consider constructions that are given access to trapdoor
permutations over \emph{any} domain size. Specifically, we consider an additional parameter, which we refer to as
the \emph{security-parameter expansion} of the construction. Informally, the proof of security in a
fully black-box reduction gives a way to translate (in a black-box manner) an adversary $\Ss$
that breaks the binding of the commitment scheme into an adversary $\Ac$ that breaks the security of
the trapdoor permutation. Such a reduction is $\ell (n)$-security-parameter expanding, if
whenever the machine $\Ac$ tries to invert a permutation over $n$ bits, it invokes $\Ss$ on security
parameters which are at most $1^{\ell (n)}$. It should be noted that any reduction in which
$\ell (n)$ is significantly larger than $n$, may only be weakly security preserving (for a taxonomy
of security preserving reductions see \cite[Lecture 2]{Luby96}).

Our lower bound proof takes into consideration the security parameter expansion, and therefore our
statements apply for the most general form of fully black-box reductions. In particular, in case
that $\ell (n) = O(n)$, our theorems imply that the required number of rounds is $\Omega \left(
n / \log n \right)$ and the number of bits send by the sender is $\Omega(n)$. In the general case (where $\ell (n)$ may be any polynomial in $n$), our theorems imply that the required number of rounds and the number of bits send by the sender is $n^{\Omega(1)}$ (which as argued above is tight as well).

\subsubsection{Additional Implications}
 Our main results described above can be extended to any
cryptographic protocol which implies statistically hiding commitment schemes in a fully black-box
manner, as long as the reduction essentially preserves the round complexity or the communication complexity of the underlying protocol.
Specifically, we derive similar lower bounds on the round
complexity and communication complexity of fully black-box reductions from trapdoor permutations of single-server private information retrieval, interactive hashing, and oblivious transfer that guarantees statistical
security for one of the parties. To obtain the above bounds we use known reductions from the listed primitives to statistically hiding commitment schemes. The only exception is the lower bound on the communication complexity on single-server private information retrieval. In this case, the parameters of known reduction (due to \citet{BeimelIKM99}) fail too short to yield the desired lower bound, and we had to come up with a new reduction (given in \cref{section:PIR2Com}).

\subsection{Related Work and Follow-Up Work}

\citet{ImpagliazzoR89} showed that there are no black-box reductions of
key-agrement protocols to one-way permutations and substantial additional work in this line
followed (\cf \cite{GertnerKMRV00,Rudich88,Simon98}). \citet*{KimST99}
initiated a new line of impossibility results, providing a lower bound on the \emph{efficiency}
of black-box reductions (rather than on their feasibility). They proved a lower bound on the
efficiency, in terms of the number of calls to the underlying primitive, of any black-box reduction
of universal one-way hash functions to one-way permutations.  \citet{GennaroT00} has improved \cite{KimST99} to match the known upper bound, and their technique has yielded tight lower
bounds on the efficiency of several other black-box reductions \cite{GennaroGK03, GennaroGKT05,
GennaroT00,HorvitzK05}. 
In all the above results, the measure of efficiency under consideration is the number of calls to the underlying primitives.

With respect to the \emph{round complexity} of statistically hiding commitments, \citet{Fischlin02} showed
that every black-box reduction of statistically hiding commitments to trapdoor permutations, has at least
two rounds. His result follows Simon's oracle separation of collision-resistant hash functions from
one-way permutations \cite{Simon98}. \citet{Wee07} considered a \emph{restricted} class of
black-box reductions of statistically hiding commitments to one-way permutations; informally, \cite{Wee07} considered constructions in which the sender first queries the one-way permutation on several independent
inputs. Once the interaction with the receiver starts, the sender only access the outputs of these
queries (and not the inputs) and does not perform any additional queries. \citet{Wee07} showed that every
black-box reduction of the above class has $\Omega \left( n / \log n \right)$ communication
rounds. From the technical point of view, our techniques are inspired by those of Fischlin \cite{Fischlin02} and Wee \cite{Wee07} by significantly refining and generalizing the approach that an oracle-aided attacker can re-sample its view of the protocol (we refer the reader to \cref{sec:intro:Technique} for more details on our approach).

The question of deriving lower bounds on the round complexity of black-box reductions, was also
addressed in the context of zero-knowledge protocols \cite{CanettiKPR03, DworkNS04, GoldreichK86,
HadaT98, KilianRP05, Rosen04}, to name a few. In this context, however, the black-box access is to the, possibly
cheating, verifier and not to any underlying primitive.

Extensions in the spirit of the one we present here to the \citet{GennaroT00} ``reconstruction lemma'', where used in several works, \eg \cite{HaitnerHol09,Pietrzak08,PietrzakRS12,DodisHT12}. In addition, the separation oracle ``$\Sam$" we present here (see \cref{sec:intro:Technique}), was found to be useful in other separation results \cite{HaitnerMGX10,GordonWXY10,PassV10,RosenS09,RosenS10,BrakerskiKSY11}.

\subsection{Overview of the Technique}\label{sec:intro:Technique}
For the sake of simplicity we concentrate below on the round complexity lower bound of fully black-box
reductions of statistically hiding commitment from one-way permutations (see \cref{section:Intro:LowCom} for the communication complexity lower bound). We also assume that the sender's secret in the commitment protocol is a single uniform bit (\ie it is a bit commitment). Let us start by considering Simon's oracle \cite{Simon98} for ruling out a black-box construction of a
family of collision resistant hash functions from one-way permutations.

\subsubsection{Simon's Oracle \texorpdfstring{$\Cf$}{ColFinder}} Simon's oracle $\Cf$ gets as an input a circuit $C$, possibly with
$\pi$ gates,\footnote{In fact, $\Cf$ also accepts circuits $C$ with $\Cf$ gates. \cite{Simon98} use this extension to give a \emph{single} oracle \wrt which one-way permutations exist, but no collision resistance hash functions. Since the focus of our work is fully black-box reductions, we ignore this extension here and leave it as an open problem to extend our approach to the semi black-box setting.} where $\pi$ is a random permutation. It then outputs two elements $w_1$ and $w_2$ that are uniformly distributed subject to the requirement $C(w_1) = C(w_2)$.\footnote{Consider, for example, sampling $w_1$ uniformly at random from the domain of $C$, and then sampling $w_2$ uniformly at random from the set $C^{-1}(C(w_1))$.} Clearly, in the presence of $\Cf$ no family of collision resistant hash functions exists (the adversary simply queries $\Cf$ with the hash function circuit to find a collision). In order to rule out the existence of any two-round statistically hiding commitment scheme relative to $\Cf$, \citet{Fischlin02} used the following adversary $\Ss$ to break any such scheme: assume \wlg that the first message $q_1$ is sent by \Rc and consider the circuit $C_{q_1}$ defined by $q_1$ and \Sc as follows: $C_{q_1}$ gets as an input the random coins of \Sc and outputs the answer that \Sc replies on receiving the message $q_1$ from \Rc. In the commit stage after receiving the message $q_1$, the cheating $\Ss$ constructs $C_{q_1}$, queries $\Cf(C_{q_1})$ to get $w_1$ and $w_2$, and answers as $\Sc(w_1)$ would (\ie by $C_{q_1}(w_1)$). In the reveal stage, $\Ss$ uses both $w_1$ and $w_2$ to open the commitment (\ie once using the random coins $w_1$ and then using $w_2$). Since the protocol is statistically hiding, the set of the sender's random coins that are consistent with
this commit stage transcript is divided to almost equal size parts by the values of their secret bits. Therefore, with probability roughly half $w_1$ and $w_2$ will differ on the value of \Sc's secret bit and the binding of the commitment will be violated.

In order to obtain the black-box impossibility results (both of \cite{Simon98} and of
\cite{Fischlin02}), it is left to show that $\pi$ is one-way in the presence of $\Cf$. Let \Ac be a
circuit trying to invert $\pi$ on a random $y\in \zn$ using $\Cf$, and lets assume for now that \Ac
makes only a single call to $\Cf$. Intuitively, the way we could hope this query to \Cf with
input $C$ could help is by ``hitting" $y$ in the following sense: we say that \Cf \emph{hits} $y$
on input $C$, if the computations of $C(w_1)$ or of $C(w_2)$ query $\pi$ on $\pi^{-1}(y)$. Now we
note that for every input circuit $C$ each one of $w_1$ and $w_2$ (the outputs of \Cf on $C$) is
\emph{individually} uniform. Therefore, the probability that \Cf hits $y$ on input $C$, may only
be larger by a factor two than the probability that evaluating $C$ on a uniform $w$ queries $\pi$
on $\pi^{-1}(y)$. In other words, \Ac does not gain much by querying \Cf (as \Ac can evaluate $C$
on a uniform $w$ on its own). Formalizing the above intuition is far from easy, mainly when we
consider \Ac that queries \Cf more than once. The difficulty lies in formalizing the claim that
the only useful queries are the ones in which \Cf hits $y$ (after all, the reply to a query may
give us some useful global information on $\pi$).\footnote{The  proof of our main theorem (see intuition in  \cref{sec:introPermutationAreHard}), implies an alternative proof for the above claim.}

\subsubsection{Finding Collisions in Interactive Protocols}
We would like to employ Simon's oracle for breaking the binding of more interactive protocols (with more than two rounds). Unfortunately,
the ``natural" attempts to do so seem to fail miserably. The first attempt that comes to mind might
be the following: in the commit stage, $\Ss$ follows the protocol and let $q_1,\dots,q_k$ be the
messages that \Rc sent in this stage. In the reveal stage, $\Ss$ queries \Cf to get a colliding
pair $(w_1,w_2)$ in $C_{q_1,\dots,q_k}$ --- the circuit naturally defined by the code of \Sc and
$q_1,\dots,q_k$ (\ie $C_{q_1,\dots,q_k}$ gets as an input the random coins of \Sc and outputs
the messages sent by \Sc when \Rc's messages are $q_1,\dots,q_k$). The problem is that it is very
unlikely that the outputs of \Sam on $C_{q_1,\dots,q_k}$ will be consistent with the answers that
$\Ss$ \emph{already} gave in the commit stage (we did not encounter this problem when breaking
two-round protocols, since $\Ss$ could query \Cf on $C_{q_1}$ before $\Ss$ sends its first and
only message). Alternatively, we could have changed \Cf such that it gets as an additional input $w_1$
and returns $w_2$ for which $C_{q_1,\dots,q_k}(w_1) = C_{q_1,\dots,q_k}(w_2)$ (that is, the new
\Cf finds second preimages rather than collisions). Indeed, this new \Cf does imply the
breaking of any commitment scheme, but it also implies the inversion of $\pi$.\footnote{Consider a
circuit $C$, whose input is composed of a bit $\sigma$ and an $n$-bit string $w$. The circuit $C$
is defined by $C(0,w)=\pi(w)$ and $C(1,w)=w$. Thus, in order to compute $\pi^{-1}(y)$ we can simply
invoke the new \Cf on input $C$ and $w_1=(1,y)$. With probability half \Cf will return
$w_2=(0,\pi^{-1}(y))$.} We should not be too surprised that both the above attempts failed as they
are both completely oblivious of the round complexity of $(\Sc,\Rc)$. Since one-way permutations \emph{do imply} statistically hiding commitments (in a black-box manner) \cite{NaorOVY98,HaitnerHKKMS05,HaitnerNgOnReVa09,HaitnerReVaWe09}, any oracle that breaks statistically hiding
commitments could also be used to break the underlying one-way permutations.\footnote{In addition,
in both these naive attempts the cheating sender $\Ss$ follows the commit stage honestly (as \Sc
would). It is not hard to come up with two-round protocol that works well for semi-honest commit
stage senders (consider for instance the two-message variant of \cite{NaorOVY98} where the receiver's queries are all sent in the first round).}

So the goal is to extend Simon's oracle to handle interactions, while not making it ``too strong" (so that it does not break the one-way permutations). In fact, the more interactive our oracle will be, the more powerful it will be (eventually, it will allow breaking the one-way permutations). Quantifying this growth in power is how we get the tight bounds
on the round complexity of the reduction.

\subsubsection{The Oracle \Sam}
It will be useful for us to view Simon's oracle as performing two sampling
tasks: first, it samples $w_1$ uniformly, and then it samples a second preimage $w_2$ with
$C(w_1) = C(w_2)$. As explained above, an oracle for sampling a second preimage allows inverting
the one-way permutations. What saves us in the case of  \Cf, is that $w_1$ was chosen by \Cf \emph{after} $C$ is already given. Therefore, an adversary \Ac is very limited
in setting up the second distribution from which \Cf samples (\ie the uniform distribution over
the preimages of $C(w_1)$ under $C$). In other words, this distribution is \emph{jointly} defined by \Ac and \Cf itself.

Extending the above interpretation of \Cf, our separation oracle \Sam is defined as follows: \Sam is given as input a query $\q= (w,C, C_\next)$ and outputs a preimage $w'$, where $w'$ is a uniformly distributed preimage of $C(w)$ (the purpose of the circuit $C_\next$ will be revealed later). In case $C= \perp$, algorithm \Sam outputs a uniform element in the domain.

While the above \Sam \emph{can} be used for inverting random permutations when used by an \emph{arbitrary} algorithm, it is not the case when used by low-depth \emph{normal form} algorithms; an algorithm \Ac is in a normal-form, if it makes the query $\q= (w,C \neq \perp, C_\next)$ to \Sam only if it has \emph{previously} made the query $\q'= (\cdot ,\cdot, C)$ to \Sam, and got $w$ as the answer (namely, the third input to $\Sam$ is used for ``committing'' to $C$ before seeing $w$).\footnote{An additional important restriction, that we will not discuss here, is that $C_\next$ is an \emph{extension} of the circuit $C$, where extension means that $C_\next(w) = (C(w), \widetilde{C}(w))$ for some circuit $\widetilde{C}$ and for every $w$.} A normal-form algorithm is of \emph{depth} $d$ if $d$ is the length of the longest chain of \Sam queries it makes (\ie $\Sam(\cdot ,\cdot, C_2)= w_2, \Sam(w_2,C_2, C_3)= w_3,\dots, \Sam(w_d,C_d, \cdot)= w_{d+1}$). While restricted, it turns out that normal-form algorithms of depth $(d+1)$ are strong enough, with the aid of \Sam, for breaking $d$-round statically hiding commitments.\footnote{In the preliminary versions \cite{HaitnerHRS07,HaitnerHS08}, we equipped \Sam with a signature-based mechanism to enforces normal-form behaviour and depth restriction on the queries it is asked upon. While yielding a simpler (and easier to comprehend) characterization of the power of \Sam (\ie useful for breaking commitments, not useful for inverting random permutations), the signature-based  mechanism had  significantly complicated the whole text.}

Assume there exists a fully black-box reduction from an $o(n / \log n)$-round statically hiding commitments to one-way permutations. By the above observation, the reduction should invert a random permutation when given oracle to an $o(n / \log n)$-depth normal-form algorithm \Ss with oracle access to \Sam. Since the reduction has no direct access to \Sam (but only via accessing \Ss), it is easy to see that the reduction itself is an $o(n / \log n)$-depth normal-form algorithm. This implies a contradiction, since low-depth normal form algorithms cannot invert random permutations.

\subsubsection{A \texorpdfstring{$(d+1)$}{d+1}-depth Normal-Form Algorithm that Breaks \texorpdfstring{$d$}{d}-Round Commitments}\label{sec:introPowerofSam}
Given a $d$-round statistically hiding commitment, the $(d+1)$-depth normal-forma algorithm $\Ss$ for breaking the commitments operates as follows: after getting the first message $q_1$, it constructs $C_{q_1}$ (the circuit that computes \Sc's first message) and queries \Sam on $(\perp,\perp,C_{q_1})$ to get input $w_1$, and sends $C_{q_1}(w_1)$ back to \Rc. On getting the $i$'th receiver message $q_i$, the adversary $\Ss$ constructs $C_{q_1,\dots,q_i}$ (the circuit that computes \Sc's first $i$ messages), queries \Sam on $(w_{i-1},C_{q_1,\dots,q_{i-1}},C_{q_1,\dots,q_{i}})$ to get $w_i$, and sends the $i$'th message of $C_{q_1,\dots,q_i}(w_i)$ back
to \Rc. Finally, after completing the
commit stage (when answering the last receiver message $q_d$) it queries \Sam on
$(w_d,C_{q_1,\dots,q_{d}},\perp)$ to get $w_{d+1}$. Since both $w_d$ and $w_{d+1}$ are sender's
random inputs that are consistent with the commit-stage transcript, with probability
roughly half they can be used for breaking the binding of the protocol.

\subsubsection{Random Permutations Are Hard For \texorpdfstring{$o(n/\log n)$}{o(n/logn)}-Depth Normal-Form Algorithms}\label{sec:introPermutationAreHard}
To complete our impossibility result, it is left to prove that \Sam cannot be used by $d(n)\in o(n/\log n)$-depth normal-form algorithms to invert a random permutation $\pi$. Let \Ac be such a $o(n/\log n)$-depth normal-form algorithms. A \Sam query $(\cdot,C,\cdot)$ is \emph{$y$-hitting} (\wrt $\pi$), if it is answered with $w$, such that $C(w')$ queries $\pi$ on $\pi^{-1}(y)$. Where $\Ac$ \emph{hits} on input $y$, if it makes a $y$-hitting query. Given the above definition, our proof is two folded. We first show that a normal-form algorithm that hits on a random $y$ with high probability, implies an algorithm that, with significant probability, inverts $\pi$ \emph{without hitting} (the proof of this part, influenced by the work of \citet{Wee07}, is the most technical part of the paper). We then extend the reconstruction technique of \citet{GennaroT00}, to show that a non-hitting algorithm is unlikely to invert $\pi$.

\paragraph{From normal-form hitting algorithms to non-hitting inverters.}
Let \Ac be an algorithm that hits on a random $y$ with high probability. The idea is that if $\Ac(y)$ hits, then it ``knew" how to invert $y$ \emph{before} making the hitting \Sam call. Assume for simplicity of notation that $\Ac(y)$'s queries are of the form $q_1 = (\bot, \bot,C_2), q_2= (w_2,C_2, C_3), \cdots, q_d = (w_d,C_d, \cdot)$, where $w_{i+1}$ is \Sam answer on the query $q_i$ (this essentially follows from \Ac being in a $d$-depth normal form). And let $i^\ast$ be such that $\q_{i^\ast}$ hits $y$ with high probability (this follows from the assumption about \Ac being a good hitter). Since $d \in o( n / \log n)$, there exists a location $i$ such that the probability $\q_i$ to hit $y$ is larger than the probability that $q_{i-1}$ hits $y$ by a arbitrary large polynomial. Further, an average argument yields that the probability that $C_i(w_i)$ queries $\pi$ on $\pi^{-1}(y)$, is unlikely to be much smaller than the probability that $\q_i$ hits $y$ (which is the probability that $C_i(w_{i+1})$ queries $\pi$ on $\pi^{-1}(y)$).

Combining the above understandings, we design $\Mc$ that with
non-negligible probability, inverts $\pi$ on $y$ without hitting. Algorithm $\Mc$ emulates \Ac while following each \Sam query $(w_i,C_i,C_{i+1})$ made by \Ac receiving a reply $w_{i+1}$, it evaluates, in addition, $C_{i+1}(w_{i+1})$. If $C_{i+1}$ queries $\pi$ on $x= \pi^{-1}(y)$, then $\Mc$ halts and outputs $x$ (otherwise, it continues with the emulation of \Ac). We argue that with sufficiently
large probability, if the first hinting query of \Ac is $q_i = (w_i,C_i,C_{i+1})$, then
$\Mc$'s computation of $C_i(w_i)$ queries $\pi$ on $\pi^{-1}(y)$. Therefore, $\Mc$ retrieves
$\pi^{-1}(y)$ \emph{before} making the hitting query.

\paragraph{Random permutations are hard for non-hitting inverters.}
\citet{GennaroT00} presented a very elegant argument for proving that random
permutations are hard to invert also for non-uniform adversaries (previous proofs, \eg
\cite{ImpagliazzoR89}, only ruled out uniform adversaries). Let \Ac be a circuit and let $\pi$ be a
permutation that \Ac inverts on a non-negligible fraction of its outputs. \cite{GennaroT00}
showed  that $\pi$ has a  ``short"  description relative to \Ac. (Intuitively, \Ac saves on the description of $\pi$ as it allows us to reconstruct $\pi$ on (many of) the $x$'s for which $\Ac^{\pi}(\pi(x))=x$). Therefore, by a counting argument, there is only a tiny fraction of permutations which \Ac inverts well.

The formal proof strongly relies on a bound on the
number of $\pi$ gates in \Ac: when we use \Ac to reconstruct $\pi$ on $x$ we need all the
$\pi$-queries made by $\Ac^{\pi}(\pi(x))$ (apart perhaps of the query for $\pi(x)$ itself) to already
be reconstructed.

Consider an adversary $\Ac$ that, with significant probability, inverts $\pi$ without hitting. Recall that when queried on $(w,C,\cdot)$, the oracle \Sam returns a random inverse of $C(w)$. We would like to apply the argument of \cite{GennaroT00} to claim that relative to \Ac and \Sam there is a short description of $\pi$. We are faced with a
substantial obstacle, however, as  \Sam might make a huge amount of $\pi$
queries.\footnote{Consider for example $C$ such that on input $w$ it truncates the last bit of
$\pi(w)$ and outputs the result. Finding collisions in $C$ requires knowledge of $\pi$ almost
entirely.} On the intuitive level, we overcome this  obstacle by exploiting the fact that while  \Sam does not have an efficient \emph{deterministic} implementation, it does have an efficient \emph{non-deterministic} one: simply guess where the collision occur, and verify that this is indeed the case. Formalizing the above approach requires much care both in the definition and analysis of
\Sam, and critically use the assumption that $\Ac$ is non hitting. We defer more details to \cref{subsection:hitting}.

\subsubsection{Low Sender-Communication Commitments}\label{section:Intro:LowCom}
The lower bound for low sender-communication statistically hiding commitment follows from the fact that the oracle \Sam, described above, can be used for breaking the binding of such commitments, and moreover, this can be done by low-depth normal-form algorithms. The idea is fairly straightforward; given an $o(n)$-communication commitment, the $o(n/\log n)$-depth normal-form algorithm \Ss for breaking the commitment, acts \emph{honestly} in the commit stage (say by committing to zero), and only then uses \Sam for finding decommitments to both zero and one.

Specifically, after the commit phase is over $\Ss$ partitions $\trans$ into $d \in o(n/\log n)$ blocks $\trans_1,\dots,\trans_d$, where $\trans_i$ contains the $(i-1)\cdot \log(n) +1,\cdots,i\cdot \log(n)$ bits sent by \Sc (for simplicity we assume here that in each round the sender communicates a single bit to the receiver). Then \Ss iteratively applies \Sam, such that after the $i$'th iteration, \Ss obtains random coins $w_i$ that are consistent with $\trans_{1,\dots,i}$. If successful, \Ss makes an additional call $(w_d,C_{q_1,\dots,q_{d}},\perp)$, where $C_{q_1,\dots,q_{d}}$ is as in \cref{sec:introPowerofSam}, to obtain additional coins $w_d'$ consistent with $\trans$, and then uses $w_d$ and $w_d'$ to break the commitment.

It is left to describe how \Ss obtains a consistent $w_{i+1}$, given that it has previously obtained a consistent $w_i$. For that, \Ss keeps calling \Sam on $(w_d,C_{q_1,\dots,q_{i}},C_{q_1,\dots,q_{i+1}})$ until it is replied with $w_{i+1}$ that is consistent with $\trans_1,\dots,\trans_{i+1}$. Since $\trans_{i+1}$ contains only $\log n$ bits sent by \Sc, we expect \Ss to succeeds with high probability after about $n$ such attempts.

\subsection{Paper Organization}
Notations and formal definitions are given in \cref{section:Preliminaries}, where the oracle \Sam and the separation oracle discussed above is formally defined in \cref{section:SeperationOracle}. In \cref{section:PowerOfSam} we show how to use \Sam (by normal form algorithms) to find  collisions in any low-round complexity or low sender communication protocols, wherein \cref{section:inverting} we show that in the hands of normal-form algorithms, \Sam is not useful for inverting random permutations. In \cref{section:LowerBounds} we combine the above fact to derive our lower bounds on statistically hiding commitment schemes, where applications of the above results to
other cryptographic protocols,  are given in \cref{section:implications}. Finally, in \cref{section:PIR2Com} we give a refined reduction of low-communication statistically hiding commitment schemes from low-communication single-server private information retrieval, that implies a lower bound of low-communication private information retrieval schemes.

\section{Preliminaries}\label{section:Preliminaries}
\subsection{Conventions and Basic Notations}
All logarithms are in base two.  We use calligraphic letters to denote sets, uppercase for random variables, and lowercase for values. Let $\poly$ be the set of all polynomials $p : \mathbb{N} \rightarrow \mathbb{N}$. A function $\mu \colon \N \rightarrow [0,1]$ is \textit{negligible}, denoted $\mu(n) = \negl(n)$, if $\mu(n) < 1/p(n)$ for all $p\in \poly$ and large enough $n$. For $n\in \N$, let $[n]=\set{1,\ldots,n}$. For a finite
set $\cX$, denote by $x \la \cX$ the experiment of choosing an element of
$\cX$ according to the uniform distribution, and by $U_n$ the uniform distribution over the set $\zn$. Similarly, for a distribution $D$ over a set $\Uni$, denote by $u \la D$ the experiment of choosing an element of $\Uni$ according to the distribution
$D$. The statistical distance between two distributions $P$ and $Q$
over a set $\Uni$, denoted $\SD(P, Q)$, is defined as $\frac{1}{2} \sum_{u \in \Uni} \size{\prob{P}{u} - \prob{Q}{u} }$. Given an event ${\sf E}$, we denote by $\SD(X, Y\mid {\sf E})$ the statistical distance between the conditional distributions $X|_{\sf E}$ and $Y|_{\sf E}$.

\subsection{Algorithms and Circuits}\label{sec:prelim:Alg}
Let \ppt stand for probabilistic algorithm (\ie Turing machines) that runs in \emph{strict} polynomial time. The input and output length of a circuit $\Cc$, denoted $\Inp(\Cc)$ and $\Out(\Cc)$, are the number of input wires and output wires in $\Cc$ respectively. Given a circuit family $\Cc = \set{\Cc_n}_{n\in \N}$ and input $x\in \zn$, let $\Cc(x)$ stands for $\Cc_n(x)$.

An oracle-aided algorithm $\Ac$ is an interactive Turing machines equipped with an additional tape called the \emph{oracle tape}; the Turing machine can make a query to the oracle by writing a string $q$ on its tape. It then receives a string $ans$ (denoting the answer for this query) on the oracle tape. Giving a deterministic function $\oracle$, we denote by $\Ac^\oracle$ the algorithm defined by $\Ac$ with oracle access $\oracle$. The definition naturally extends to circuits, where in this case the circuit is equipped with \emph{oracle gates}. In all the above cases, we allow the access to several different oracles, where this is nothing but a syntactic sugar to denote a \emph{single} oracle that answers the queries it is asked upon by the relevant oracle, according to some syntax imposed on the queries (\ie each query starts with a string telling the oracle it refers to).\footnote{The above only consider ``oracles" that implement   \emph{deterministic} functions. We will not consider random or state-full oracles.}

When dealing with an execution of an oracle-aided Turing-machines, we identify the queries according to their chronological order, where when dealing with circuits, we assume an arbitrary order that respects the topological structure of the circuit (\ie a query asked in a gate of depth $i$, appears before any of the queries asked in gates of depth larger than $i$).

A $q$-query oracle-aided algorithm asks at most $q(n)$ oracle queries on input of length $n$, where in a $q$-query oracle-aided circuit family $\set{\Ac_n}_{n\in \N}$, the circuit $\Cc_n$ has at most $q(n)$ oracle gates.

An oracle-aided function mapping $n$-bit strings to $\ell(n)$-bit strings, stands for a \emph{deterministic} oracle-aided algorithm that given access to any oracle and $n$-bit input, outputs $\ell(n)$-bit string.\footnote{Since we only consider deterministic stateless oracles, for any fixing of the oracle, such algorithm indeed computes a function from $n$ bits to $\ell(n)$ bits.}

\subsection{Interactive Protocols}\label{section:prelimInteractiveP}
A two-party protocol $\pi=\vect{\Ac,\Bc}$ is a pair of \ppt's. The communication between the Turing machines $\Ac$ and $\Bc$ is carried out in rounds. Each round consists of a message sent from $\Ac$ to $\Bc$ followed by a message sent from
the $\Bc$ to $\Ac$. We call $\pi$ an $m$-round protocol, if for \emph{every} possible random coins for the parties, the number of rounds is at \emph{exactly} $m$. A communication transcript $\trans$ (\ie the ``transcript") is the list of messages exchanged between the parties in an execution of the protocol, where $\trans_{1,\ldots,j}$ denotes the first $j$ messages in $\trans$. A view of a party contains its input, its random tape and the messages exchanged by the parties during the execution. Specifically, $\Ac$'s view is a tuple $v_\Ac = (i_\Ac,r_\Ac,\trans)$, where $i_\Ac$ is $\Ac$'s input, $r_\Ac$ are $\Ac$'s coins, and $\trans$ is the transcript of the execution. Let the random variable $\exec{(\Ac(i_\Ac),\Bc(i_\Bc)(i))}$ denote the common transcript, the parties' local outputs and the parties's views in a random execution of $(\Ac(i_\Ac),\Bc(i_\Bc)(i))$ (\ie the private inputs of $\Ac$ and $\Bc$ are $i_\Ac$ and $i_\Bc$ respectively, and $i$ is the common input). We naturally refer to the different parts of $\exec{\cdot}$ with $\exec{\cdot}_\trans$, $\exec{\cdot}_{\out^\Ac}$, $\exec{\cdot}_{\out^\Bc}$, $\exec{\cdot}_{\view^\Ac}$ and $\exec{\cdot}_{\view^\Bc}$, respectively.

 The above notation naturally extends to oracle-aided protocols, where the main distinction is that the view of an oracle-aided party also contains the answers it got from the oracle.

\subsection{Random Permutations and One-Way Permutations}
For $n\in \N$, let $\Pi_n$ be the set of all permutations over $\zn$, and let $\Pi$ be the set all infinite collections $\pi = (\pi_1,\pi_2,\dots)$ with $\pi_i \in \Pi_i$ for every $n\in \N$ (note that the set $\Pi$ is not countable, and as a result our probability analysis in this paper deals with non-countable probability spaces). Our lower bound proof is based on analyzing random instances of such permutation collections.
\begin{definition}[random permutations]
A random choice of $\pi=\set{\pi_n}_{n\in \N}$ from $\Pi$, denoted $\pi \la \Pi$, means that $\pi_n$, for every $i\in \N$, is chosen uniformly at random and independently from $\Pi_n$.
\end{definition}

A collection of permutations is hard (i.e., one-way), if no algorithm can invert it with high probability.
\begin{definition}[one-way permutations]
A collection of permutations $\pi \in \Pi$ is {\sf $\s(n)$-hard}, if for
every oracle-aided algorithm \Ac of running time $\s(n)$ and all sufficiently large $n$, it holds that
$$\prob{y\la \zn}{\Ac^{\pi}(1^n, y) = \pi^{-1}_n(y)} \leq \frac{1}{\s(n)},$$
where the probability is taken also over the random coins of \Ac. The permutation $\pi$ is {\sf polynomially-hard}, if it is $\s(n)$-hard for some $\s(n) = n^{\omega(1)}$.
\end{definition}
It is well known (\cf \cite{GennaroGKT05,ImpagliazzoR89}) that random permutations (and also random trapdoor permutations, see below) are hard to invert when given oracle access to the permutation.
\begin{theorem}[\cite{GennaroGKT05}]\label{thm:RandomPermAreHard}
For large enough $n\in \N$ and any $2^{n/5}$-query circuit $\Cc$, it holds that
$$\prob{\pi_n \la \Pi_n}{\prob{y\la \zn}{\Cc^{\pi_n}(y) = \pi_n^{-1}(y)} > 2^{-n/5}} < 2^{-2^{n/2}}.$$
\end{theorem}

In this paper we make a step further, showing that random permutations (and trapdoor random permutations) are to invert even in the presence of the \emph{exponential-time} oracle \Sam.

\subsection{Random Trapdoor Permutations and One-Way Trapdoor Permutations}
A collection of trapdoor permutations is represented as a triplet
$\tau = \left( G, F, F^{-1} \right)$. Informally, $G$ corresponds
to a key generation procedure, which is queried on a string $td$
(intended as the ``trapdoor'') and produces a corresponding public
key $pk$. The procedure $F$ is the actual permutation, which is
queried on a public key $pk$ and an input $x$. Finally, the
procedure $F^{-1}$ is the inverse of $F$ ---  $G(td) = pk$ and
$F(pk, x) = y$, implies $F^{-1}(td, y) = x$.

\begin{definition}[trapdoor permutations]\label{def:TDP}
Let $\Tau$ the set of all function triplets $\tau =
( G, F, F^{-1})$ with
\begin{enumerate}
\item $G \in \Pi$.

\item For every $n\in \N$ and $pk \in
\zn$, the function $F_{pk}$ over $\zn$ defined as $F_{pk}(x) = F(pk, x)$, is in $\Pi_n$.

\item For every $n\in \N$ and $sk,x \in \zn$, it holds that $F^{-1} (sk, F(G(sk), x)) = x$.
\end{enumerate}
A tuple $\tau \in \Tau$ is called  a {\sf family of trapdoor permutations}.
\end{definition}
As in the case of standard permutations, we consider random instances of such trapdoor permutations collections.
\begin{definition}[random trapdoor permutations]\label{def:RandTDP}
A random choice $\tau =(G,F,F^{-1})$ from $\Tau$, denoted $\tau \la \Tau$, means that $G \la \Pi$, and every $n\in \N$ and $pk\in \zn$, the permutation $F_{pk}$, defined in \cref{def:TDP}, is chosen uniformly at random and independently from $\Pi_n$.
\end{definition}

A collection of trapdoor permutations is hard, if no algorithm, equipped with only the public key, can invert it with high probability.
\begin{definition}\label{def:HardnessOfTDP}
A family of trapdoor permutations $\tau = (G, F, F^{-1}) \in \Tau$ is {\sf $\s(n)$-hard}, if
\begin{align*}
\prob{td \la \zn;y\la \zn}{\Ac^{\tau}(1^n, G(td), y) = F^{-1} (td, y)} \leq
\frac{1}{\s(n)},
\end{align*}
for every oracle-aided algorithm \Ac of running time $\s(n)$ and all sufficiently large $n$, where the probability is also taken over the random coins of \Ac. The family $\tau$ is {\sf polynomially hard}, if it is $\s(n)$-hard for some $\s(n) = n^{\omega(1)}$.
\end{definition}

Since we are concerned with providing a lower bound, we do not consider the most general definition of trapdoor permutations. (see \cite{Goldreich01} for such a definition). In addition, \cref{def:TDP} refers to the difficulty of inverting the permutation
$F_{pk}$ on a uniformly distributed image $y$, when given only $pk = G(td)$ and $y$. Some
applications, however, require \emph{enhanced} hardness conditions. For example, it may be required (\cf \cite[Appendix \Cc]{Goldreich04}) that it is hard to invert $F_{pk}$ on $y$ even given the
random coins used in the generation of $y$. Our formulation captures such hardness
condition, and therefore the impossibility results proved in this paper hold also for
enhanced trapdoor permutations.\footnote{A different enhancement, used by \cite{Haitner04},
requires the permutations' domain to be polynomially dense in $\zn$. Clearly, our formulation is
polynomially dense.} Finally, since the generator $\G$ of an $\s$-hard trapdoor permutations family, of the above type, is an $s$-hard \emph{one-way permutation} (\ie no algorithm of running
time $s(n)$ inverts with probability better than $1/s(n)$), the lower bounds we state here \wrt families of trapdoor permutations, yield analog bounds for one-way permutation.

\subsection{Commitment Schemes}\label{sec:SHC}
A commitment scheme is a two-stage interactive protocol between a sender and a receiver.
Informally, after the first stage of the protocol, which is referred to as the \emph{commit stage},
the sender is bound to at most one value, not yet revealed to the receiver. In the second stage,
which is referred to as the \emph{reveal stage}, the sender reveals its committed value to the
receiver. In this paper, where we are interested in proving an impossibility result for commitment
schemes, it will be sufficient for us to deal with bit-commitment schemes, \ie commitment schemes
in which the committed value is only one bit.

\begin{definition}[bit-commitment scheme]\label{def:commitment}
A {\sf bit-commitment scheme} is a triplet of \ppt's $(\Sc,\Rc, \Vc)$ such that
\begin{align*}
\Pr_{(\decom, \com) \la \exec{(\Sc(b), \Rc)(1^n)}_{\out^{\Sc},\out^{\Rc}}}[\Vc (\com, \decom) = b] =1,
\end{align*}
for both $b\in \zo$ and all $n\in \N$.\footnote{Note that there is no loss of generality in
assuming that the decommitment stage in non interactive. This is since
any such interactive algorithm can be replaced with a non-interactive
one as follows: let $\decom$ is the internal state of $\Sc$ when the decommitment starts, and let $\Vc(\decom)$ simulate the sender and the interactive verifier in the interactive decommitment stage.}

\end{definition}

The security of a commitment scheme can be defined in two complementary ways, protecting against
either an all-powerful sender or an all-powerful receiver. In this paper, we deal with commitment
schemes of the latter type, which are referred to as \emph{statistically hiding} commitments.
\begin{definition}[statistical hiding]\label{def:hiding}
Let $\Com = (\Sc, \cdot, \cdot)$ be a bit-commitment scheme. For algorithm $\Rs$, bit $b$ and integer $n$, let $\Trans^\Rs(b,n) = \exec{(\Sc(b), \Rc)(1^n)}_{\trans}$. The scheme $\Com$ is {\sf $\rho(n)$-hiding}, if $\SD\left(\Trans^\Rs(0,n),\Trans^\Rs(1,n)\right)\leq \rho(n)$ for \emph{any} algorithm $\Rs$ and large enough $n$.\footnote{It is more common to require that $\Rs$'s views (and not the transcripts) are statistically close. Since. however, we put no restriction on the computation power of \Rs, the two definitions are equivalent.} $\Com$ is {\sf statistically hiding}, if it is $\rho(n)$-hiding for some negligible function $\rho(n)$.
When limiting the above to $\Rs = \Rc$, then $\Com$ is called {\sf honest-receiver $\rho(n)$-hiding/statistically hiding}.
\end{definition}

\begin{definition}[computational binding]\label{def:binding}
A bit-commitment scheme $\Com = (\cdot, \Rc, \Vc)$ is {\sf $\mu(n)$-binding}, if
\[ \pr{{((decom, decom'),\com) \la \exec{(\Ss, \Rc)(1^n)}_{\out^\Ss,\out^\Rc}}\colon \MyAtop{\Vc (\com, \decom) = 0,}{\Vc (\com, {\sf decom'}) = 1 }} < \mu(n) \enspace \]%
for any \ppt $\Ss$ and sufficiently large $n$. $\Com$ is {\sf computationally binding}, if it is $\mu(n)$-binding [\resp {\sf honest-sender $\mu(n)$-binding}] for some negligible function $\mu(n)$, and is {\sf weakly binding}, if it is $(1 - 1/p(n))$-binding for some polynomial $p(n)$.

When limiting the above $\Ss$ that acts {\sf honestly} in the commit stage,\footnote{I.e., in the commitment stage $\Ss$ acts as $\Ss(b;r)$, for some $b\in \zo$ and $r$ that is uniformly chosen from the possible coins for $\Sc$.} then $\Com$ is called {\sf honest-sender $\mu(n)$-binding/computationally binding/weakly binding}.
\end{definition}

\subsection{Black-Box Reductions}\label{subsection:reductions}
A reduction of a primitive $P$ to a primitive $Q$ is a construction of $P$ out of $Q$. Such a
construction consists of showing that if there exists an implementation $\Cc$ of $Q$, then there
exists an implementation $\Mc_\Cc$ of $P$. This is equivalent to showing that for every adversary that
breaks $\Mc_\Cc$, there exists an adversary that breaks $\Cc$. Such a reduction is \emph{semi black box}, if it ignores the internal structure of $Q$'s implementation, and it is \emph{fully black box}, if
the proof of correctness is black-box as well (\ie the adversary for breaking $Q$ ignores the
internal structure of both $Q$'s implementation and of the [alleged] adversary breaking $P$).
Semi-black-box reductions are less restricted and thus more powerful than fully black-box
reductions. A taxonomy of black-box reductions was provided by \citet{ReingoldTV04}, and the reader is referred to their paper for a more complete and formal view
of these notions.

We now formally define the class of constructions considered in this paper. Our main result is
concerned with the particular setting of fully black-box constructions of weakly binding
statistically hiding commitment schemes from trapdoor permutations. We focus here on a specific
definition for these particular primitives and we refer the reader to \cite{ReingoldTV04} for a
more general definition.

\begin{definition}\label{definition:fully}
A fully black-box construction of weakly binding, statistically hiding commitment
scheme from $\s(n)$-hard family of trapdoor permutations, is a quadruple of oracle-aided
\ppt's $(\Sc, \Rc, \Vc, \Ac)$ such that the following hold:
\begin{description}
\item[Correctness and hiding] The scheme $\Com^\tau = (\Sc^{\tau}, \Rc^{\tau}, \Vc^{\tau})$ is a correct, honest-receiver statistically hiding commitment scheme for every $\tau \in \Tau$.

\item[Black-box proof of binding:] For every $\tau = (G,F,F^{-1}) \in \Tau$ and every algorithm $\Ss$ such that $\Ss$ breaks the weakly binding of
$(\Sc^{\tau}, \Rc^{\tau}, \Vc^{\tau})$, according to \cref{def:binding}, it holds that
\begin{align*}
\prob{td \la \zn;y\la \zn}{\Ac^{\tau}(1^n, G (td), y) = F^{-1} (td, y)} >\frac{1}{\s(n)}
\end{align*}
for infinitely many $n$'s.\footnote{A natural relaxation of \cref{definition:fully} is to consider the running time of the ``security proof" \Ac as an additional parameter. Allowing it, for instance, to run at exponential time when the trapdoor permutation of interest are ``exponentially hard" (\ie $\s(n) = 2^{cn}$). For the sake of presentation clarity, however, we chose no to consider such generalization.}
\end{description}

The construction is of honest-sender commitment, if the above only considers honest senders.
\end{definition}

It would be useful for us to consider the following property of fully black-box reduction: consider a malicious sender $\Ss$ that breaks the binding of the commitment scheme and consider the machine \Ac that wishes to break the security of the trapdoor permutation. Then, \Ac receives a security parameter $1^n$ and invokes $\Ss$ in a black-box manner. \cref{definition:fully}, however, does not restrict the range of security parameters that \Ac is allowed to invoke $\Ss$ on. For example, \Ac may invoke $\Ss$ on security parameter $1^{n^2}$, or even on security parameter $1^{\Theta(\s(n))}$, where $\s(n)$ is the running time of \Ac. The following definition will enable us to capture this property of the construction, and again, we present a specific definition for our setting.
\begin{definition}\label{definition:ParameterExpanding}
A black-box construction $(\Sc, \Rc,\Vc, \Ac)$ according to \cref{definition:fully} is {\sf $\ell$-security-parameter expanding},
if for every malicious sender $\Ss$, the machine \Ac on security parameter $1^n$ invokes $\Ss$ on security
parameter at most $1^{\ell(n)}$.
\end{definition}

\Ignore{%
A \emph{relativizing reduction} of $P$ to $Q$ is a reduction that
holds relative to any oracle $\mathcal{O}$, \ie a proof showing
that for every oracle $\mathcal{O}$, if there exists an
implementation of $Q$ relative to $\mathcal{O}$ then there exists
also an implementation of $P$ relative to $\mathcal{O}$.
Informally, an implementation relative to an oracle $\mathcal{O}$
means that all parties are given access to the oracle
$\mathcal{O}$. The notion of relativizing reductions is at least
as strong as the notion of fully black-box reductions. In
particular, in our setting when dealing with impossibility
results, showing that there is no relativizing reduction of $P$ to
$Q$ implies that there is no fully black-box reduction. We note
that, in most cases, this implies that there is no \emph{semi-black-box} reduction as well. Again, the reader is referred
to \cite{ReingoldTV04} for an elaborated discussion.

In this paper, we base our separation result on a slightly weaker
form of relativizing reductions. We show that there exists a pair
of oracles $(\tau, \mathcal{O}^{\tau})$ relative to which $\tau$
is an implementation $P$ (a collection of trapdoor permutations),
and there is no implementation of $Q$ (a statistically hiding
commitment scheme with a low round complexity) of the form
$\Cc^{\tau}$. In other words, our current statement of the result
restricts implementations of $Q$ to have oracle access only to one
of the oracles. Nevertheless, such a separation result still
implies that there is no fully black-box reduction, and in fact,
we have a sound belief that the above restriction is not
essential, \ie that our proof can be generalized to imply that
there is no relativizing reduction. We elaborate more on the
specific flavor of our result and on the possibility of
generalizing our statements in \cref{section:FurtherIssues}. }

\section{The Oracle \Sam and the Separation Oracle}\label{section:SeperationOracle}
In this section we describe the oracle that is later used for proving our lower bounds. The oracle
is of the form $\left( \tau , \Sam^{\tau,\h} \right)$, where $\tau$ is a family of
trapdoor permutations (\ie $\tau \in \Tau$), and $\Sam^{\tau,\h}$ is an oracle that, very informally, receives as input a
description of a circuit $\Cc$ (which may contain $\tau$-gates) and a string $w$, and outputs (using $\h$ as its source of ``randomness'' as described below) a
uniformly distributed preimage of $\Cc(w)$ under the mapping defined by $\Cc$. For generality, we define \Sam for an arbitrary oracle $\oracle$ and not necessarily for $\tau\in \Tau$. In \cref{section:inverting} we use this generalization for first showing that \Sam is not useful for inverting random permutations, and then use this result for proving that \Sam is not useful for inverting random trapdoor permutations.

Moving to the formal description, a valid input (\ie query) to \Sam is a tuple of the form $(w,\Cc,\Cc_\next)$, where $\Cc$ and $\Cc_\next$ are oracle-aided circuits of the same input length $\Inp$, and $w\in \zo^\Inp$. The parameters  $\Cc$ and $w$ are allowed to (simultaneously) take the value $\perp$. Let $\cQ$ stand for the family of all valid queries, and for $\q = (w,\Cc,\Cc_\next) \in \cQ$ let $\Inp(\q)$ stand for the input length of $\Cc_\next$. Let $\calH$ be the ensemble of permutation families $\set{\h = \set{\h_\q}_{\q\in\cQ} : \h_\q \in \Pi_{\Inp(\q)}}$;  that is, each $h\in \calH$ is an infinite set of hash functions, indexed by $\q\in\cQ$. The definition yields that for $\h\la \calH$, the function $\h_\q$ is uniformly random  permutation over $\zo^{\Inp(\q)}$.
We define $\Sam$ as follows.

\begin{algorithm}[\Sam]\label{fig:Sam}

\item {Input:} $\q = (w,\Cc, \Cc_\next) \in \cQ$.

\item {Oracles:} $\oracle$ and $\h \in \calH$.

\item {Operation:} Let $\Inp = \Inp(\q)$.

\begin{itemize}
\item If $\Cc = \bot$, output $\h_\q(0^\Inp)$.

\item Else, output $\h_\q(v)$, where $v$ is the lexicographically smallest $v \in \zo^\Inp$ with $\Cc^\oracle(\h_\q(v)) = \Cc^\oracle(w)$.

\end{itemize}%
\end{algorithm}
\Sam answers arbitrarily on  queries  not in $\cQ$. Note that the input parameter $\Cc_\next$ was merely used to determine the value of $\Inp$, but it will be crucial for the bookkeeping we employ below.

As mentioned in the introduction, algorithm \Sam can be used for inverting \emph{any} oracle, and thus there are no one-way function, or trapdoor permutation, relative to \Sam. Below we define a restricted class of algorithms, called ``normal form algorithms", for which \Sam is not useful for inverting one-way functions, but is useful for breaking the binding of any low round-complexity, or low sender-communication complexity, commitment.

\paragraph{Normal form algorithms.}
Towards defining what normal form algorithms are, we associate the following structure with the queries \Sam is asked upon (the reader is referred to \cref{Figure:Forest} for a specific example).

\begin{definition}[Query forest]\label{def:query forest}
Let $\qv$ be an ordered list $\set{\q_1,w_1,\dots,\q_t,w_t}$ of \Sam queries/answers. A query $\q_j = (\cdot,\cdot, \Cc) \in \qv $ is the {\sf parent}, \wrt $\qv$, of all queries in $\qv$ of the from $\q_i = (w_j,\Cc, \cdot)$ with $i>j$, that do not have a lower index parent in $\qv$. We let $\parent(\q) = \q'$ denote that $\q'$ is the parent of $\q$, and let $\parent(\q) = \perp$ in case $q$ has no parent according to the above definition. The {\sf depth} of $\qv$ is the depth of the above forest.
An oracle-aided algorithm \Ac is of {\sf query depth} $d$, denote a $d$-depth algorithm, if, when given access to \Sam and an $n$-bit input, the resulting queries/answers list it makes to \Sam is of depth at most $d(n)$.
\end{definition}


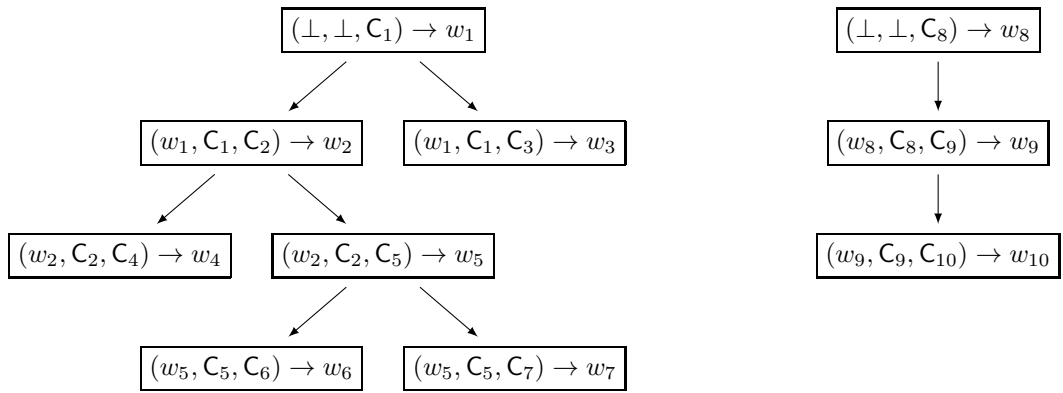
\begin{figure}\label{Figure:Forest}
\begin{small}
\begin{center}
\noindent
\begin{multicols}{2}

\noindent
\begin{tikzpicture}[edge from parent/.style={draw,->,>=latex}]
\node(0){$\boxed{(\bot, \bot, \Cc_1) \rightarrow w_1}$}
    [sibling distance=35mm]
    child{node{$\boxed{(w_1, \Cc_1, \Cc_2) \rightarrow w_2}$}
        child{node{{$\boxed{(w_2, \Cc_2, \Cc_4) \rightarrow w_4}$}}}
        child{node{{$\boxed{(w_2, \Cc_2, \Cc_5) \rightarrow w_5}$}}
            child{node{{$\boxed{(w_5, \Cc_5, \Cc_6) \rightarrow w_6}$}}}
            child{node{{$\boxed{(w_5, \Cc_5, \Cc_7) \rightarrow w_7}$}}}
            }
        }
    child{node{$\boxed{(w_1, \Cc_1, \Cc_3) \rightarrow w_3}$}};
\end{tikzpicture}

\noindent
\begin{tikzpicture}[edge from parent/.style={draw,->,>=latex}]

\node(0){$\boxed{(\bot, \bot, \Cc_8) \rightarrow w_8}$}
    child{node{$\boxed{(w_8, \Cc_8, \Cc_9) \rightarrow w_9}$}
        child{node{{$\boxed{(w_9, \Cc_9, \Cc_{10}) \rightarrow w_{10}}$}}
        }
    };
\end{tikzpicture}
\end{multicols}
\caption{An example of a query forest that consists of two trees.}
\end{center}
\end{small}

\end{figure}

We also formally define what a ``circuit extension" means.
\begin{definition}[circuit extension]\label{def:CircuitExtension}
A circuit $\Cc'$ is a {\sf extension} of an $\Inp$-bit input, $o$-bit output circuit $\Cc$, if $\Cc'$ has $\Inp$ input wires, and the function defined by the first $o$ output wires of $\Cc'$ (assuming some arbitrary order on the wires) is identical to the function defined by the circuit $\Cc$.
\end{definition}
\noindent Namely, a circuit $\Cc'$ is an ``extension" of the circuit $\Cc$, if it contains $\Cc$ as a ``sub-circuit". Equipped with the above two definitions, we define normal form algorithms as follows.
\begin{definition}[normal-form algorithms]\label{def:NormalFormAlgorithm}
An ordered list $\qv=\set{\q_1,w_1,\dots,\q_t,w_t}$ of \Sam queries/answers is in a {\sf normal form}, if $\q_1,\ldots,\q_t\in \cQ$, exists no $i\neq j \in [t]$ with  $q_i = (\cdot,\Cc_\next)$ and $q_j = (\cdot,\Cc_\next)$ for the {\sf same} circuit $\Cc_\next$, and the following holds for every $\q = (w,\Cc\neq \perp, \Cc_\next)\in \qv$:
\begin{enumerate}
 \item $\Cc_\next$ is an extension of $\Cc$, and
 \item $\parent(\q) \neq \perp$.
\end{enumerate}
An oracle-aided algorithm \Ac  is of a normal form, if, when given access to \Sam, the resulting list of queries/answers it makes to \Sam is always in normal form.
\end{definition}
Note that in the query forest defined by a normal-form algorithms, the roots are all of the form $(\cdot,\perp, \cdot)$. The above definitions naturally extend to oracle-aided (families of) circuits, assuming a reasonable order on the circuit gates (see \cref{sec:prelim:Alg}).

While restricted, normal-from algorithms are not at all useless. Specifically, combining the fully black-box reduction from $\Theta(n/\log n)$-round statistically hiding commitment to one-way permutation due to \cite{HaitnerR12,KoshibaS06} (extending \cite{NaorOVY98}) and \cref{theorem:collisionInLowRound}, yields the existence of an $\Theta(n/\log n)$-depth normal-form algorithm, that uses $\Sam^{\pi,\cdot}$ to invert \emph{any} $\pi \in \Pi$. In contrast, in \cref{section:inverting} we show that an $o(n/\log n)$-depth normal-form algorithm \emph{cannot} invert a random $\pi \in \Pi$.

We will also note that an algorithm with oracle access (\ie black-box access) to a normal-form algorithm, and without direct access to \Sam, is a normal-form algorithm by itself.
\begin{proposition}\label{prop:InseritNormalForm}
let \Ac be oracle-aided algorithm, let \Bc be a $d$-depth normal-from algorithm, and let $\Cc$ be the algorithm that given oracle-access to \Sam, acts as $\Ac^{\Bc^\Sam}$ (in particular \Ac does not make direct calls to \Sam). Then algorithm $\Cc$ is in a normal form. Assume further that on input of length $n$, algorithm \Ac calls \Bc on input of maximum length $\ell(n)$, then \Cc is of depth $ d(\ell(n))$.
\end{proposition}
\begin{proof}
 Since \Ac accesses \Bc in a black-box manner, the interaction of \Cc with \Sam is the combined (possibly partial) interactions of \Bc with \Sam done in this execution. Since \Bc is in a normal form, the list of \Sam queries/answers of each of these partial interactions is in a normal form. It follows that the joint list is in such a from, and therefore so is \Cc.

 The depth restriction of \Cc immediately follows from the above observation, and the depth restriction of \Bc.
\end{proof}

\paragraph{Augmented query complexity.}
We use the following measure for the query complexity of \Sam-aided algorithms.
\begin{definition}[augmented query complexity]\label{def:AugmentQC}
The {\sf augmented query complexity} of an algorithm $\Ac$ on input $x$ with oracle access to $\Sam^{\oracle,\h}$, is the number of oracle calls that \Ac makes, counting each call of the form $\Sam(\cdot,\Cc, \cdot)$ as $\aq(\Cc)$ --- the (standard) query complexity of $\Cc$. Algorithm \Ac is has {\sf augmented query complexity \aq} (sometimes denoted, \Ac is a $\aq$-\augQ algorithm), if on input of length $n$, and any choice of $\oracle,\h$, it makes at most $\aq(n)$ augmented queries.
\end{definition}

\paragraph{Trivial circuit extension.}
While the above definition dictates \Sam to use the \emph{same} ``randomness" when queried twice on the same query $\q$ (same function $\h_\q$ is used), it is simple to effectively make \Sam to use \emph{independent} randomness on the ``same" query (\ie by making a dummy change, one that does not effect the circuit input/output behaviour,  to the circuit part of $\q$).

\begin{definition}[trivial circuit extension]
A circuit $\Cc'$ is a {\sf trivial extension} of the circuit $\Cc$, if both circuits computes the {\sf same} function. For a circuit $\Cc$ and $i\in \N$, let $\extension_i(\Cc)$ be the circuit $\Cc$ augmented with $i$ OR gates that have no effect on the output (\ie their output is ignored).
\end{definition}
Note that $\set{\extension_i(\Cc)}_{i\in \N}$ are \emph{distinct, trivial} extensions of $\Cc$.

\section{The Power of \Sam}\label{section:PowerOfSam}
In this section we present normal-form algorithms that use \Sam for finding collisions in \emph{any} protocol of low round complexity, or of low communication complexity, aided with \emph{any} oracle. Namely, a cheating (normal-form) party can use \Sam to interact with the other party such that the following hold: 1) at the end of the protocol the cheating party outputs several independent random inputs that are consistent with the execution of the protocol, and 2) the transcript of the resulting execution has the same distribution as of a random honest execution of the protocol. The case of low round complexity follows directly from the definition of \Sam, where for the low communication complexity case we have to work slightly harder. Along the way, we show that \Sam can be used to find collisions in any oracle-aided function with ``short" outputs (\ie the function output is significantly shorter than its input).

The depth parameter of the attackers presented below are functions of the round or communication complexity of the protocol, or of the function's output length. For being useful in applications such as the ones given in \cref{section:LowerBounds,section:implications}, this parameter needs to be ``small". Hence, the attacker described below are only useful for the type of protocols and functions we considered above.

\subsection{Finding Collisions in Protocols of Low Round Complexity}\label{sec:CollisionInLowRoundCOmp}
In the following we focus on no-input protocols that get the security parameter $1^n$ as their common input.

\begin{theorem}\label{theorem:collisionInLowRound}
For every $d$-round, $t$-query oracle-aided protocol $\left(
\Ac, \Bc\right)$, there exists a deterministic normal-from algorithm $\As$ such that the following hold for every $n,k\in \N$ and function $\co$: let $\H$ be uniformly distributed over $\calH$ and let $(\Transs,(R_1,\dots,R_k)) = \exec{(\As^{\co,\Sam^{\co, \H}}(1^k), \Bc^\co)(1^n)}_{\trans,\out^\As}$,\footnote{I.e., the common transcript and $\Ac$'s output in a random execution of $(\As^{\Sam^{\co, \h}}(1^k), \Bc^\co)(1^n)$.} then
\begin{enumerate}
 \item $\Transs$ has the same distribution as $\Trans = \exec{(\Ac^\co, \Bc^\co)(1^n)}_\trans$.
 \item $R_1,\ldots, R_k$ are sampled independently from the distribution over the random coins of $\Ac$ that are consistent with $\co$, $\Transs$ and $1^n$.

 \item $\As$ makes queries of depth at most $d(n)+1$.
 \item $\As$ makes  $k + d(n)$ \Sam-queries, all on $t(n)$-query circuits.

  \item Assuming $\Ac$ is a \ppt, then $\As$ runs in time $p(n)\cdot k$ for some $p\in \poly$.

\end{enumerate}
\end{theorem}
\begin{proof}
For ease of notation, we assume that $\Bc$ sends the first message in $(\Ac,\Bc)$. We fix $n$ and $k$, and omit the security parameter $1^n$ whenever its value is clear from the context.

In order for $\As$ to interact with \Sam, it identifies $\Ac$ with the sequence of circuits $\Ac_1, \ldots, \Ac_d$ for which the following is an accurate description of $\Ac$'s actions: upon reliving the $i$'th message $b_i$ from \Bc, $\Ac$ sends $a_i = \Ac_i (r_\Ac, (b_1,\dots,b_i))$ to $\Bc$, where $r_\Ac$ are $\Ac$'s random coins, and $b_1,\dots,b_{i-1}$ are the first $i-1$ messages sent by $\Bc$. We assume \wlg that each message $a_i$ contains the previous messages $a_1, \ldots a_{i-1}$ as its prefix, and therefore each circuit $\Ac_i$ is an \emph{extension} of $\Ac_{i-1}$ (as discussed in \cref{section:SeperationOracle}). Note that assuming $\Ac$ is a \ppt, then descriptions of the circuits $\Ac_1, \ldots, \Ac_d$ can be computed in polynomial-time from the description of
$\Ac$.

Given the above discussion, the oracle-aided interactive algorithm $\As$ is defined as follows.
\begin{algorithm}[\As]\label{alg:As}
\item[Input:] $1^n$ and $k\in \N$.

\item [Oracles:] $\co$ and $\Sam^{\co, \h}$ for some $h \in \calH$.

\item [Operation:]~
\begin{description}
\item[Round $1 \leq i \leq d = d(n)$:] upon receiving the $i$'th message $b_i$ from \Bc.
\begin{enumerate}
\item Let $\Ac_{b_1,\cdots,b_i}$ be the circuit $\Ac_i$ defined above with $(b_1,\cdots,b_i)$ \emph{fixed} as its second input (\ie as \Bc's first $i$ messages).

\item In case $i=1$ (first round), set $r_1 = \Sam^{\co, \h}(\bot, \bot,\Ac_{b_1})$.\\
Otherwise, set $r_i = \Sam^{\co, \h}(r_{i-1},\Ac_{b_1,\cdots,b_{i-1}}, \Ac_{b_1,\cdots,b_i})$.

\item Send $\Ac^\co_{(b_1,\cdots,b_i)}(r_i)$ back to $\Bc$.
\end{enumerate}

\item [Output phase:]~
\begin{enumerate}
\item For $j\in [k]$ set $r_{d,j} = \Sam^{\co, \h}(r_d,\Ac_{b_1,\cdots,b_d}, \extension_j(\Ac_{b_1,\cdots,b_d}))$.\footnote{Recall that $\set{\extension_j(\Ac_{b_1,\cdots,b_d}))}_{j\in [d]}$ are arbitrary \emph{distinct} extensions of $\Ac_{b_1,\cdots,b_d}$. The role of these extensions is to make \Sam to use fresh randomness in each call (\ie to apply a different part of $\h$).}

\item Output $r_{d,1},\dots,r_{d,k}$.
\end{enumerate}

\end{description}
\end{algorithm}
Note that the only role of the circuits $\set{\extension_j(\Ac_{b_1,\cdots,b_d}))}_{j\in [d]}$ used in the output phase of the above description of $\As$, is causing \Sam to use \emph{independent} randomness per call (\ie by using a different function from $\h$). It is also easy to verify that $\As$ queries $\Sam$ up to depth at most $d(n) + 1$, performs at most $(d(n) + k)\cdot t(n)$ augmented oracle queries, and that $\As$ runs in polynomial time (excluding the oracle calls) assuming that $\Ac$ is a \ppt. Since $\As$ queries \Sam up to depth $d(n) + 1$ and the assumption that $\Ac_i$ is an extension of $\Ac_{i-1}$, yields that $\As$ is indeed in a normal form. The above observation yields that $r_{d,1},\ldots, r_{d,k}$ are independently distributed conditioned on $\trans$, where each of them is uniformly distributed over the random coins of $\Ac$ that are consistent with $\co$ and $\trans$ -- the transcript generated by the interaction of $(\As^{\Sam^{\co, \H}}(1^k),\Bc^\co)$. Hence, for completing the proof all we need to prove is that the transcript induced by a random execution of $(\As^{\Sam^{\co, \H}}(1^k), \Bc^\co)$, which we denote here by $\Transs$, has the same distribution as that induced by a random execution $(\Ac^\co, \Bc^\co)$, denoted here as $\Trans$.
\begin{claim}\label{claim:TransHasTheRightDIst}
$\Transs$ and $\Trans$ are identically distributed.
\end{claim}
\begin{proof}
Notice that in each round of the protocol, $\As$ acts exactly like $\Ac$ would on the given (partial) transcript. That is, like $\Ac$ does on random coins that are sampled according to the right distribution: the distribution of $\Ac$'s coin in a random execution of $\left(\Ac, \Bc\right)$ that yields this transcript. The formal (and somewhat tedious) proof follows.

The proof is by induction on $i$, the number of messages sent so far in the protocol, that $\Transs_{1,\dots,i}$ and $\Trans_{1,\dots,i}$ are identically distributed. The base case $i=0$ is trivial. In the following we condition on $\Transs_{1,\dots,i}= \Trans_{1,\dots,i} = \trans$, and prove that under this conditioning $\Transs_{1,\dots,i+1}$ and $\Trans_{1,\dots,i+1}$ are identically distributed.

Note that both in $(\As^{\Sam^{\co,\h}},\Bc^\co)$ and in $(\Ac^\co,\Bc^\co)$, the distribution of the (conditional) parties' joint view, is a \emph{product} distribution. (This hold since the only oracle shared by the parties, \ie $\co$, is fixed.) In particular, the distribution of $\Bc$'s coins in both protocols is uniform over the possible coins for $\Bc$ that are consistent with $\co$ and $\trans$ (and the definition of \Bc). Since the next message of a party is a deterministic function of $\co$, $\trans$ and its random coins, in case the $i+1$ message is in $\Bc$'s control, it holds that $\Transs_{1,\dots,i+1}$ and $\Trans_{1,\dots,i+1}$ are identically distributed.

The complimentary case, where the $i+1$ message is in $\As$'s or in $\Ac$'s control, is slightly more complicated. Note that the $i+1$ message sent by $\As$ is determined by the value of $r_{i+1}$, returned by \Sam, exactly in the same way that the $i+1$ message sent by \Ac is determined by its random coins $r_\Ac$; in both cases, the same deterministic function is applied to $\co$, $\trans$ and the coins. We complete the proof showing that $r_{i+1}$ and $r_\Ac$ are identically distributed.

Similarly to the coins of \Bc discussed above, $r_\Ac$ are uniformly distributed over the possible coins for \Ac that are consistent with $\co$ and $\trans$ (and the definition of \Ac). The value of $r_{i+1}$ on the other hand, is determined by value of $\H_{\q_{i+1}}$, where $\q_{i+1}$ is the query \As makes to \Sam in the $i+1$ round. Since $\H_{\q_{i+1}}$ was not queried by \Sam in the first $i$ rounds of $(\As,\Bc)$, under the above the conditioning $\H_{\q_{i+1}}$ is a uniformly chosen permutation over the coins of $\Ac$. Hence, the definition of \Sam yields that, again, under the above conditioning, the coins it returns are uniformly distributed over the coins of $\Ac$ that are consistent with $\co$ and $\trans$, yielding that $r_{i+1}$ and $r_\Ac$ are identically distributed.
\end{proof}
\end{proof}

\subsection{Inverting Functions of Short Outputs}\label{sec:CollisionInShortOutputFunc}
In this section we show how to use \Sam to invert any function (\ie deterministic algorithm) with oracle access to a trapdoor permutation oracle, given that the function output is ``short". Combined with the results of \cref{section:inverting}, this would imply, for instance, that it is impossible to use in a fully black-box manner an $n$-bit one-way function to construct an $o(n)$-bit one-way function.\footnote{Note that the following theorem does not stand in contradicting with the one-wayness of a random permutation in the presence of $\Sam$, proved in \cref{section:inverting}. The functions in consideration there have long outputs.}

\begin{theorem}\label{theorem:collisionInShortOutputFunc}
For every $t$-query oracle-aided function $f\colon \zn \mapsto \zo^{\ell(n)}$, there exists a deterministic normal-from algorithm \Inv such that the following holds for every $n,k,d \in \N$, $\eps \in (0,1]$ and a function $\co$: let $(X_1,\dots,X_k)= \Inv^{\co,\Sam^{\co, \H}}(1^n,k,d,\eps,f^\co(X))$, where $X$ and $H$ are uniformly chosen from $\zn$ and $\calH$ respectively, then
\begin{enumerate}
 \item $\Pr[(X_1,\dots,X_k) = \perp] \leq \eps$.

 \item Conditioned on $(X_1,\dots,X_k) \neq \perp$, the variables $X_1,\dots,X_k$ are iid over $(f^\co)^{-1}(f^\co(X))$.

\item $\Inv$ makes queries of depth at most $d+1$.

 \item $\Inv$ makes at most $k + d \cdot 2^{\ceil{\ell(n)/d}}/\eps$ \Sam queries, all on $t$-query circuits.

 \item Assuming  $f$ is polynomial-time computable, then \Inv runs in time $p(n)\cdot (2^{\ceil{\ell(n)/d}} + k)$, for some $p\in \poly$.
\end{enumerate}
\end{theorem}

\begin{proof}
Fix $n$, $k$ and $\co$. For ease of notation we and omit the security parameter $1^n$ and assume $\ell$ is a multiple of $d$. Let $v = \ell/d$, and for $x\in \zn$ and $i\in [v]$, let $f(x)_{(i)}$ denote the $i$'th block of $f(x)$, \ie $f(x)_{(i-1)d +1,\dots,id}$. For $i\in [d]$, let $f_i$ be the circuit that on input $x\in \zn$ outputs $f(x)_{(1)},\ldots,f(x)_{(i)}$. We invert $f$ on $y = (y_1,\dots,y_d)\in (\zo^v)^d$, by gradually causing \Sam to output $x_i$ with $f_i(x_i) = (y_1,\dots,y_i)$ for $i=1$ to $d$. Doing that for $i=1$ is easy: keep calling $\Sam$ on input $(\perp,\perp,f_1)$, until it returns $x_1$ with $f_1(x_1) = y_1$. Since a call to $\Sam(\perp,\perp,f_1)$ returns \emph{uniform and independent} element in $\zn$, about $2^v$ \Sam calls yield the desired answer. Assuming that we have successfully made \Sam to answer on $(\cdot, \cdot, f_{i-1})$ with $x_{i-1}$ such that $f_{i-1}(x_{i-1}) = (y_1,\dots,y_{i-1})$, we make \Sam answer with $x_i$ such that $f_i(x_i) = (y_1,\dots,y_i)$ using similar means to the ones used to get $x_i$; keep calling $\Sam$ on input $(x_{i-1},f_{i-1},f_i)$, until it returns the right $x_i$. As in the first round, about $2^v$ \Sam calls suffices to get the desired answer. The formal definition of algorithm $\Inv$ is given below.

\begin{algorithm}[\Inv]\label{alg:Inv}~
\item[Input:] $1^n$, $k,d \in \N$, $\eps \in (0,1]$ and $y = (y_1,\dots,y_d)\in (\zo^v)^d$.

\item [Oracles:] $\co$ and $\Sam^{\co, \h}$, for some $h\in \calH$.

\item [Operation:]~

\begin{enumerate}
 \item For $i=1$ to $d$ do: \label{step:mainLoop}

 Set $j =0$, and do the following loop:
\begin{enumerate}
\item $j^{++}$.

 \item Let $x_i = \Sam^{\co, \h}(x_{i-1}, f_{i-1}^\ast, \extension_j(f_i))$. In the case $i=1$, set $f_{i-1}^\ast = x_{i-1} = \perp$.

 \item If $f_i(x_i) = (y_1,\dots,y_i)$, set $f_i^\ast = \extension_j(f_i)$ and break the inner loop.

 \item If \emph{overall} number of \Sam calls exceeds $d \cdot 2^v/\eps$, return $\perp$ and abort.
\end{enumerate}

\item For $j=1$ to $k$: set $x_{d,j} = \Sam^{\co, \h}(x_d,f_d^\ast, \extension_j(f_d^\ast))$.

 \item Return $x_{d,1},\dots,x_{d,k}$.
\end{enumerate}

\end{algorithm}

The second and third properties of $\Inv$ immediately follow from the definition of \Sam, so the only interesting part is showing that \Inv aborts (\ie outputs $\perp$) with probability at most $\eps$. Let \tInv be the unbounded version of \Inv, \ie Step $1.(d)$ is removed. It is clear that \tInv's output is identical to that of \Inv conditioned on \Inv not aborting, and that the probability that \Inv aborts is the probability that \tInv make more than $d \cdot 2^v/\eps$ \Sam calls. We show that the expected number of \Sam calls made by \tInv is bounded by $d \cdot 2^v$, and proof follows by a Markov bound.

We bound the expected number of overall \Sam calls made by \tInv in a single round of Step $1$, and the proof follows by linearity of expectation. Fix a value for $y_1,\ldots,y_{i-1}$. Let $Y=f(X)_{(i)}$ conditioned that $f_{i-1}(X) = y_1,\ldots,y_{i-1}$, and let $Y_j$ be the value of $f(x_{i})_{(i)}$ sampled in the $j$'th inner loop of a random execution of $\tInv(k,d,\eps,y_1,\ldots,y_{i-1},\ldots)$. If less than $j$ inner loops happen, we let $Y_j$ be an independent copy of $Y_1$. The definition of $\Sam$ yields that over a random choice of $\h$, the variables $Y,Y_1,Y_2,\ldots$ are iid over $\zo^v$. It follows that $\pr{Y=Y_j} \geq 2^{-v}$ for every $j$, and the expected value of the first $j$ with $Y_j =Y$ is bounded by $2^{v}$. Hence, the expected number of \Sam calls made by $\tInv$ (over the choice of $X$ and $\h$) is bounded by $d \cdot 2^v$.
\end{proof}

\subsection{Finding Collisions in Low Communication Complexity Protocols}\label{sec:CollisionInLowComCOmp}
The following theorems show how to find collision in protocols in which the communication of the ``attacking" party is low.

\begin{theorem}\label{theorem:collisionInLowComComp}
Let $\pi= \left(\Ac, \Bc\right)$ oracle-aided protocol in which $\Ac$, on input of length $n$, makes at most $t(n)$ oracle-queries and sends at most $c(n)$ bits. Then there exists a deterministic normal-from algorithm \Inv such that the following holds for every $n,k,d\in \N$, $\eps \in (0,1]$ and function $\co$: let $\H$ be uniformly distributed over $\calH$ and let $(X_1,\dots,X_k)= \Inv^{\co,\Sam^{\co, \H}}(1^n,1^k,d,\eps,\Trans)$, where $\Trans = \exec{\pi^\co)(1^n)}_\trans$, then
\begin{enumerate}
 \item $\Pr[(X_1,\dots,X_k) = \perp] \leq \eps$,
 \item Conditioned on $(X_1,\dots,X_k) \neq \perp$, the variables $X_1,\dots,X_k$ are iid over the random coins of \Ac that are consistent with $\Trans$.

\item $\Inv$ makes queries of depth at most $d+1$, and

\item $\Inv$ makes at most $k+ d \cdot 2^{\ceil{c(n)/d}}/\eps$ \Sam-queries, all on $t(n)$-query circuits.

\item Assuming that $\pi$ is polynomial-time computable, then \Inv runs in time $p(n)\cdot (2^{\ceil{c(n)/d}} + k)$, for some $p\in \poly$.

\end{enumerate}
\end{theorem}

\begin{remark}[Comparing \cref{theorem:collisionInLowComComp} to \cref{theorem:collisionInLowRound}]
Both \cref{theorem:collisionInLowRound} and \cref{theorem:collisionInLowComComp} are useful for finding collisions in the given protocols. While the attacker of \cref{theorem:collisionInLowRound} never fails, the attacker of \cref{theorem:collisionInLowComComp} (who might fail) has the advantage of not using \Sam through the execution, but only after it ends. We use this property in \cref{section:ComCompLB} to rule out constructions of \emph{honest-sender} low sender-communication commitments from trapdoor functions.
\end{remark}

\begin{proof}[Proof of \cref{theorem:collisionInLowComComp}]
We start by assuming that $\Bc$ is deterministic. Let $f\colon \zn \mapsto \zo^{c(n)}$ map $\Ac$'s random coins to the messages it send to $\Bc$ in $\pi$. Consider the algorithm \Invd that on input $(x,\trans)$ returns $\Invf(x,\trans_\Ac)$ (with the same oracles), for \Invf being the inverter \cref{theorem:collisionInShortOutputFunc} guarantees for the function $f$, and $\trans_\Ac$ being \Ac's part in $\trans$. By \cref{theorem:collisionInShortOutputFunc}, algorithm \Invd satisfies the first three and fifth properties, stated in the theorem, and makes at most $k+ d \cdot 2^{\ceil{c(n)/d}}/\eps$ \Sam-queries. Algorithm \Invd, however, might apply \Sam on circuits of query complexity larger than $t(n)$ (as they contain the queries made by \Bc).

Consider the following variant of \Invd. For a transcript $\trans$ of $\pi$, let $g_\trans \colon \zn \mapsto \zo^{c(n)}$ map $\Ac$'s random coins to the messages it sends to $\Bc$ in $\pi$, assuming that  \Bc sends \Ac the message it sends in $\trans$. On input $(x,\trans)$, algorithm \Inv returns $\Inv_{g_\trans}(x,\trans_\Ac)$ (with the same oracles), for $\Inv_{g_\trans}$ being the inverter \cref{theorem:collisionInShortOutputFunc} guarantees for the function $g_\trans$, and $\trans_\Ac$ being \Ac's part in $\trans$. The point to notice is that by construction, on the same input and a random choice of $\h$, algorithms \Inv and \Invd have \emph{exactly} the same output distribution. In follows that \Inv satisfies all the properties satisfied by \Invd, where by construction, on only invoke \Sam on $t(n)$-query circuits. Furthermore, since the implementation of \Inv in obvious to the definition of \Bc, it has the same success probability also when considering a probabilistic \Bc.
\end{proof}

\section{Random Permutations are Hard for Low-Depth Normal-Form Algorithms}\label{section:inverting}

In this section we prove that for low-depth normal-from algorithms, \Sam is not useful for inverting random permutations and random trapdoor permutations. We start with random permutations, and then extend the result to random trapdoor permutations.

Following \cite{GennaroGKT05}, we state our results in the stronger non-uniform setting. Hence, our goal is to upper bound the success probability of a circuit family having oracle access to \Sam in the task of inverting a uniformly chosen permutation $\pi \in \Pi$ on a uniformly chosen image $y \in \zn$. We relate this success probability to the maximal depth of the \Sam-queries made by the circuit family and to the augmented query complexity of the family (see \cref{def:AugmentQC}). We prove the following theorem.

\begin{theorem}\label{theorem:RandomPermutationHardForSam}
The following holds for large enough $n\in \N$: for every $\aq$-\augQ, $d$-depth, normal-form circuit \Ac such that $\aq^{3 d
+ 1} < 2^{n/8}$, it holds that
\begin{align*}
\prob{\MyAtop{\pi \la \Pi, \h \la \calH}{y \la \zn}}{\Ac^{\pi, \Sam^{\pi, \h}} (y) = \pi^{-1}(y)} \leq 2/\aq.
\end{align*}
\end{theorem}

Before turning to prove \cref{theorem:RandomPermutationHardForSam}, we first provide a brief overview of the structure of the proof. Consider a normal-from circuit \Ac trying to invert an input $y \in \zn$ (\ie to find $\pi^{-1}(y)$), while having oracle access to both $\pi$ and \Sam. We distinguish between two cases: one in which \Ac obtains information on the value $\pi^{-1}(y)$ via one of its \Sam-queries, and the other in which none of \Ac's \Sam-queries provides sufficient information for retrieving
$\pi^{-1}(y)$. Specifically, we define:

\begin{definition}[Hits]\label{definition:hit}
An execution $\Ac^{\pi, \Sam^{\pi, \h}} (y)$ is {\sf hitting}, denoted by the event $\Hit_{\Ac,\pi,\h}(y)$, if $\Ac$ makes a \Sam-query $\q = (\cdot, \Cc,\cdot)$, replied with $w$ such that the computation $\Cc^\pi(w)$ queries $\pi$ on $\pi^{-1}(y)$.
\end{definition}

The proof proceeds in two modular parts. In the first part of the proof, we consider the
case that the event $\Hit(y) = \Hit_{\Ac,\pi,\h}(y)$ does not occur, and prove a ``reconstruction lemma'' that extends an information-theoretic argument of \citet{GennaroT00}. They showed
that if a circuit \Ac manages to invert a permutation $\pi$ on a relatively large set of images,
then this permutation has a rather short representation given \Ac. We generalize their argument to deal
with circuits having oracle access to \Sam. In this part we do not restrict the depth of \Ac, neither require it to be in a normal form.

\begin{lemma}\label{lemma:reconstruction1}
The following holds for large enough $n \in \N$: let \Ac be a $2^{n/5}$-\augQ circuit, then
\begin{align*}
\prob{\pi \la \Pi, \h \la \calH}{\prob{y \la \zn}{\Ac^{\pi, \Sam^{\pi, \h}} (y) = \pi^{-1} (y)
\enspace \land \neg \Hit_{\Ac,\pi,\h}(y)} \geq 2^{-n/5}} \leq 2^{-2^{\frac{3n}5}}.
\end{align*}
\end{lemma}
Namely in the ``non-hitting case", oracle access to \Sam does not improve ones chances to invert a random permutation.

In the second part of the proof, we show that the case where the event $\Hit(y)$ does occur,
can be reduced to the case where the event $\Hit(y)$ does not occur. Specifically,
given a circuit \Ac that tries to invert a permutation $\pi$, we construct a circuit \Mc that
succeeds almost as well as \Ac, \emph{without} \Mc's \Sam-queries producing any $y$-hits. For this part, the query complexity of the circuit, its depth restriction and it being in a normal form, all play an instrumental role.
\begin{lemma}\label{lemma:Hitting}
For every $\aq$-\augQ, $d$-depth normal-from circuit \Ac there exists a $2\aq$-\augQ circuit \Mc such that the following holds: assuming that
\begin{align*}
\prob{\MyAtop{\pi \la \Pi, \h \la \calH}{y \la \zn}}{\Hit_{\Ac,\pi,\h}(y)} \geq \eps
\end{align*}
for $\eps \in [0,1/\aq]$, then
\begin{align*}
\prob{\MyAtop{\pi \la \Pi, \h \la \calH}{y \la \zn}}{\Mc^{\pi, \Sam^{\pi, \h}} (y) = \pi^{-1}(y) \land \neg \Hit_{\Mc,\pi,\h}(y)} \ge (\eps/2)^{3 d + 1}.
\end{align*}
\end{lemma}
In what follows we show that \cref{theorem:RandomPermutationHardForSam} is a straightforward corollary of \cref{lemma:reconstruction1,lemma:Hitting}. In \cref{subsection:TDPHardForSam} we
extend our statement to deal with trapdoor permutations. Then, in \cref{subsection:reconstruction,subsection:hitting} we prove \cref{lemma:reconstruction1,lemma:Hitting}, respectively.

\begin{proof}[Proof of \cref{theorem:RandomPermutationHardForSam}]
Assume towards a contradiction that for infinitely many $n$'s, there exists a $\aq$-query, $d$-depth normal-from circuit \Ac such that $\aq^{3 d + 1} < 2^{n/8}$ and
\begin{align*}
\prob{\MyAtop{\pi \la \Pi, \h \la \calH}{y \la \zn}}{\Ac^{\pi, \Sam^{\pi, \h}} (y) = \pi^{-1}(y)} \geq 2/\aq.
\end{align*}
Consider now the circuit $\Ac'$ that emulates $\Ac$ and makes sure that whenever $\Ac$ inverts $y$ then the event $\Hit_{\Ac',\pi,\h}(y)$ occurs. Note that $\Ac'$ can be easily implemented based on $\Ac$ by performing two additional queries to $\Sam$ (containing a circuit with a $\pi$-gate that has hardwired the output of $\Ac$). Thus, $\Ac'$ is a $(\aq+2)$-query $d$-depth normal-from circuit, and it holds that
\begin{align*}
\prob{\MyAtop{\pi \la \Pi, \h \la \calH}{y \la \zn}}{\Hit_{\Ac',\pi,\h}(y)} \geq 2/\aq.
\end{align*}
\cref{lemma:Hitting} implies that for infinitely many $n$'s there exists an $(2(\aq+2) \leq 2^{n/7})$-\augQ circuit \Mc such that
\begin{align*}
\prob{\MyAtop{\pi \la \Pi, \h \la \calH}{y \la \zn}}{\Mc^{\pi, \Sam^{\pi, \h}} (y) = \pi^{-1}(y) \land \neg \Hit_{\Mc,\pi,\h}(y)} \ge \left(\frac{1}{\aq}\right)^{3 d + 1} > \frac{1}{2^{n/8}},
\end{align*}
in contradiction to \cref{lemma:reconstruction1}.
\end{proof}

\subsection{Extension to Trapdoor Permutations}\label{subsection:TDPHardForSam}
We prove the following theorem:
\begin{theorem}\label{theorem:TDPHardForSam}
For $\aq$-\augQ, $d$-depth normal-form circuit \Ac with $(3\aq)^{3 d
+ 1} < 2^{n/8}$ and large enough $n$, it holds that
\begin{align*}
\alpha \eqdef \prob{\MyAtop{\tau = ( G, F, F^{-1})\la \Tau, \h \la \calH}{td \la \zn, y \la \zn}}{\Ac^{\tau, \Sam^{\tau, \h}} (G (td),y) = F^{-1}(td,y)} \leq 4/\aq.
\end{align*}
\end{theorem}
Assume \Ac inverts $F$ with probability $5/\aq$. If \Ac queries $F^{-1}(td,\cdot)$ with probability $2.5/\aq$, then it can be used to invert (the random permutation) $G$ with this probability, in contradiction to \cref{theorem:RandomPermutationHardForSam}. If the latter does not happen, then \Ac inverts $F$ with probability $2.5/\aq$ \emph{without} using $F^{-1}$, which is again in contradiction to \cref{theorem:RandomPermutationHardForSam}. Formal proof follows.

\begin{proof}
Let $n$ be sufficiently large as required for \cref{theorem:RandomPermutationHardForSam}. For $\tau \in \Tau$ and $td\in \zn$, let $\tau_{\neg td}$ be the variant of $\tau$ that answers on queries of the form $F^{-1}(td,\cdot)$ with $\perp$. We claim that
\begin{align}\label{eq:brr}
\beta \eqdef\prob{\MyAtop{\tau = ( G, F, F^{-1}) \la \Tau, \h \la \calH}{td \la \zn,y \la \zn}}{\Ac^{\tau_{\neg td}, \Sam^{\tau_{\neg td}, \h}} (G(td),y) = F^{-1}(td,y)} > \alpha - 2/\aq.
\end{align}
Assuming \cref{eq:brr} holds, then it still holds for some fixing of $G$, $td$ and $\set{F_{pk'}}_{pk' \neq G(tk)}$. Hence, there exists a $\aq$-\augQ, $d$-depth normal-form circuit \Bc, with the above fixing ``hardwired" into it, such that
\begin{align*}
\prob{\MyAtop{F_{pk} \la \Pi, \h \la \calH}{y \la \zn}}{\Bc^{F_{pk}, \Sam^{F_{pk}, \h}} (\pi(y)) = (F_{pk})^{-1}(y)} > \alpha - 2/\aq,
\end{align*}
and \cref{theorem:RandomPermutationHardForSam} yields that $\alpha \leq 4/\aq$.

The rest of the proof is devoted for proving \cref{eq:brr}. The augmented queries made by an algorithm with oracle to \Sam, are those queries made by the algorithm directly, plus those queries made by $\Cc(w)$ and $\Cc(w')$, for each query $w'= \Sam(w,\Cc,\cdot)$ made by the algorithm. It is easy to verify that
\begin{align}\label{eq:brr1}
\prob{\MyAtop{\tau = ( G, F, F^{-1}) \la \Tau, \h \la \calH}{td \la \zn,y \la \zn}}{\Ac^{\tau_{\neg td}, \Sam^{\tau_{\neg td}, \h}} (G(td),y) \mbox{ makes an augmented query $F^{-1}(td,\cdot)$}} \geq \alpha - \beta
\end{align}
For a circuit $\Cc$ and $pk \in \zn$, let $\Cc^{pk}$ be the variant of $\Cc$ that before each query of the form $F^{-1}(td,\cdot)$, it queries $G$ on $td$, and if the answer is $pk$, it replies to the query $F^{-1}$ with $\perp$ (without making the call). Let $\Dc$ the variant of $\Ac^{pk}$ that on input $(pk,y)$, replaces each \Sam query $(\cdot,\Cc,\Cc_\next)$ done by $\Ac^{pk}$, with the query $(\cdot,\Cc^{pk},\Cc_\next^{pk})$. That is, $\Dc^{\tau, \Sam^{\tau, \h}}(G(td),y)$ emulates $\Ac^{\tau_{\neg td}, \Sam^{\tau_{\neg td}, \h}}(G(td),y)$. It follows that
\begin{align}
\prob{\MyAtop{\tau = ( G, F, F^{-1}) \la \Tau, \h \la \calH}{td \la \zn,y \la \zn}}{\Dc^{\tau, \Sam^{\tau, \h}} (pk=G(td),y) \mbox{ makes the augmented query $G(td)$}} \geq \alpha - \beta
\end{align}
Let $\Ec$ be the variant of $\Dc$, that if one of its augmented queries is of the form $G(td') = pk$, it halts and return $td'$. It is clear that
\begin{align}
\prob{\MyAtop{\tau = ( G, F, F^{-1}) \la \Tau, \h \la \calH}{td \la \zn,y \la \zn}}{\Ec^{\tau, \Sam^{\tau, \h}} (pk=G(td),y) = td} \geq \alpha - \beta
\end{align}
In particular, there exists a \emph{fix} value $(y, F)$ for which the above holds \wrt this fixing.\footnote{By the definition of $\Tau$, such fixing does not change the distribution of $G$ (\ie $G$ is a uniform random permutation giving this fixing).} Let $I$ be the function that inverts $F$ given only the public key. That is, $I(pk,y) = (F_{pk})^{-1}$ (recall that $F_{pk}(y) = F(pk,y)$). Let $\Mc$ the variant of \Ec with this fixed value of $(y, F,I)$ ``hardwired" into it that replaces each call $F^{-1}(td',y')$ made by \Ec, with $I(G(td',y'))$. It is clear that
\begin{align*}
\prob{\MyAtop{G \la \Pi, \h \la \calH}{td \la \zn}}{\Mc^{G, \Sam^{G, \h}} (G(td)) = G^{-1}(td)} \geq \alpha - \beta,
\end{align*}
and, by inspection, \Mc is a $3\aq$-\augQ, $d$-depth normal-form circuit. \cref{theorem:RandomPermutationHardForSam} yields that $\alpha - \beta \leq 2/\aq$, and \cref{eq:brr} follows.
\end{proof}

\subsection{The Reconstruction Lemma --- Proving \texorpdfstring{\cref{lemma:reconstruction1}}{\cref{lemma:reconstruction1}}}\label{subsection:reconstruction}
The following extends the reconstruction lemma of \citet{GennaroT00}. The idea
underlying the claim is the following: if a circuit \Ac manages to invert a permutation $\pi$ on
some set, then given the circuit \Ac, the permutation $\pi$ can be described without specifying its
value on a relatively large fraction of this set.

\begin{claim}\label{claim:reconstruction2}
There exists a deterministic algorithm \Decoder such that the following holds for every $\aq$-\augQ circuit \Ac, $\pi\in \Pi$, $\h \in \calH$ and $n\in \N$. Assuming that
\begin{align*}
\prob{y \la \zn}{\Ac^{\pi, \Sam^{\pi, \h}} (y) = \pi^{-1}_n (y)
 \land \neg \Hit_{\Ac,\pi,\h}(y)} \ge \epsilon,
\end{align*}
then there exists an $\left(2 \log \binom{2^n}{a} + \log ((2^n - a)!)\right)$-bit string \aux, such that $\Decoder(\aux, \Ac, \h,\pi_{-n}) = \pi_n$, where $a \ge \epsilon 2^n / \left( 2 \aq\right)$ and $\pi_{-n} = \set{\pi_i}_{i\in \N\setminus \set{n} }$.
\end{claim}

\begin{proof}
Denote by $\cI \subseteq \zn$ the set of points $y \in \zn$ on which $\Ac^{\pi, \Sam^{\pi, \h}}$ successfully inverts $\pi_n$ with no $y$-hits. We claim that there exists a relatively large set $\cY \subseteq
\cI$, such that the value of $\pi^{-1}_n$ on the set $\cY$, is determined by the description of, \Ac, $\h$, $\pi_{-n}$, and the set $\cZ = \set{(y,\pi_n^{-1}(y))\colon y \in \zn \setminus \cY}$.

The set $\cY$ is defined via the following process.
\begin{algorithm}~
\item Set $\cY = \emptyset$, and repeat until $\cI = \emptyset$:
\begin{enumerate}
 \item Remove the lexicographically smallest element $y$ from $\cI$ and insert it into $\cY$.
 \item Let $\set{(w_1,\Cc_1,\cdot),w_1',\dots,(w_\aq,\Cc_\aq,\cdot),w_\aq'}$ be the queries made by $\Ac^{\pi,\Sam^{\pi,h}}(y)$ to \Sam and their answers, and let $y_1, \ldots, y_\aq$ be the outputs of the $\pi_n$-gates in the computations of $\Cc^\pi_1(w_1), \Cc^\pi_1(w'_1), \ldots, \Cc^\pi_\aq(w_\aq),\Cc^\pi_\aq(w'_\aq)$ and the outputs of all \Ac's direct queries to $\pi_n$. Then, remove $y_1, \ldots, y_\aq$ from $\cI$.
\end{enumerate}
\end{algorithm}
 Since at each iteration of the above process one element is inserted into the set $\cY$ and at most $2\aq$ elements are removed from the set $\cI$, and since the $\cI$ initially contains at least $\epsilon 2^n$ elements, when the process terminates we have that
\begin{align}
a \eqdef \size{\cY} \ge \epsilon 2^n / 2 \aq
\end{align}
In addition, note that given the set $\cY$ and $\cX = \pi^{-1}_n (\cY)$, the set $\cZ$ can be described using $\log ((2^n - |\cY|)!)$ bits (by giving the order of the elements of $\zn \setminus \cY$, induced by applying $\pi_n$ on $\cX$). It follows that $\cY$, can be described by a string $\aux$ with
\begin{align}
\size{\aux} \leq 2 \log \binom{2^n}{a} + \log ((2^n - a)!)
\end{align}
We complete the proof by presenting the algorithm \Decoder that reconstructs
$\pi_n$ from the description of \Ac, $\h$, $\pi_{-n}$, $\cY$ and $\cZ$.
\begin{algorithm}[\Decoder]\label{alg:decoder}
\item [Input:] The description of \Ac, $\h$, $\pi_{-n}$, $\cY$ and $\cZ$.
\item[Operation:] For each $y \in \cY$ taken in lexicographical increasing order:
 \begin{enumerate}
 \item Emulate $\Ac^{\pi, \Sam^{\pi, \h}} (y)$ by answering \Ac's query as follows:
\begin{enumerate}

\item On a $\pi$-query $\q\in \zo^\ast$:
 \begin{itemize}
 \item If $\pi(\q)$ is defined by $\pi_{-n}$, the set $\cZ$ or the previously reconstructed values of $\pi_n$, answer with this value.
 \item Else, halt the emulation and set $\pi_{n}^{-1}(y)= \q$.
 \end{itemize}

 \item On a \Sam-query $\q = (w,\Cc, \Cc_\next) \in \cQ$:\footnote{We assume \wlg that \Ac's \Sam-queries are always in $\cQ$, since it can answer other \Sam-queries (\ie not in $\cQ$) by itself (by answering $\perp$).}
\begin{enumerate}
\item If $\Cc = \perp$ answer with $\h_\q(0^\Inp)$, where $\Inp$ is the input length of $\Cc_\next$.

\item Else answer with $\h_\q(v)$, where $v\in \zn$ is the minimal element such that $\Cc^\pi(\h_\q(v))$ can be evaluated (\ie all $\pi_n$-queries made by $\Cc$ are defined by $\pi_{-n}$, the set $\cZ$ or the previously reconstructed values of $\pi_n$) and its resulting value is $\Cc^\pi(w)$.
\end{enumerate}

 \end{enumerate}
 \item If the emulation reached its end and output $x$, set $\pi_{n}^{-1}(y)= x$.
 \end{enumerate}
\end{algorithm}
Assume that \Decoder has recontracted $\pi_{n}$ correctly for the first $k$ elements of $\cY$, we proved that it also does so for the $(k+1)$ element $y$ of $\cY$. To do that we show that on each query $\q$ asked by \Ac during the emulation done by \Decoder, either \Decoder halts on $\q$ and then $\q=\pi^{-1}_{n}(y)$, or \Decoder answers $\q$ correctly. (Hence, \Decoder sets the right value for $\pi^{-1}_{n}(y)$).

We first handle the case that $\q$ is a $\pi$-query. It is easy to verify that \Decoder answers correctly in case it does not halt. If halting, it must be the case that $\q\in \pi_n^{-1}(\cY)$ and $\pi_n (\q) \geq_\lex y$ (otherwise, $\pi_n (\q)$ would have been previously constructed). On the other hand, the definition of $\cY$ yields that $\pi_n (\q) \leq_\lex y$ (otherwise, $(\pi_n (\q),\q)$ would have added to $\cZ$), yielding that $\q=\pi^{-1}_{n}(y)$.

In case $\q$ is a \Sam-query $(w,\Cc, \Cc_\next)\in \cQ$, we assume \wlg that $\Cc \neq \perp$ (the case $\Cc = \perp$ is clear), and show that \Decoder returns $\h_\q(v)$ for the lexicographically smallest $v$ such that $\Cc^\pi(\h(v)) = \Cc^\pi(w)$ (hence, it answers correctly). Let $v_0$ be this minimal $v$. It is sufficient to show
that \Decoder has enough information to evaluate $\Cc^\pi(\h_\q(v_0))$. Indeed, since no $y$-hit happens in the computation of $\Ac^{\pi, \Sam^{\pi, \h}} (y)$, the evaluation of $\Cc^\pi(\h_\q(v_0))$ does \emph{not} query $\pi_n$ on $\pi^{-1}_{n}(y)$. Hence, the definition of $\cY$ guarantees that the answers to \emph{all} queries asked by $\Cc^\pi(\h_\q(v_0))$ are described in $\cZ$, or were previously reconstructed during the emulation of \Decoder.
\end{proof}

Now we are able to prove the following lemma, which (by holding for any \emph{fix} choice of $\pi_{-n}$ and $\h$) is a stronger form of \cref{lemma:reconstruction1}.

\begin{lemma}\label{lemma:reconstruction-stronger}
The following holds for all sufficiently large $n$, $\pi_{-n} = \set{\pi_i \in \Pi_i}_{i\in \N \setminus \set{n}}$, $\h\in \calH$ and $2^{n/5}$-\augQ circuit \Ac:
\begin{align*}
\prob{\pi_n \la \Pi_n}{\prob{y \la \zn}{\Ac^{\pi, \Sam^{\pi, \h}} (y) = \pi^{-1}_n (y)
\enspace \land \neg \Hit_{\Ac,\pi,\h}(y)} \geq 2^{-n/5}} \leq 2^{-2^{\frac{3n}5}}.
\end{align*}
\end{lemma}
\begin{proof}
\cref{claim:reconstruction2} implies that the fraction of permutations $\pi_n\in \Pi_n$ for
which
\begin{align}
\prob{y \la \zn}{\Ac^{\pi, \Sam^{\pi, \h}} (y) = \pi^{-1}_n (y)
\enspace \land \neg \Hit_{\Ac,\pi,\h}(y)} \ge 2^{-n/5}
\end{align}
is at most $\alpha = \frac{\binom{{2^n}}{a}^2 ({2^n} - a)!}{{2^n}!} = \frac{\binom{{2^n}}{a}}{a!}$, for $a \ge 2^{-n/5} \cdot {2^n} / (2 \cdot 2^{n/5}) = 2^{\frac{3n}5}/2$. Using the inequalities $a!
\ge (a/e)^a$ and $\binom{{2^n}}{a} \leq (2^n e/a)^a$, it holds that $\alpha \leq \left( \frac{{2^n} e^2}{a^2} \right)^a \leq \left( \frac{4 e^2}{2^{n/5}} \right)^a \leq 2^{-a}$ for sufficiently large $n$.
\end{proof}

\subsection{Avoiding \texorpdfstring{$y$}{y}-Hits by \texorpdfstring{\Sam}{Sam} -- Proving \texorpdfstring{\cref{lemma:Hitting}}{\cref{lemma:Hitting}}}\label{subsection:hitting}
Fix $n\in \N$ and a $\aq$-\augQ, $d$-depth normal-from circuit \Ac such that
\begin{align}\label{eq:AcSuccess}
\prob{\MyAtop{\pi \la \Pi, \h \la \calH}{y \la \zn}}{\Hit_{\Ac,\pi,\h}(y)} \geq \eps
\end{align}
for $\eps \in [0,1/\aq]$. We prove \cref{lemma:Hitting} by presenting a $2\aq$-\augQ, $d$-depth normal-from circuit \Mc, with
\begin{align}\label{eq:McSuccess}
\prob{\MyAtop{\pi \la \Pi, \h \la \calH}{y \la \zn}}{\Mc^{\pi, \Sam^{\pi, \h}} (y) = \pi^{-1}(y) \land \neg \Hit_{\Mc,\pi,\h}(y)} \ge (\eps/2)^{3 d + 1}
\end{align}

Let $\q = (w,\Cc,\cdot)$ be a \Sam-query asked in $\Ac^{\Sam^{\pi,h}}(y)$ that produces a $y$-hit (\ie $\Cc^\pi(\Sam^{\pi,h}(\q))$ queries $\pi$ on $\pi^{-1}(y)$). Since \Ac is in a normal form, \emph{previously} to asking $\q$ it made a \Sam-query $\q'= (\cdot,\cdot,\Cc)$, and answered by $w$. The main observation (see \cref{sec:singlePass,subsection:HittingExtension}) is that, with high probability, $\Cc^\pi(w)$ also queries $\pi$ on $\pi^{-1}(y)$. This suggests the following circuit for inverting random permutations with no hits.
\begin{algorithm}[\Mc]\label{alg:Mc}
\item[Input:] $y\in \zn$.

\item [Oracle:] $\pi \in \Pi$ and $\Sam = \Sam^{\pi, \h}$ for some $\h \in \calH$.

\item [Operation:]~
\begin{enumerate}
 \item Emulate $\Ac^{\pi,\Sam^{\pi, \h}}(y)$ while adding the following check each \Sam-query $(\cdot,\cdot, \Cc_\next)$ that \Ac makes that answered with $w$:

 If $\Cc_\next^\pi(w)$ queries $\pi$ on $x = \pi^{-1}(y)$, return $x$ and halt.

 \item Return $\perp$.
\end{enumerate}
\end{algorithm}
The rest of the proof is devoted for proving that \cref{eq:McSuccess} holds for the above definition of \Mc. The heart of the proof lies in the following lemma.
\begin{lemma}\label{lemma:HittingFixY}
The following holds for every $\pi\in \Pi$ and $y \in \zn$. Assume that
\begin{align}\label{equation:hittingassumption2}
\prob{\h \la \calH}{\Hit_{\Ac,\pi,\h}(y)} \ge \delta
\end{align}%
for $\delta \in [0,1/\aq]$, then
\begin{align*}
\prob{\h \la \calH}{\Mc^{\pi, \Sam^{\pi, \h}} (y) =
\pi^{-1}(y) \land \neg \Hit_{\Mc,\pi,\h}(y)} \ge \delta^{3 d}.
\end{align*}
\end{lemma}

We prove \cref{lemma:HittingFixY} in the next section, but first let us use it for concluding the proof \cref{lemma:Hitting}.

\begin{proof}[Proof of \cref{lemma:Hitting}]
Let $T =\{ (y, \pi) \in \zn\times \Pi \colon \Pr_{\h\la \calH}[\Ac^{\pi, \Sam^{\pi,
\h }} (y) = \pi^{-1} (y) \land \Hit_{\Ac,\pi,\h}(y)] \ge \eps/2\}$. The assumed success probability of \Ac (as stated in \cref{eq:AcSuccess}) together with a simple averaging argument, yield that
\begin{align}
\prob{y \la \zn, \pi \la \Pi}{(y, \pi) \in T} \ge \eps/2
\end{align}
Hence, by \cref{lemma:HittingFixY}
\begin{align*}
\prob{\h \la \calH}{\Mc^{\pi, \Sam^{\pi, \h}} (y) =
\pi^{-1}(y) \land \neg \Hit_{\Mc,\pi,\h}(y)} \ge (\eps/2)^{3 d}
\end{align*}
for every $(y, \pi) \in T$. We conclude that
\begin{align*}
\lefteqn{\prob{\MyAtop{\pi \la \Pi, \h \la \calH}{y \la \zn}}{\Mc^{\pi, \Sam^{\pi, \h}} (y) = \pi^{-1}(y) \land \neg \Hit_{\Mc,\pi,\h}(y)}}\\
& \ge \prob{\pi\la \Pi,y \la \zn}{(y, \pi) \in T} \cdot
\prob{\MyAtop{\pi \la \Pi, \h \la \calH}{y \la \zn}}{\Mc^{\pi, \Sam^{\pi, \h}} (y) = \pi^{-1}(y) \land \neg \Hit_{\Mc,\pi,\h}(y)\mid (y, \pi) \in T}\\
& \ge \eps/2 \cdot (\eps/2)^{3 d} = (\eps/2)^{3 d +1}.
\end{align*}
\end{proof}

\subsubsection{Proving \texorpdfstring{\cref{lemma:HittingFixY}}{\cref{lemma:HittingFixY}} --- The Single-Path Case}\label{sec:singlePass}
In this section we prove \cref{lemma:HittingFixY} for a simplified case that captures the main difficulties of the proof. The extension for the general case is given in \cref{subsection:HittingExtension}. In this simplified case \Ac queries \Sam on exactly $d$ queries that lie \emph{along a single path} --- \Ac queries \Sam with $\q_1,\ldots, \q_d$ satisfying $\parent(\q_i) = \q_{i-1}$ for every $2 \leq i \leq d$ (\ie $\q_i = (w_i,\Cc_i,\Cc_{\next,i})$ implies that $w_i$ is \Sam's answer on $\q_{i-1} = (\cdot,\cdot,\Cc_i)$).

In the following we fix $y\in \zn$, $\pi\in \Pi$ and $\delta \in [0,1/\aq]$ such that
\begin{align}\label{equation:delta2}
\prob{\h \la \calH}{\Hit \eqdef \Hit_{\Ac,\pi,\h}(y)} \ge \delta.
\end{align}
We let $\hit(\Cc,w)$, for circuit $\Cc$ and string $w$, be the event that $\Cc(w)$ queries $\pi$ on $\pi^{-1} (y)$ (hereafter, $\Cc(x)$ stands for $\Cc^\pi(x)$), and use the following random variables.
\begin{definition}
The following random variables are defined \wrt a random execution of $\Ac^{\pi, \Sam^{\pi,\H}}(y)$, where $\H$ is uniformly drawn from $\calH$.\footnote{Since $\Ac$ is a circuit, and hence deterministic, these random variables are functions of $\H$.}
\begin{itemize}

\item $\Q_1=(W_1=\perp,\Cc_1= \perp, \Cc_2),\Q_2=(W_2,\Cc_2,\Cc_3),\ldots, \Q_d =(W_d,\Cc_d, \cdot)$ denote \Ac's queries to \Sam.\footnote{Our simplifying assumption yields that \Ac's queries are indeed of the above structure, and that $W_{i+1}$ is \Sam's answer on $\Q_i$ for every $i\in [d-1]$. In particular, it holds that $\Q_1,\dots,\Q_i$ are determined by $W_1,\dots,W_i$ and $\Cc_1,\dots,\Cc_i$.}

\item $\Hit_i$, for $i\in [d]$, is the event $\hit(\Cc_i,W_{i+1})$, letting $\hit(\perp,\cdot)= \emptyset$, and let $\Hit_{\leq i} \eqdef \bigcup_{j\in [i]} \Hit_j$. (Note that $\Hit = \Hit_{\leq d}$.)

\item $D_i$, for $i\in [d]$, is the distribution of \Sam's answer on the $i$'th query $\Q_i$ --- $D_1$ is the uniform distribution over $\zo^\Inp$, and for $2 \leq i \leq d$, $D_i$ is the uniform distribution over the set $\Cc_i^{-1} (\Cc_i(W_i))$.

 \item $\alpha_0 = \alpha_1 = 0$, and for $2 \leq i \leq d$ let $\alpha_i=\prob{w \la D_i}{\hit(\Cc_i,w)}$.

\item $\Jump_i$, for $i\in [d]$, is the event $\alpha_i > \max \set{ \frac{8} {\delta^2} \cdot \alpha_{i-1}, (\frac{\delta^2}{8})^{d + 1} }$. $\Jump_{\leq i} \eqdef \bigcup_{j\in [i]} \Jump_j$ and $\Jump \eqdef \Jump_{\leq d}$.
\end{itemize}
\end{definition}

It is instructive to view the interaction between \Ac and \Sam
as a $d$ round game, where in the $i$'th round \Ac chooses a query $\Q_i$, and the oracle \Sam samples $W_{i+1}$ from the distribution $D_i$. The goal of the circuit \Ac in this game is to cause $\hit(\Cc_i,W_{i+1})$ (\ie causing the event $\Hit$ to happen).

By \cref{equation:delta2}, \Ac produces a $y$-hit (\ie causing the event $\Hit$) with probability at least $\delta$, and therefore wins the game with at least this probability. Our first observation is that the latter induces that the event $\Jump$ occurs with probability at least $\delta/2$. Intuitively, in case $\Jump$ does not occur, then the $\alpha_i$'s are too small in order to produce a $y$-hit with noticeable probability.

\begin{claim}\label{claim:Jump}
$\Pr[\Jump] > \delta/2$.
\end{claim}
The proof of \cref{claim:Jump} immediately follows from the next observation.
\begin{claim}\label{claim:HitImpliesJump}
$\pr{\exists i\in [d] \colon \Hit_{\leq i}\land \neg \Jump_{\leq i}} \le \delta^5/512 $.
\end{claim}
Namely, we expect no hit unless a jump has previously occurred.
\begin{proof}[Proof of \cref{claim:HitImpliesJump}]
We prove that
\begin{align}\label{eq:HitImpliesJump}
\Pr[\Hit_i \mid \neg \Jump_{\leq i} ] \le \frac1d \cdot \frac {\delta^5}{512}
\end{align}
for every $2 \leq i\le d$, and the proof of the \cref{claim:HitImpliesJump} follows by union bound.

Assuming $ \set{\neg \Jump_{\leq i}}$, we first show that
\begin{align}\label{eq:HitImpliesJump2}
\alpha_j \leq \left( \frac{\delta^2}{8} \right)^{d - j + 3}
\end{align}
for every $2\leq j \le i$. For $j=2$ compute
\begin{align}
\alpha_2 &\leq \max \set{ \frac{8} {\delta^2} \cdot \alpha_{1}, \left(\frac{\delta^2}{8}\right)^{d + 1}}= \max \set{\frac{8} {\delta^2} \cdot 0, \left(\frac{\delta^2}{8}\right)^{d + 1}} = \left(\frac{\delta^2}{8}\right)^{d - 2 + 3},
\end{align}
where the inequality holds since we assume $\neg \Jump_{\leq i} $. Assuming \cref{eq:HitImpliesJump2} holds for $2\le j-1 \leq i-1$, compute
\begin{align*}
\alpha_j &\leq \max \set{ \frac{8} {\delta^2} \cdot \alpha_{j-1}, \left(\frac{\delta^2}{8}\right)^{d + 1}}\\
&\leq \max \set{\frac{8} {\delta^2} \cdot \left(\frac{\delta^2}{8}\right)^{d - j + 4}, \left(\frac{\delta^2}{8}\right)^{d + 1}}\\
&= \left(\frac{\delta^2}{8}\right)^{d - j + 3}
\end{align*}
proving \cref{eq:HitImpliesJump2}. The first inequality holds since we assume $\neg \Jump_{\leq i} $, and the second one by the induction hypothesis. It follows that
\begin{align*}
\Pr[\Hit_i \mid \neg \Jump_{\leq i} ] \leq \left(\frac{\delta^2}{8}\right)^3\leq \frac1d \cdot \frac {\delta^5}{512 },
\end{align*}
where the last inequality holds since (by the statement of the lemma) $\delta \leq 1/\aq \leq 1/d$.
\end{proof}
Given the above, the proof of \cref{claim:Jump} is immediate.
\begin{proof}[Proof of \cref{claim:Jump}]
Compute
\begin{align*}
\Pr[\Jump] &\geq \Pr[\Hit] - \Pr[\Hit \land \neg \Jump]\\
& \geq \delta - \delta^5/512 > \delta/2,
\end{align*}
where the second inequality follows by \cref{claim:HitImpliesJump}.
\end{proof}

Consider \Mc's point of view in the aforementioned game. Recall that following each query $Q_i =
(W_i,\Cc_i, \Cc_{i+1})$, the circuit \Mc evaluates $\Cc_{i+1} (W_{i+1})$, and if a $\pi$-gate in this computation has input $\pi^{-1} (y)$, then \Mc outputs $\pi^{-1} (y)$ and halts. Algorithm \Mc ``wins" the game (\ie inverts $\pi$ with no hit), if it manages to retrieve
$\pi^{-1} (y)$ \emph{before} \Ac produces any $y$-hits. Let $\beta_i$ be the probability that \Mc outputs
$\pi^{-1} (y)$ and halts after query $\q_i$.
\begin{definition}
For $i\in [d]$ let $\beta_i = \prob{w \la D_i}{\hit(\Cc_{i+1},w)}$.
\end{definition}
The game between \Ac and \Mc can be now described as follows: in the $i$'th round, \Ac chooses a query $\q_i$, which
determines $\beta_i$, and \Sam samples $w_{i+1}$, which determines $\alpha_{i + 1}$. If $\q_i$ implies ``high" $\beta_i$, then \Mc has high probability in winning the game (\ie with high probability the computation $\Cc_{i+1} (w_{i+1})$ done by \Mc finds $\pi^{-1} (y)$. Therefore to win the game, \Ac should not choose high $\beta_i$. We claim, however, that if
$\beta_i$ is low, then with high probability $\alpha_{i + 1}$ will be low as well. But if
$\alpha_{i+1}$ is low, then \Ac has a low probability of producing a $y$-hit in the next query
$\q_{i+1}$. This means that in order for \Ac to win the game, at some point it must ``take a risk''
and produce high $\beta_i$.

The following claim states that conditioned on $\Q_1, \ldots, \Q_i$, the expectation of
$\alpha_{i+1}$ is $\beta_i$. Therefore, if $\beta_i$ is
low then $\alpha_{i+1}$ is low with high probability. Note that under the above conditioning (which determines the value of $W_1\cdots,W_i$), the value of $\beta_i$ is determined, while $\alpha_{i+1}$ is still a random variable, to be determined by the value of $W_{i+1} = \Sam(\Q_i)$.

\begin{claim}\label{claim:expectation}
$\Ex[\alpha_{i+1} \mid W_1, \ldots, W_i] = \beta_i$ for every $i\in [d-1]$.
\end{claim}

\begin{proof}
Fix $i\in [d-1]$ and a fixing $w_1,\dots,w_i$ for $W_1, \ldots, W_i$ (which implies a fixing $\q_1=(w_1,\Cc_1,\Cc_2), \ldots, \q_i= (w_i,\Cc_i,\Cc_{i+1})$ for $\Q_1, \ldots, \Q_i$). We write
\begin{align}
\Ex[\alpha_{i+1}] &= \sum_{z\in \zo^\Out} \prob{w\la D_i}{\Cc_{i+1}(w) = z} \cdot \prob{w \la \Cc_{i + 1}^{-1}(z)}{\hit(\Cc_{i + 1}, w)}\\
&= \sum_{z} \frac{\size{\Cc^{-1}_{i+1}(z)}}{\size{\cS}} \cdot \frac{\size{ \set{ w \in
\Cc^{-1}_{i + 1} (z) \colon \hit(\Cc_{i + 1}, w)}}}{\size{\Cc^{-1}_{i + 1} (z) }}\nonumber\\
&= \sum_{z} \frac{\size{\set{ w \in \Cc^{-1}_{i + 1} (z) \colon\hit(\Cc_{i + 1}, w)} }}{\size{\cS}},\nonumber
\end{align}
where $\cS \eqdef \Cc^{-1}_i (\Cc_i(w_i))$ and $\Out$ is the output length of $\Cc_{i+1}$. Note that while $\Q_{i+1}$ is not determined by $\q_1, \ldots, \q_i$, the circuit $\Cc_{i + 1}$ is. In addition, since \Ac is in a normal form,\footnote{This is the only place throughout the whole proof, where the normal-from assumption that \Ac is in a normal form is being used.} the circuit $\Cc_{i +1}$ is an extension of $\Cc_i$. Thus
\begin{align*}
\Ex[\alpha_{i + 1}] & = \sum_{z\in \zo^\Out} \frac{\size{\set{ w \in \Cc^{-1}_{i + 1} (z) \colon\hit(\Cc_{i + 1}, w)} }}{\size{\cS}}\\
& = \frac{\size{\set{ w \in \cS\colon \hit(\Cc_{i+1}, w)}}}{\size{\cS}} \\
& = \prob{w \la \cS}{ \hit(\Cc_{i+1}, w)} \\
& = \prob{w \la D_i}{ \hit(\Cc_{i+1}, w)} \\
& = \beta_i.
\end{align*}%
\end{proof}

Up to this point, we have reached the conclusion that in order for \Ac to win the game, it must be
that at least one of the $\alpha_{i + 1}$'s is high, implying that $\Jump$ occurs. We
have also seen that the latter requires \Ac to choose a query $\q_i$ that determines a high
$\beta_i$. We claim that in this case, it holds that $\beta_i$ is \emph{significantly
larger} than $\alpha_i$. Formally,

\begin{definition}
Let $\Gap_i$, for $i\in [d]$, be the event $\set{\beta_i > \max \setb{ 2 \alpha_i, \left(\frac{\delta^2}{8}\right)^{d+2}}}$.
\end{definition}
The following claim states that with a noticeable probability, there exists an index $i$ such that $\Gap_i$ occurs and $\Jump_{\leq i}$ does not occur. In other words, $\beta_i$ is significantly larger than $\alpha_j$ for all $j\leq i$. We later show that this $\beta_i$ enables \Mc to retrieve $\pi^{-1} (y)$ before \Ac produces any $y$-hits.

\begin{claim}\label{claim:GapFirst}
Let $\GapFirst$ be the event $\set{\exists i\in [d] \colon \Gap_i \land \neg \Jump_{\leq i}}$, then $\pr{\GapFirst} \ge \delta/4$.
\end{claim}

For proving \cref{claim:GapFirst} we use the following claim, showing that unless $\beta_i$ is significantly larger than $\alpha_i$, then $\alpha_{i + 1}$ is not significantly larger than $\alpha_i$.
\begin{claim}\label{claim:jump-given-notgap}
$\Pr[\Jump_{i + 1} \mid \neg \Gap_i] \leq \delta^2/4$ for every $i \in [d-1]$.
\end{claim}

\begin{proof}
We write
\begin{align}
\pr{\Jump_{i + 1} \mid \neg \Gap_i} \leq \pr{\alpha_{i + 1} > \beta_i \cdot \delta^2/4 } + \pr{\Jump_{i + 1} \mid \neg \Gap_i \land \set{\alpha_{i + 1} \leq \beta_i \cdot \delta^2/4} }
\end{align}%
\cref{claim:expectation} and Markov's inequality imply that
\begin{align}
\pr{\alpha_{i + 1} > \frac{4}{\delta^2}\cdot \beta_i} \leq \delta^2/4
\end{align}

Since the event $\set{\Jump_{i + 1} \cap \neg \Gap_i \cap \set{\alpha_{i + 1} \leq \frac{4}{\delta^2}\cdot \beta_i}}$ is empty, we conclude that
\begin{align*}
\pr{\Jump_{i + 1} \mid \neg \Gap_i} \leq \delta^2/4.
\end{align*}
\end{proof}
Given \cref{claim:jump-given-notgap}, we prove \cref{claim:GapFirst} as follows.
\begin{proof}[Proof of \cref{claim:GapFirst}]
Compute
\begin{align}
\pr{\Jump \land \neg \GapFirst}&\le \sum_{i\in [d]} \pr{\Jump_i \land \neg \Jump_{\leq i-1} \land \neg \Gap_i}\\
&\leq \sum_{i\in [d]} \pr{\Jump_i \mid \neg \Gap_i}\nonumber\\
&\leq d\cdot \delta^2/4\nonumber\\
&\leq \delta/4,\nonumber
\end{align}%
where the before to last inequality holds by \cref{claim:jump-given-notgap}, and last inequality holds since $\delta \leq 1/\aq \leq 1/d$.

Since (by \cref{claim:Jump}) $\pr{\Jump} \geq \delta/2$, it follows that $\pr{\GapFirst} \geq \delta/4$.
\end{proof}

\paragraph{Putting it together}
Given the above observation, we are ready to prove \cref{lemma:HittingFixY}.
\begin{proof}[Proof of \cref{lemma:HittingFixY} (Single-path case)]
Let $I$ be the smallest index in $[d-1]$ for which $\Gap_{i}$ occurs, letting $I=\perp$ in case no such event happens. Note that whenever $\GapFirst$ happens, then $I\neq \perp$ and $\Jump_i$ does not occur for all $i\in [I]$. Compute
\begin{align}
\Pr[\Hit_{\leq I-1} \mid \GapFirst] &= \Pr[\Hit_{\leq I-1}\land \neg \Jump_{\leq I-1}\mid \GapFirst]\\
&\leq \prb{\exists i\in [d] \colon \Hit_{\leq i}\land \neg \Jump_{\leq i}} \cdot \frac1 {\Pr[\GapFirst]}\nonumber\\
&\le\frac{\delta^5}{512 } \cdot \frac4\delta < 1/2,\nonumber
\end{align}
where the second inequality holds by \cref{claim:GapFirst,claim:HitImpliesJump}. In addition, for the event $E = \set{\hit(\Cc_{I+1},W_{I+1}) \cap \neg \hit(\Cc_{I},W_{I+1})}$ it holds that
\begin{align}
\lefteqn{\Pr[E\mid \GapFirst, \neg \Hit_{\leq I-1}]}\\
&\geq \pr{\hit(\Cc_{I+1},W_{I+1})\mid \GapFirst, \neg \Hit_{\leq I-1}} - \pr{ \hit(\Cc_{I},W_{I+1})\mid \GapFirst, \neg \Hit_{\leq I-1}}\nonumber\\
&= \exx{\beta_I\mid \GapFirst, \neg \Hit_{\leq I-1}} - \exx{\alpha_I\mid \GapFirst, \neg \Hit_{\leq I-1}}\nonumber\\
&= \exx{\beta_I - \alpha_I\mid \GapFirst, \neg \Hit_{\leq I-1}}\nonumber\\
 & \geq \frac12 \cdot \left(\frac{\delta^2}{8}\right)^{d+2},\nonumber
\end{align}
where the second inequality holds by the definition of $\Gap$. We conclude that
\begin{align*}
\lefteqn{\pr{\Mc^{\pi, \Sam^{\pi, \H}} (y) = \pi^{-1}(y) \land \neg \Hit}} \\
&= \pr{\exists i \colon \hit(\Cc_{i+1},W_{i+1}) \land \neg \Hit_{\leq i}}\\
&\geq \pr{\GapFirst} \cdot \pr{\hit(\Cc_{I+1},W_{I+1}) \land \neg \Hit_{\leq I} \mid \GapFirst}\\
&= \pr{\GapFirst} \cdot \pr{E \land \neg \Hit_{\leq I-1} \mid \GapFirst}\\
&\ge \pr{\GapFirst}\cdot \Pr[\neg \Hit_{\leq I-1} \mid \GapFirst] \cdot \Pr[E \mid \GapFirst, \neg\Hit_{\leq I-1}]\\
&\ge \frac\delta4\cdot \frac12 \cdot \frac12 \cdot \left(\frac{\delta^2}{8}\right)^{d+2}\\
&= \left(\frac{\delta^2}{8}\right)^{d+3}\\
&\geq \delta^{3d}.
\end{align*}%
\end{proof}

%

\subsubsection{Proving \texorpdfstring{\cref{lemma:HittingFixY}}{\cref{lemma:HittingFixY}} --- The General Case}\label{subsection:HittingExtension}
This section extends the proof given in \cref{lemma:HittingFixY} to the general case where \Ac's
queries are not necessarily along a single path. The extension below is mainly technical, and requires no more than
refining some events and notations.


Assuming that $\Q_1, \ldots, \Q_\aq$ are the \Sam-queries asked by \Ac, let $\parent(i)$, for $i\in [s]$, be the index of the query $\parent(\Q_i)$ (\ie the index of $\Q_i$'s parent, see \cref{def:query forest}) in the above query list, letting $\parent(i) = 0$ in case $\parent(\Q_j) = \perp$. Note that unlike the single-path case studies in \cref{sec:singlePass}, the values of $\parent(1),\dots,\parent(\aq)$ are not predetermined (in particular $\parent(i)$ might not be $i-1$). This difference reflects the fact that \Ac might \emph{repetitively} ask the same query, each time dictating \Sam to use \emph{fresh randomness} by slightly modifying the value of the parameter $\Cc_\next$, until the answer serves it best: until there is a big jump in the value of $\alpha$. It turns out that while repeating a query does increase the probability of \Ac to ``win" the game against \Mc (\ie to make a $y$-hit before \Mc inverts $y$), since the expected value $\alpha_i$ is the value of $\alpha_{\parent(i)}$, such repetition does not increase the \Ac's winning probability by too much.

We now describe in detail the required technical changes to the proof of \cref{lemma:HittingFixY}. Fix $y\in \zn$, $\pi\in \Pi$ and $\delta \in [0,1/\aq]$ such that $\prob{\h \la \calH}{\Hit \eqdef \Hit_{\Ac,\pi,\h}(y)} \ge \delta$. The definitions of the following random variables are natural generalization of those given in \cref{sec:singlePass}. Recall that $\hit(\Cc,w)$ is the event that $\Cc(w)$ queries $\pi$ on $\pi^{-1} (y)$. The index $i$ in the following definitions takes values in $[\aq]$.
\begin{definition}
The following random variables are defined \wrt a random execution of $\Ac^{\pi, \Sam^{\pi,\H}}(y)$, where $\H$ is uniformly drawn from $\calH$.

\begin{itemize}
\item $\Q_1 = (W_1,\Cc_1,\Cc_{\next,1}, \ldots, \Q_s = (W_\aq,\Cc_\aq,\Cc_{\next,\aq})$, denote \Ac's queries to \Sam, and $W^\ans_1,\dots,W^\ans_\aq$ denote their answers.\footnote{Since $\Ac$ is in a normal form, for every $i\in [\aq]$ it holds that $W_i = W^\ans_{\parent(i)}$ and $\Cc_i = \Cc_{\next,\parent(i)}$, letting $W_0= \Cc_{\next,0}= \perp$. Note that $\Q_1,\dots,\Q_i$ are determined by $W^\ans_1,\dots,W^\ans_{i-1}$ and the circuits $\Cc_1,\Cc_{\next,1}, \ldots, \Cc_i,\Cc_{\next,i}$.}

 \item $\Hit_i$ is the event $\hit(\Cc_i,W^\ans_{i})$, letting $\hit(\perp,\cdot)= \emptyset$, and $\Hit_{\leq i} \eqdef \bigcup_{j\in [i]} \Hit_j$.

\item $D_i$ is the uniform distribution over $\Cc_i^{-1} (\Cc_i(W_i))$ in case $W_i \neq \perp$, and the uniform distribution over $\zo^\Inp$ otherwise.

 \item $\alpha_i = \prob{w \la D_i}{\hit(\Cc_i,w)}$ in case $W_i \neq \perp$, and $\alpha_i = 0$ otherwise.

 \item $\beta_i = \prob{w \la D_i}{\hit(\Cc_{\next,i},w)}$.

 \item $\Gap_i$ is the event $\beta_i > \max \set{ 2 \alpha_i, \left(\frac{\delta^2}{8}\right)^{d+2}}$

\end{itemize}
\end{definition}

In addition, we make use of the following definition.

\begin{definition}~
\begin{itemize}
\item $D^\dec_i$ is the uniform distribution over $\Cc_{\next,i}^{-1} (\Cc_{\next,i}(W^\ans_i))$, and $\alpha^\dec_i = \prob{w \la D^\ans_i}{\hit(\Cc_{\next,i},w)}$.\footnote{Note that for any $j$ with $p(j)=i$, if such exists, it holds that $D^\dec_i = D_j$ and $\alpha^\dec_i = \alpha_j$.}

\item $\PJump_i$ (for ``Potential $\Jump$'') is the event $\alpha^\ans_i > \max \set{ \frac{8} {\delta^2} \cdot \alpha_i}$.
\end{itemize}
\end{definition}

The following claims are analogues to the claims given in \cref{sec:singlePass}.

\begin{claim}\label{claim:GJump}
$\Pr[\PJump] > \delta/2$.
\end{claim}
\begin{proof}
Same as the proof of \cref{claim:Jump}, replacing \cref{claim:HitImpliesJump} with \cref{claim:GHitImpliesJump}.
\end{proof}
\begin{claim}\label{claim:GHitImpliesJump}
$\pr{\exists i\in [\aq] \colon \Hit_{\leq i}\land \neg \PJump_{\leq i-1}} \le \delta^5/512 $.
\end{claim}
\newcommand{\depth}{{\rm depth}}
\begin{proof}
Same as the proof of \cref{claim:Jump}, replacing \cref{eq:HitImpliesJump2} with
\begin{align}\label{eq:GHitImpliesJump2}
\alpha_{j} \leq \left(\frac{\delta^2}{8}\right)^{d - \depth(j) + 3},
\end{align}
where $\depth(j)= 0$ for $j=0$, and $\depth(\parent(j))+1$ otherwise.
\end{proof}

\begin{claim}\label{claim:Gexpectation}
$\Ex[\alpha^\ans_i \mid W^\ans_1,\dots,W^\ans_{i-1}] = \beta_i$ for every $i\in [\aq]$.
\end{claim}
\begin{proof}
Same as the proof of \cref{claim:expectation}, replacing $\Cc_{i+1}$ with $\Cc_{\next,i}$.
\end{proof}

\begin{claim}\label{claim:GGapFirst}
Let $\GapFirst$ the event $\set{\exists i\in [\aq] \colon \Gap_i \land \neg \PJump_{\leq i-1}}$, then $\pr{\GapFirst} \ge \delta/4$.
\end{claim}
\begin{proof}
Same as the proof of \cref{claim:GapFirst}, replacing \cref{claim:jump-given-notgap} with \cref{claim:Gjump-given-notgap}, and recalling that $\delta < 1/\aq$.
\end{proof}

\begin{claim}\label{claim:Gjump-given-notgap}
$\Pr[\PJump_{i} \mid \neg \Gap_i] \leq \delta^2/4$ for every $i \in [\aq]$.
\end{claim}
\begin{proof}
Same as the proof of \cref{claim:jump-given-notgap}, replacing $\alpha_{i + 1}$ with $\alpha^\ans_i$, and $\Jump_{i + 1}$ with $\PJump_i$.
\end{proof}

As in \cref{sec:singlePass}, the proof of \cref{lemma:HittingFixY} (here for the general case) easily follow the above claims.
\begin{proof}[Proof of \cref{lemma:HittingFixY} (General case)]
Same lines as the proof given in \cref{sec:singlePass}, replacing \cref{claim:GapFirst,claim:HitImpliesJump} with \cref{claim:GGapFirst,claim:GHitImpliesJump}.
\end{proof}

\section{Lower Bounds on Statistically Hiding commitments}\label{section:LowerBounds}

In this section we combine the results presented in \cref{section:PowerOfSam,section:inverting} to derive our lower bounds on black-box constructions of statistically-hiding commitments from trapdoor permutations. Throughout the section, we assume for ease of notation that the integer functions $d$, $c$ and $\s$, measuring the round and sender communication complexity of the considered commitment scheme, and the hardness of the considered trapdoor permutations family, respectively, are non-decreasing.

\subsection{The Round Complexity Lower Bound}\label{section:RoundCompLB}
In this section we give lower bound on the \emph{round complexity} of black-box constructions of statistically hiding commitment from trapdoor permutations. We first give two results for the case where the reduction is to polynomially hard families . The first result is for ``security-preserving" constructions, and the second one is for arbitrary ones.

\begin{theorem}[restating \cref{thm:Intro:RoundComplexityLB}]
Any $O(n)$-security-parameter-expanding, fully black-box construction of a weakly binding and
honest-receiver statistically hiding commitment scheme from a polynomially hard family of trapdoor permutations has
$\Omega \left( n / \log n \right)$ communication rounds.
\end{theorem}

\begin{theorem}
Any fully black-box construction of a weakly binding and honest-receiver statistically hiding
commitment scheme from a polynomially hard family of trapdoor permutations has $n^{\Omega(1)}$ communication rounds.
\end{theorem}

The above two theorems are in fact corollaries of the more general statement given below, stated for trapdoor permutations of arbitrary hardness.
\begin{theorem}\label{theorem:SHCRoundCompLB}
For every $\ell$-security-parameter-expanding, fully black-box construction of a $d$-round
weakly binding and honest-receiver statistically hiding commitment scheme from an $\s \geq n^{\omega(1)}$-hard family of trapdoor permutations, it holds that $d(\ell(n)) \in \Omega \left( n / \log \s(n) \right)$.
\end{theorem}
\begin{proof}
Let $(\Com= (\Sc, \Rc, \Vc))$ be an $\ell$-security-parameter-expanding fully black-box construction
of a $d$-round, $\delta $-binding and honest-receiver, honest-sender, statistically hiding commitment scheme from an $\s$-hard family of trapdoor permutations, where $\delta(n) = 1 - 1/p(n)$ for some $p\in \poly$, and let $m = m(n)$ be a bound on the running time of $\Rc$ on security parameter $n$. \cref{theorem:collisionInLowRound} yields that  relative to most fixing of $(\tau,\Sam^{\tau,\h})$, there exists an efficient breaker for the binding of \Com.
\begin{claim}\label{corollary:BreakingLowRCSHC}
There exists a $(d +1)$-depth, deterministic $\poly$-\augQ, normal-from oracle-aided algorithm $\Ss$ such that the following holds for every $\tau\in \Tau$: for $h\in \calH$ and $r_\Rc \in \zm$ let $\TwoDecom^\h_n(r_\Rc)$ be the event that  $\Vc(\com, \decom) \neq \Vc (\com, \decom') \in \zo$ for $((\decom, \decom'),\com) = \exec{(\Ss^{\tau,\Sam^{\tau,\h}}, \Rc^\tau(r_\Rc))(1^n)}_{\out^\Ss,\out^\Rc}$, and let $\NoBreak^{\tau, \h}_n$ be the event that  $\ppr{r_\Rc \la \zm}{\neg \TwoDecom^{\h}_n(r_\Rc)} \geq \delta(n)$. Then $\ppr{\h\la \calH}{\NoBreak^\h_n} \in O(1/n^2)$.
\end{claim}
We defer the proof of \cref{corollary:BreakingLowRCSHC} of to \cref{sec:BreakingLowRCSHC}, and first use it for proving \cref{theorem:SHCRoundCompLB}. \cref{corollary:BreakingLowRCSHC} yields that
\begin{align}
\sum_{n = 1}^{\infty} \prob{\h\la \calH}{\NoBreak^{\tau, \h}_n} < \infty
\end{align}
for any $\tau\in \Tau$, where $\Ss$ and $\NoBreak^{\h}_n$ are as in the claim statement. By the Borel-Cantelli lemma, the probability over the choices of $\h\la\calH$ that $\NoBreak^{\h}_n$ occurs for infinitely many $n$'s, is \emph{zero}. It follows that with probability one over the choice of $(\tau,\h)\la \Tau \times \calH$, it holds that $\Ss^{\tau,\Sam^{\tau,\h}}$ breaks the weak binding of $\Com$. Hence, with probability one over the choice of $(\tau =(G,F,F^{-1}),\h)$, it holds that
\begin{align}\label{equation:EisBounded}
\prob{td \la \zn,y\la \zn}{\Ac^{\tau, \Sc^{\tau,\Sam^{\tau,\h}}}(1^n, G (td), y) = F^{-1} (td, y)}> \frac{1}{\s(n)}
\end{align}%
for infinitely many $n$'s. Since \cref{equation:EisBounded} holds \wrt measure one of the oracles $(\tau,\h)$, we have that
\begin{align}\label{equation:AbreakTau}
\prob{\MyAtop{\tau \la \Tau,\h\la \calH}{td \la \zn,y\la \zn}}{\Ac^{\tau, \Ss^{\tau,\Sam^{\tau,\h}}}(1^n, G (td), y) = F^{-1} (td, y)}> \frac{1}{\s(n)}
\end{align}
for infinitely many $n$'s.

By \cref{prop:InseritNormalForm}, the circuit $\Ac^\Ss$ (\ie the circuit that  given oracle access to $\tau$ and $\Sam^{\tau,\h}$, acts as $\Ac^{\tau, \Ss^{\tau,\Sam^{\tau,\h}}}$) is in a normal form and of depth $d'(n) = d(\ell(n))+1$. Hence, \cref{equation:AbreakTau} yields the existence of an $\aq = 4\s$-\augQ, normal form, $d'$-depth, oracle-aided circuit family $\As = \set{\As_n}_{n\in \N}$ with
\begin{align}
\prob{\MyAtop{\tau \la \Tau,\h\la \calH}{td \la \zn,y\la \zn}}{\As_n^{\tau, \Sam^{\tau,\h}}(G (td), y) = F^{-1} (td, y)}
> \frac{4}{\aq(n)}
\end{align}
for infinitely many $n$'s.

\cref{theorem:TDPHardForSam} yields that $2^{n/8} \leq (4\aq(n))^{3 d(\ell(n)) +1} \leq (4\s(n))^{6 d(\ell(n)) +2}$, implying that $d(\ell(n)) \in \Omega(n/\log \s(n))$.
\end{proof}

\subsubsection{Proving \texorpdfstring{\cref{corollary:BreakingLowRCSHC}}{\cref{corollary:BreakingLowRCSHC}}}\label{sec:BreakingLowRCSHC}
\begin{proof}[Proof \cref{corollary:BreakingLowRCSHC}]
Let \As be the deterministic, polynomial-\augQ algorithm guaranteed by \cref{theorem:collisionInLowRound} for the protocol $(\Sc, \Rc)$. Recall the following holds for every $\h \in \calH$ and $k\in \N$: following the execution of $(\As^{\Sam^{\tau, \h}}(1^k), \Rc^\tau)(1^n))$ that yields a transcript $\trans$, algorithm $\Ac$ outputs a set $\set{(b_i,r_i))}_{i\in [k]}$ such that the $k$ pairs are independent uniform values for the input and random coins of $\Sc$, consistent with $\trans$. Also recall that over a uniform choice of $\h \la \calH$, the value of $\trans$ has the same distribution has the one induced by $\exec{\Sc^\tau, \Rc^\tau)(1^n)}$.

Algorithm $\Sc$ with oracle access to $\tau$ and $\Sam^{\tau,\h}$, acts through the interaction with $\Rc$ as $\Ac^{\tau,\Sam^{\tau,\h}}$ would on input $(1^n,1^n)$ (\ie we set $k=n$). If in the set output by \Ac there exist two pairs $(0,r_0)$ and $(1,r_1)$, \Ss uses them to generate two decommitments $\decom_0$ and $\decom_1$. Note that, if such pairs were found, then it holds that $\Vc(\com,\decom_0)=0$ and $\Vc(\com,\decom_1)=1$, where $\com$ is the commitment output by $\Rc$ when interacting with $\Ss$. In the following we prove that \Ss finds such a good couple of pairs with save but negligible probability over the choice of $\h\in \calH$.

We next define a set of ``good'' transcripts that enable $\Ss$ to reveal to both $0$ and $1$ with overwhelming
probability. For $n\in \N$ and $b\in \zo$, let $\Trans_n^b = \exec{\Sc^\tau(b), \Rc^\tau)(1^n)}_\trans$ (\ie the random variable induced by the transcript of a random execution of $(\Sc^\tau, \Rc^\tau)$, where $\Sc$'s input bit is $b$), let $\Trans = \Trans^u_n$, for $u\la \zo$, and let $\Balanced_n = \set{\trans\in \Supp(\Trans_n)\colon \frac12 \leq \frac{\Pr_{\Trans^0_n}[\trans]}{\Pr_{\Trans^1_n}[\trans]}\leq \frac32} $. Since \Com is statistically hiding (at least, against the honest receiver), it follows that
\begin{align}\label{eq:BreakingLowRCSHC:1}
\ppr{\Trans_n}{\neg \Balanced_n} = \negl(n)
\end{align}

For $\h \in\calH$ and $r\in \zm$, let $\trans_n^{\h,r_\Rc} = \exec{\Ss^{\tau,\Sam^{\tau,\h}}, \Rc^\tau(r_\Rc))(1^n)}_\trans$. \cref{theorem:collisionInLowRound} yields that $\trans_n^{\h,r_\Rc}$, for uniformly chosen values of $\h$ and $r_\Rc$, and $\Trans_n$, are identically distributed, and that
\begin{align}\label{eq:BreakingLowRCSHC:2}
\ex{\h\la \calH,r_\Rc \la \zm}{\neg \TwoDecom^{\h}_n(r_\Rc)\mid \trans_n^{\h,r_\Rc} \in \Balanced} \leq \left(\frac{2}{3}\right)^{n-1}
\end{align}

We conclude that
\begin{align*}
\lefteqn{\ppr{\h\la \calH}{\NoBreak^\h_n}}\\
& = \ppr{h\la \calH}{\ppr{r_\Rc \la \zm}{\neg \TwoDecom^{\h}_n(r_\Rc)} \geq \delta(n)}\\
&\leq \frac{\ex{\h\la \calH,r_\Rc \la \zm}{\neg\TwoDecom^{\h}_n(r_\Rc)}}{\delta(n)}\\
&\leq \frac1{\delta(n)}\cdot\left(\ppr{\Trans_n}{\neg \Balanced_n}+ \ex{\h\la \calH,r_\Rc \la \zm}{\neg \TwoDecom^{\h}_n(r_\Rc)\mid \trans_n^{\h,r_\Rc} \in \Balanced}\right )\\
&\leq \negl(n),
\end{align*}
where the last inequality follows from \cref{eq:BreakingLowRCSHC:1,eq:BreakingLowRCSHC:2}.
\end{proof}

\subsection{The Communication Complexity Lower Bound}\label{section:ComCompLB}
In this section we give lower bound on the \emph{sender communication complexity} of black-box constructions of statistically hiding commitment from trapdoor permutations. We first give two results for the case where the reduction is to polynomially hard families. The first result is for ``security-preserving" construction, and the second one is for arbitrary one.

\begin{theorem}[restating \cref{thm:Intro:CommComplexityLB}]\label{theorem:SHCComCompLB}
In every $O(n)$-security-parameter-expanding, fully black-box construction of a weakly binding,
honest-receiver, honest-sender statistically hiding commitment scheme from a polynomially hard family of trapdoor permutations, the sender sends $\Omega(n)$ bit.
\end{theorem}

\begin{theorem}
In every fully black-box construction of a weakly binding and honest-receiver, honest-receiver statistically hiding
commitment scheme from a polynomially hard family of trapdoor permutations, the sender sends $n^{\Omega(1)}$ bits.
\end{theorem}

The above two theorems are in fact corollaries of the more general statement given below, for trapdoor permutations of arbitrary hardness.
\begin{theorem}\label{theorem:SHCComCompLBGeneral}
In every $\ell$-security-parameter-expanding fully black-box construction of a $d$-round
weakly binding and honest-receiver, and honest-receiver statistically hiding commitment scheme from an $\s(n) \geq n^{\omega(1)}$-hard
family of trapdoor permutations in which the sender communicates $c(\cdot)$ bits, it holds that $c(\ell(n)) \in \Omega (n)$.
\end{theorem}

\begin{proof}
The proof follows in large parts the proof of \cref{theorem:SHCRoundCompLB} given above, so we only mention the significant differences.

Let $(\Com= (\Sc, \Rc, \Vc))$ be an $\ell$-security-parameter-expanding fully black-box construction
of a $c$-communication-complexity, $\delta $-binding and honest-receiver, honest-sender, statistically hiding commitment scheme from an $\s$-hard family of trapdoor permutations, where $\delta(n) = 1 - 1/p(n)$ for some $p\in \poly$, and let $m = m(n)$ be a bound on the running time of \Sc and \Rc on security parameter $n$. \cref{theorem:collisionInLowComComp} yields that the following holds.
\begin{claim}\label{corollary:BreakingLowCCCom}
There exists a $d= \left(\ceil{\frac {4c}{\log s}}+1\right)$-depth, deterministic $O(\sqrt[3]{\s})$-\augQ, normal-from oracle-aided algorithm $\Ss$ such that the following holds for every $\tau\in \Tau$: for $h\in \calH$ and $r_\Sc,r_\Rc \in \zm$, let $\TwoDecom^\h_n(r_\Sc,r_\Rc)$ be one if and only if $\Vc(\com, \decom) \neq \Vc (\com, \decom') \in \zo$ for $\com =\exec{(\Sc^\tau(0,r_\Sc), \Rc(r_\Rc)^\tau)(1^n)}_{\out^\Rc}$ and $(\decom,\decom') = \Ss^{\tau,\Sam^{\tau,\h}}(\com)$, and let $\NoBreak^{\tau, \h}_n$ be the event that $\ppr{r_\Sc\la \zm, r_\Rc \la \zm}{\neg \TwoDecom^{\h}_n(r_\Sc,r_\Rc)} \geq \delta(n)$. Then $\ppr{\h\la \calH}{\NoBreak^\h_n} \in O(1/n^2)$.
\end{claim}
The proof of \cref{corollary:BreakingLowCCCom} follows similar lines to that of \cref{corollary:BreakingLowRCSHC}, see more details in \cref{sec:BreakingLowCCCom}.
Similarly to the proof of \cref{theorem:SHCRoundCompLB}, \cref{corollary:BreakingLowCCCom}  yields that there exists an $\aq = 4\s$-\augQ, normal form, $d+1$-depth, oracle-aided circuit family $\As = \set{\As_n}_{n\in \N}$ with
\begin{align}
\prob{\MyAtop{\tau \la \Tau,\h\la \calH}{td \la \zn,y\la \zn}}{\As_n^{\tau, \Sam^{\tau,\h}}(G (td), y) = F^{-1} (td, y)}
> \frac{4}{\aq(n)}
\end{align}
for infinitely many $n$'s. By \cref{theorem:TDPHardForSam}, it follows that $d(\ell(n)) \in \Omega(n/\log \s(n))$. Since, by our simplifying assumption, $\s$ is non-decreasing, it follows that $c(\ell(n)) \in \Omega\left(\frac{n\cdot \log(s(\ell(n)))}{\log s(n)}\right) \in \Omega(n)$.
\end{proof}

\subsubsection{Proving \texorpdfstring{\cref{corollary:BreakingLowCCCom}}{\cref{corollary:BreakingLowCCCom}}}\label{sec:BreakingLowCCCom}
\begin{proof}[Proof \cref{corollary:BreakingLowCCCom}]
The proof follows in large parts the proof of \cref{corollary:BreakingLowRCSHC} given above, so we only mention the significant differences.

Let \Inv be the algorithm guaranteed \cref{theorem:collisionInLowComComp} for the protocol $(\Sc, \Rc)$. Following an execution $\exec{(\Sc^\tau(0), \Rc^\tau)(1^n)}$ resulting in transcript $\trans$, algorithm $\Sc^{\tau,\Sam^{\tau,\h}}$ calls $\Inv^{\tau,\Sam^{\tau,\h}}(1^n,1^n,d(n)-1,\eps(n) = \delta(n)/n^2,\trans)$ to get set of pairs $\set{(b_i,r_i))}_{i\in [n]}$. If there exists two pairs $(0,r_0)$ and $(1,r_1)$ in the above set, $\Ss$ uses them to generate two decommitments $\decom_0$ and $\decom_1$. Note that number of augmented queries done by \Ss is bounded by $\poly(n) \cdot 2^{\ceil{c(n)/d(n)}} \leq \poly(n) \cdot \sqrt[4]{\s(n)} \in O(\sqrt[3]{\s(n)})$.

Let $\Balanced_n$ be as in \cref{corollary:BreakingLowCCCom}, and for $r_\Sc,r_\Rc \in \zm$, let $\trans_n^{r_\Sc,r_\Rc} = \exec{(\Sc^\tau(0,r_\Sc), \Rc^\tau(r_\Rc))(1^n)}_\trans$. Since $\Com$ is statistically hiding, it follows (see \cref{corollary:BreakingLowCCCom}) that
\begin{align}\label{eq:BreakingLowCCCom:1}
\ppr{r_\Sc\la \zm,r_\Rc \la \zm}{\trans_n^{r_\Sc,r_\Rc} \notin \Balanced_n} = \negl(n)
\end{align}
Let $\Fail_n = \set{(\h,r_\Sc,r_\Rc) \in \calH \times (\zm)^2 \colon \Inv^{\tau,\Sam^{\tau,\h}}(1^n,1^n,d(n)-1,\eps(n),\trans_n^{r_\Sc,r_\Rc}) = \perp}$. \cref{theorem:collisionInLowRound} yields that
\begin{align}\label{eq:BreakingLowCCCom:2}
\ppr{\h \la\calH,r_\Sc\la \zm,r_\Rc \la \zm}{\Fail_n} \leq \eps(n)
\end{align}
It is also easy to verify that (see again \cref{corollary:BreakingLowCCCom}) that
\begin{align}\label{eq:BreakingLowCCCom:3}
\ex{\MyAtop{\h \la\calH,r_\Sc\la \zm}{r_\Rc \la \zm}}{\neg \TwoDecom^{\h}_n(r_\Sc,r_\Rc)\mid \trans_n^{r_\Rc} \in \Balanced_n \land (\h,r_\Sc,r_\Rc) \notin \Fail_n} = \negl(n)
\end{align}
We conclude that
\begin{align*}
\lefteqn{\ppr{\h\la \calH}{\NoBreak^\h_n}}\\
& = \ppr{h\la \calH}{\ppr{r_\Sc \la \zm,r_\Rc \la \zm}{\neg \TwoDecom^{\h}_n(r_\Sc,r_\Rc)} \geq \delta(n)}\\
&\leq \frac{\ex{\h\la \calH,r_\Sc \la \zm,r_\Rc \la \zm}{\neg\TwoDecom^{\h}_n(r_\Sc,r_\Rc)}}{\delta(n)}\\
&\leq \frac1{\delta(n)}\cdot \left( \ppr{\h \la\calH,r_\Sc\la \zm,r_\Rc \la \zm}{\trans_n^{r_\Sc,r_\Rc} \notin \Balanced_n \lor (\h,r_\Sc,r_\Rc) \in \Fail_n} \right. \\
&\quad \left. + \ex{\h\la \calH,r_\Rc \la \zm}{\neg \TwoDecom^{\h}_n(r_\Sc,r_\Rc)\mid \trans_n^{\h,r_\Sc,r_\Rc} \in \Balanced \land (\h,r_\Sc,r_\Rc) \notin \Fail_n} \right)\\
&\leq \frac{\eps(n) + \negl(n)}{\delta(n)} \in O(1/n^2),
\end{align*}
where the last inequality follows from \cref{eq:BreakingLowCCCom:1,eq:BreakingLowCCCom:2,eq:BreakingLowCCCom:3}
\end{proof}

\newcommand{\pir}{\ensuremath{ \MathAlg{P}}\xspace}
\newcommand{\Ser}{\ensuremath{ \MathAlg{Srv}}\xspace}
\newcommand{\Sen}{\ensuremath{ \MathAlg{Snd}}\xspace}

\section{Implications to Other Cryptographic Protocols}\label{section:implications}
Our lower bounds on the round complexity and the communication complexity of statistically hiding commitment schemes imply similar lower
bounds for several other cryptographic protocols. Specifically, our results can be extended to any cryptographic
protocol that can be used to construct a weakly-binding statistically hiding
commitment scheme in a fully-black-box manner while essentially preserving the round complexity or communication complexity of the underlying protocol. In this section we derive new such lower
bound for interactive hashing, oblivious transfer, and single-server private information retrieval protocols. For simplicity, we state these lower bounds for constructions that are security preserving (\ie $O(n)$-security-parameter expanding), and we note that more general statements, as in \cref{theorem:SHCRoundCompLB}, could be derived as well.

We note that our lower bound proof for the round complexity of statistically hiding commitment schemes did not rely on any
malicious behavior by the receiver. Therefore, our lower bound
holds even for schemes in which the statistical hiding property is guaranteed only against honest receivers. Similarly, our lower bound proof for the communication complexity of statistically hiding commitment schemes did not rely on any malicious behavior by the sender during the commit stage. Therefore, our lower bound
holds even for schemes in which the (weak) binding property is guaranteed only against honest senders.

\subsection{Interactive Hashing}\label{de:InteractiveHashing}
Interactive hashing was introduced by \citet{NaorOVY98} and is
a protocol that allows a sender $\Sc$ to commit to a value $y$ while only revealing to the
receiver $\Rc$ the value $(h, z = h(y))$, where $h$ is a 2-to-1 hash function chosen
interactively during the protocol.\footnote{Several extensions to this definition were suggested,
see \cite{HaitnerR06a, NguyenOV06}.} The two security properties of interactive hashing are binding
($\Sc$ is bound by the protocol to producing at most one value of $y$ which is consistent
with the transcript) and hiding ($\Rc$ does not obtain any information about $y$, except
for $h(y)$).

Naor et al.\ constructed an interactive hashing protocol from any one-way permutation, and showed that it implies in a
fully black-box manner a statistically-hiding commitment scheme. The construction of Naor et al.\ preserves the communication complexity of the underlying interactive hashing protocol, but it does not preserve the round complexity. However, in subsequent work \cite{HaitnerR06a,KoshibaS06} it was shown that it is in fact possible to preserve the number of rounds. Combined with our lower bounds on the round complexity and communication complexity of statistically-hiding commitment schemes, this directly implies the following corollary:

\begin{corollary}\label{corollary:IH}
Any $O(n)$-security-parameter expanding fully black-box construction of an interactive hashing
protocol from a family of trapdoor permutations has round complexity $\Omega ( n / \log n )$ and communication complexity $\Omega(n)$.
\end{corollary}

We note that \citet{Wee07}  showed that a restricted class of
fully black-box constructions of interactive hashing from one-way permutations has $\Omega \left( n / \log n \right)$ rounds. Thus, \cref{corollary:IH} extends Wee's lower bound both to include the most general form of such constructions, and to trapdoor permutations.



\subsection{Oblivious Transfer}
Oblivious transfer (OT), introduced by \citet{Rabin81}, is a fundamental primitive in
cryptography. In particular, it was shown to imply secure multiparty computation
\cite{GoldreichMW87, Kilian88, Yao86}. OT has several equivalent formulations, and we consider the
formulation of $\binom{2}{1}$-OT, defined by \citet*{EvenGL85}.
$\binom{2}{1}$-OT is a protocol between two parties, a sender and a receiver. The sender's input
consists of two secret bits $(b_0, b_1)$, and the receiver's input consists of a value $i \in
\zo$. At the end of the protocol, the receiver should learn the bit $b_i$ while the sender does
not learn the value $i$. The security of the protocol guarantees that even a cheating receiver
should not be able to learn the bit $b_{1 - i}$, and a cheating sender should not be able to learn
$i$.

Given any $\binom{2}{1}$-OT protocol that guarantees statistical security for the sender, Fischlin \cite{Fischlin02} showed how to construct a weakly-binding statistically hiding commitment scheme. The construction is fully black-box and preserves the round complexity and the communication complexity. In addition, Wolf and Wullschleger \cite{WolfW06} showed that any $\binom{2}{1}$-OT protocol that guarantees statistical security for the sender can be transformed into a $\binom{2}{1}$-OT protocol that guarantees statistical security for the receiver. Their transformation is full black-box and preserves the round complexity and the communication complexity. Thus, by combining these with our lower bounds we
obtain the following corollary:

\begin{corollary}\label{corollary:OT}
Any $O(n)$-security-parameter expanding fully black-box construction of a $\binom{2}{1}$-OT
protocol that guarantees statistical security for one of the parties from a family of trapdoor
permutations has round complexity $\Omega ( n / \log n )$ and communication complexity $\Omega(n)$.
\end{corollary}

We stress that there exist constructions of semi-honest receiver $\binom{2}{1}$-OT protocols,
relying on specific number-theoretic assumptions, where the sender enjoys statistical security with
a constant number of rounds (\eg \citet{AielloIR01} and \citet{NaorP01}). Hence, as for statistically hiding commitment schemes, we demonstrate a large gap between the round complexity of OT constructions based on general assumptions and OT constructions based on
specific number-theoretic assumptions.

\subsection{Single-Server Private Information Retrieval}
A single-server private information retrieval (PIR) scheme \cite{ChorGKS95} is a protocol between a
server and a user. The server holds a database $x \in \zn$, and the user holds an index $i
\in [n]$ to an entry of the database. Informally, the user wishes to retrieve the $i$'th entry of
the database, without revealing to the server the value $i$. A naive solution is to have the user
download the entire database, however, the total communication complexity of this solution is $n$
bits. Based on specific number-theoretic assumptions, several schemes with sublinear communication
complexity were developed (see \cite{CachinMS99, Chang04, GentryR05, Lipmaa05, KushilevitzO97}, and
a recent survey by \citet{OstrovskyS07}). The only non-trivial construction
based on general computational assumptions is due to \citet{KushilevitzO00}. Assuming the existence of trapdoor permutations, they constructed an
interactive protocol whose communication complexity is $n - o(n)$ bits.

\citet*{BeimelIKM99} showed that any single-server PIR protocol
with communication complexity of at most $n / 2$ bits, can be used to construct a weakly-binding
statistically hiding commitment scheme. Their construction is fully black-box and preserves
the number of rounds. Thus, by combining this with our lower bound on the round complexity for statistically hiding commitment schemes, we obtain the following corollary:

\begin{corollary}\label{corollary:PIR}
Any $O(n)$-security-parameter expanding fully black-box construction of a single-server PIR
protocol for an $n$-bit database from a family of trapdoor permutations, in which the server
communicates less than $n / 2$ bits, has communication complexity $\Omega \left( n / \log n \right)$.
\end{corollary}

\cref{corollary:IH} yields in particular an $\Omega \left( n / \log n \right)$ lower bound on the communication complexity of such single-server PIR protocols (and, in particular, on the number of bits that the server must communicate). We note that the construction of Beimel et al.\ does not preserve the communication complexity of the underlying PIR protocol. Therefore, our lower bound on the communication complexity of statistically hiding commitment schemes cannot be directly used for deriving a similar lower bound for PIR protocols. Nevertheless, in \cref{section:PIR2Com} we refine the construction of \citeauthor{BeimelIKM99} to a construction which, in particular, preserves the communication complexity. We thus obtain the following corollary:

\begin{corollary}
In any $O(n)$-security-parameter expanding fully black-box construction of a single-server PIR
protocol for an $n$-bit database from a family of trapdoor permutations, the server communicates $\Omega(n)$ bits.
\end{corollary}

\section*{Acknowledgment}
We thank Mohammad Mahmoody and Rafael Pass for useful discussions.

\appendix
\section{From PIR to Statistically-Hiding Commitments}\label{section:PIR2Com}
The relation between single-server PIR and commitment schemes was first explored by \citet*{BeimelIKM99}, who showed that any single-server PIR protocol in which
the server communicates at most $n/2$ bits to the user (where $n$ is the size of the server's
database), can be used to construct a weakly binding statistically hiding bit-commitment scheme. In
particular, this served as the first indication that the existence of low-communication PIR
protocols implies the existence of one-way functions. In this section we refine the relation
between these two fundamental primitives by improving their reduction. Our improvements are the following:
\begin{enumerate}
\item The construction of \cite{BeimelIKM99} preserves the round complexity of the underlying
single-server PIR, but it does not preserve its communication complexity. In their construction the
sender is always required to send $\Omega(n)$ bits during the commit stage of the commitment
scheme. We show that it is possible to preserve both the round complexity and the communication
complexity. In our construction the number of bits sent by the sender during the commit stage of
the commitment scheme is essentially the number of bits sent by the server in the PIR protocol.

\item The construction of \cite{BeimelIKM99} requires an execution of the single-server PIR protocol
for every committed bit (that is, they constructed a bit-commitment scheme). We show that it is
possible to commit to a super-logarithmic number of bits while executing the underlying
single-server PIR protocol only once.

\item The construction of \cite{BeimelIKM99} was presented for single-server PIR protocols in which the
server communicates at most $n/2$ bits. Our construction applies to any single-server PIR protocol
in which the server communicates up to $n - \omega(\log n)$ bits.
\end{enumerate}

In the remainder of this section we first state the theorem resulting from our construction. Then, we formally define single-server PIR, provide a few additional preliminaries, and present our construction.

\begin{theorem}\label{theorem:COMisGood}
Assume there exists a single-server PIR protocol in which the server communicates $n - k(n)$ bits, where $n$ is the size of the server's database and $k(n) \geq 2 d(n)$ for $d(n) \in \omega( \log n)$.

Then, there exists a weakly binding statistically hiding commitment scheme for $d(n) / 6$
bits, in which the sender communicates at most $n - k(n) + 2d(n)$ bits during the commit stage.
Moreover, the construction is fully black box.
\end{theorem}

\paragraph{An overview of the construction.} Let $(\Server, \User)$ be a single-server PIR protocol in
which the server communicates $n - \omega(\log n)$ bits, where $n$ is the size of the server's
database. Consider the following commitment scheme to a string $s$. The commit stage consists of
the sender and the receiver first choosing random inputs $x \in \zn$ and $i \in [n]$,
respectively, and executing the PIR protocol $(\Server, \User)$ on these inputs (that is, the sender plays the
role of the server with database $x$, and the receiver plays the role of the user with index $i$).
As a consequence, the receiver obtains a bit $x_i$, which by the correctness of the PIR protocol is the
$i$'th bit of $x$. Notice that since the sender communicated only $n - \omega(\log n)$,
the random variable corresponding to $x$ still has $\omega(\log n)$ min-entropy from the receiver's
point of view. We take advantage of this fact, and have the sender choose a uniform seed $t$ for a
strong-extractor $\ext$, and send the pair $(t, \ext(x,t) \oplus s)$ to the receiver. That is, we
exploit the remaining min-entropy of the database $x$ in order to mask the committed string $s$ in
a statistical manner. In the reveal stage, the sender sends the pair $(x, s)$ to the receiver. The
binding property follows from the security of the PIR protocol: in the reveal stage, the
sender must send a value $x$ whose $i$'th bit is consistent with the bit obtained by the receiver
during the commit stage -- but this bit not known to the sender.

\subsection{Single-Server Private Information Retrieval --- Definition}
A single-server Private Information Retrieval (PIR) scheme is a protocol between a server and a
user. The server holds a database $x \in \zn$ and the user holds an index $i \in [n]$ to an entry
of the database. The user wishes to retrieve the $i$'th entry of the database,
without revealing the index $i$ to the server. More formally, a single-server PIR scheme is defined
via a pair of probabilistic polynomial-time Turing-machines $(\Server, \User)$ such that:
\begin{itemize}
\item $\Server$ receives as input a string $x \in \zn$. Following its interaction it
does not have any output.

\item $\User$ receives as input an index $i \in [n]$. Following its interaction it outputs
a value $b \in \{0, 1, \bot\}$.
\end{itemize}

Denote by $b \leftarrow \langle \Server(x), \User(i) \rangle$ the experiment in which $\Server$ and
$\User$ interact (using the given inputs and uniformly chosen random coins), and then
$\User$ outputs the value $b$. It is required that there exists a negligible function
$\nu(n)$, such that for all sufficiently large $n$, and for every string $x = x_1 \circ \cdots
\circ x_n \in \zn$, it holds that $x_i \leftarrow \langle \Server(x), \User(i) \rangle$ with
probability at least $1 - \nu(n)$ over the random coins of both $\Server$ and $\User$.

In order to define the security properties of such schemes, we first introduce the following
notation. Given a single-server PIR scheme $(\Server, \User)$ and a Turing-machine $\Servers$ (a
malicious server), we denote by ${\sf view}_{\langle \Servers,\User(i) \rangle}(n)$ the
distribution on the view of $\Servers$ when interacting with $\User(i)$ where $i \in [n]$. This
view consists of its random coins and of the sequence of messages it receives from $\User$,
and the distribution is taken over the random coins of both $\Servers$ and $\User$.

\begin{definition}\label{def:PIR-security}
A single-server PIR scheme $(\Server, \User)$ is secure if for every probabilistic polynomial-time
Turing-machines $\Servers$ and $\Dc$, and for every two sequences of indices $\{ i_n \}_{n =
1}^{\infty}$ and $\{ j_n \}_{n = 1}^{\infty}$ where $i_n, j_n \in [n]$ for every $n$, it holds that
\[ \left| \pr{v \leftarrow {\sf view}_{\langle \Servers,\User(i_n) \rangle}(n) \colon \Dc(v) = 1} - \pr{v \leftarrow {\sf view}_{\langle \Servers,\User(j_n)
\rangle}(n): \Dc(v) = 1} \right| \leq \nu(n) , \]%
for some negligible function $\nu(n)$ and for all sufficiently large $n$.
\end{definition}

\subsection{Additional Preliminaries}\label{section:additionalpreliminaries}
The min-entropy of a distribution $D$ over a set $\cX$ is defined as $\Hmin(D) = \min_{x \in \cX} \log
1/\prob{D}{x}$. The following standard fact (\cf \cite[Fact 2.6]{SahaiV03}) will be useful for us in analyzing
statistically close distributions.
\begin{fact}\label{fact:SV03}
Let $P$ and $Q$ be two distributions with $\SD(P,Q) <
\epsilon$, then
\begin{align*}
\Pr_{x\la P}\left[( 1 - \sqrt{\epsilon}) \cdot \prob{P}{x} < \prob{Q}{x} <( 1 + \sqrt{\epsilon}) \cdot \prob{P}{x}\right] \geq 1 - 2 \sqrt{\epsilon}.
\end{align*}
\end{fact}

\begin{definition}\label{definition:extractor}
A function $E \colon \zn \times \zo^d \rightarrow \zo^m$ is a {\sf $(k, \epsilon)$-extractor}, if
for every distribution $X$ over $\zn$ with $\Hmin(X) \ge k$ the distribution $E(X,U_d)$ is
$\epsilon$-close to uniform. $E$ is a {\sf strong $(k, \epsilon)$-extractor}, if the function $E'(x,
y) = y \circ E(x, y)$ is a $(k, \epsilon)$-extractor (where $\circ$ denotes concatenation).
\end{definition}

In our construction of a statistically hiding commitment  from single-server PIR, we will be
using the following explicit construction of strong extractors, which is an immediate corollary of
\cite[Corollary 3.4]{SrinivasanZ99}.

\begin{proposition}\label{proposition:Extractor}
For any $k \in \omega( \log n )$, there exists an explicit strong $(k, 2^{1-k})$-extractor $\ext \colon \zn \times \zo^{3k} \rightarrow \zo^{k/2}$.
\end{proposition}

\subsection{The Construction}
Fix $d(n)$, $k(n)$ and a single-server PIR protocol $\pir = (\Server, \User)$ as in \cref{theorem:COMisGood}.
\cref{fig:PIR2Com} describes our construction of the commitment scheme $\Com = (\Sc, \Rc)$. In
the construction we use a strong $\left( d(n)/3, 2^{1 - d(n)/3} \right)$-extractor
$\ext \colon \zn \times \zo^{d(n)} \rightarrow \zo^{d(n) / 6}$ whose existence is
guaranteed by \cref{proposition:Extractor}.

\begin{protocol}[Protocol $\Com = \left( \Sc, \Rc \right)$]\label{fig:PIR2Com}~ 

\item Common input: security parameter $1^n$.

\item Sender's input: $s \in \zo^{d(n) / 6}$.
\medskip
\item Commit stage:
\begin{enumerate}
\item $\Sc$ chooses a uniformly distributed $x \in \zn$.

\item $\Rc$ chooses a uniformly distributed index $i \in [n]$.

\item $\Sc$ and $\Rc$ execute the single-server PIR protocol $(\Server, \User)$ for database of length $n$, where
$\Sc$ acts as the server with input $x$ and $\Rc$ acts as the user with input $i$. As a result, $\Rc$
obtains a bit $x_i \in \zo$.

\item $\Sc$ chooses a uniformly distributed seed $t \in \zo^{d(n)}$, computes $y = \ext(x, t) \oplus s$,
and sends $(t, y)$ to $\Rc$.
\end{enumerate}

\medskip
\item Reveal stage:

\begin{enumerate}
\item $\Sc$ sends $(s, x)$ to $\Rc$.

\item If the $i$'th bit of $x$ equals $x_i$ and $y = \ext(x, t) \oplus s$, then $\Rc$ outputs $s$.

Otherwise, $\Rc$ outputs $\bot$.
\end{enumerate}
\end{protocol}

The correctness of $\Com$ follows directly from the correctness of the PIR protocol. In addition, notice that
the total number of bits communicated by the sender in the commit stage is the total number of bits
that the server communicates in the PIR protocol plus the seed length and the output length of the extractor
$\ext$. Thus, the sender communicates less than $n - k(n) + 2 d(n)$ bits during the commit stage.
In \cref{lemma:COMisHiding} we prove that $\Com$ is statistically hiding, and in \cref{lemma:COMisBinding} we prove that $\Com$ is weakly binding. We note that the proof of hiding
does not rely on any computational properties of the underlying PIR protocol, but only on
the assumed bound on the number of bits communicated by the server.

\begin{lemma}\label{lemma:COMisHiding}
$\Com$ is statistically hiding.
\end{lemma}

\begin{proof}
We have to show that for any computationally unbounded receiver $\Rc^*$ and for any two strings
$s_0$ and $s_1$, the statistical distance between the distributions $\{ {\sf view}_{\langle
\Sc(s_0),\mathcal{R}^* \rangle}(n) \}$ and $\{ {\sf view}_{\langle \Sc(s_1),\mathcal{R}^* \rangle}(n)
\}$ (see \cref{def:binding}) is negligible in $n$. The transcript of the commit
stage consists of the transcript $\trans_{\pir}$ of the execution of $\pir$ and of the pair $(t,
\ext(x, t) \oplus s)$, where $s$ is the committed string. Note that since $\trans_{\pir}$ is
independent of the committed string, it is sufficient to prove that the statistically distance
between the distribution of $(t, \ext(x, t))$ given $\trans_{\pir}$ and the uniform distribution is
negligible in $n$.

We argue that due to the bound on the number of bits communicated by the server in $\pir$, then
even after executing $\pir$, the database $x$ still has sufficient min-entropy in order to
guarantee that $(t, \ext(x, t))$ is sufficiently close to uniform. More specifically, let $\Rc^*$ be
an all-powerful receiver (recall that without loss of generality such an $\Rc^*$ is deterministic),
and denote by $X$ the random variable corresponding to the value $x$ in $\Com$. The
following claim states the with high probability $X$ has high min-entropy from $\Rc^*$'s point of
view.

\begin{claim}\label{Claim:hiding}
It holds that
\[ \prob{\trans_{\pir} \la \Com}{\Hmin(X \mid \trans_{\pir})< \frac{k(n)}{6} } < 2^{- \frac{k(n)}{4}}
, \]%
where $\trans_{\pir}$ is the transcript of the embedded execution of $\pir$ in $\Com$.
\end{claim}

\begin{proof}
For any value of $r$, the random coins used by $\Sc$ in the execution of $\pir$, let $f_{r}:\zn
\mapsto \zo^{n-k(n)}$ be the function that maps $x$ to the value of $\trans_{\pir}$ generated by
the interaction of $(\Sc(x,r),\Rc^*)$, and let ${\rm Col}(x,r) \eqdef \set{x'\in \zn: f_r(x') =
f_r(x)}$. Since $f_r$ has at most $2^{n-k(n)}$ possible outputs, it follows that
\begin{align}\label{Equation:Hiding1}
\prob{x, r}{\size{{\rm Col}(x, r)} < 2^{\frac{k(n)}{2} + 1}} < \frac{2^{n - k(n)} \cdot
2^{\frac{k(n)}{2} + 1}}{2^n} = 2^{1 - \frac{k(n)}{2}} .
\end{align}%
Let
\[ {\rm BAD} = \left\{ \trans_{\pir} \enspace \colon \enspace \prob{x, r}{\left. \size{{\rm Col}(x, r)}
< 2^{\frac{k(n)}{2} + 1} \mbox{ } \right| \mbox{ } \trans_{\pir} } > 2^{\frac{k(n)}{4}} \cdot 2^{1
- \frac{k(n)}{2}} \right\} , \]%
a a standard averaging argument yields that
\begin{align}
  \prob{\trans_{\pir} \la \Com}{\trans_{\pir} \in {\rm BAD}} \leq 2^{- \frac{k(n)}{4}}
\end{align}

Denote by $U_r$ the random variable corresponding to $r$ in the execution of $\Com$. The
following holds every value of $x$ and $\trans_{\pir}$:
\begin{align}
& \pr{X=x \enspace | \enspace \trans_{\pir}} \label{Equation:Hiding2} \\
& \quad = \pr{\left. X=x \land \size{{\rm Col}(X,U_r)} < 2^{\frac{k(n)}{2}+1} \enspace \right|
\enspace \trans_{\pir}} \nonumber \\
& \quad \quad + \pr{\left. X=x \land \size{{\rm Col}(X,U_r)} \geq 2^{\frac{k(n)}{2}+1}
\enspace \right| \enspace \trans_{\pir}} \nonumber \\
& \quad \leq \pr{\left. \size{{\rm Col}(X,U_r)} < 2^{\frac{k(n)}{2}+1} \enspace \right| \enspace
\trans_{\pir}} + 2^{- \left( \frac{k(n)}{2}+1 \right)} . \nonumber
\end{align}

Note that if $\Hmin(X \mid \trans_{\pir}) < k(n)/6$ for some $\trans_{\pir}$, then there
exists an $x$ for which
\[ \pr{X = x \enspace | \enspace \trans_{\pir}} \ge 2^{- \frac{k(n)}{6}} , \]%
and therefore \cref{Equation:Hiding2} implies that
\begin{align}
\pr{\left. \size{{\rm Col}(X,U_r)} < 2^{\frac{k(n)}{2}+1} \enspace \right| \enspace \trans_{\pir}}
> 2^{- \frac{k(n)}{6}} - 2^{- \left( \frac{k(n)}{2}+1 \right)} > 2^{1 - \frac{k(n)}{4}}
\end{align}

Thus,
\begin{align*}
\prob{\trans_{\pir} \la \Com}{\Hmin(X \mid \trans_{\pir}) < \frac{k(n)}{6} } &  \leq \prob{\trans_{\pir} \la \Com}{\pr{\left. \size{{\rm Col}(X,U_r)} <
2^{\frac{k(n)}{2}+1} \enspace \right| \enspace \trans_{\pir}} > 2^{1 - \frac{k(n)}{4}}} \\
&   \leq \prob{\trans_{\pir} \la \Com}{\trans_{\pir} \in {\rm BAD}} \\
&   \leq 2^{- \frac{k(n)}{4}}.
\end{align*}
\end{proof}

Since $d(n) \in \omega(\log n)$ and $k(n)/6 \ge d(n)/3$, \cref{Claim:hiding} implies that
with probability $1 - {\rm neg}(n)$, the extractor $\ext$ guarantees that the statistical distance
between the pair $(t, \ext(x, t))$ (given $\trans_{\pir}$) and the uniform distribution is at most
$2^{1 - d(n)/3}$ (which is again negligible in $n$). Therefore $\Com$ is
statistically hiding. More specifically, for every string $s \in \zo^{d(n)/6}$ it holds that
\begin{align}
& \SD \left( \{ \trans_{\pir}, t , \ext(X,t) \oplus s \}, \{ \trans_{\pir}, U_{7 d(n) / 6}
\} \right) \\
& \quad \quad \leq \pr{\Hmin(X \mid \trans_{\pir}) < \frac{k(n)}{6}} \nonumber\\
& \quad \quad \quad + \SD \left( \{ \trans_{\pir}, t , \ext(X,t) \oplus s \}, \{
\trans_{\pir}, U_{7 d(n) / 6} \} \enspace \left|\enspace \Hmin(X \mid \trans_{\pir}) \ge \frac{k(n)}{6} \right. \right)\nonumber \\
& \quad \quad \leq 2^{- \frac{k(n)}{4}} + 2^{1 - \frac{d(n)}{3}}.\nonumber
\end{align}%
Therefore, for any two strings $s_0, s_1 \in \zo^{d(n)/6}$ it holds  that
\begin{align*}
& \SD\left( \left\{ {\sf view}_{\langle \Sc(s_0),\mathcal{R}^* \rangle}(n) \right\}, \left\{
{\sf view}_{\langle \Sc(s_1),\mathcal{R}^* \rangle}(n) \right\} \right) \\
& \quad \quad = \SD \left( \{
\trans_{\pir},
t, \ext(X,t) \oplus s_0 \}, \{ \trans_{\pir}, t, \ext(X,t) \oplus s_1 \} \right) \\
& \quad \quad \leq 2 \cdot \left( 2^{- \frac{k(n)}{4}} + 2^{1 - \frac{d(n)}{3}} \right)
,
\end{align*}%
which is negligible in $n$ as required.
\end{proof}

Let $U_r$ be the random variable taking the value of $r$ in the execution of $\Com$. By the above
equation, the following holds every value of $x$ and $\trans_{\pir}$.
\begin{align*}
& \Pr[X=x \mid \trans_{\pir}] \\
& = \Pr[X=x \land \size{{\rm Col}(X,U_r)} < 2^{\frac{k(n)}{2}+1}
\mid \trans_{\pir}] + \Pr[X=x \land \size{{\rm Col}(X,U_r)} \geq 2^{\frac{k(n)}{2}+1} \mid \trans_{\pir}] \\
& \leq  \Pr[\size{{\rm Col}(X,U_r)} < 2^{\frac{k(n)}{2}+1} \mid \trans_{\pir}] +
2^{-(\frac{k(n)}{2}+1)} .
\end{align*}

We conclude that,
\begin{align*}
&\Pr[\trans_{\pir} \la \Com\colon \Hmin(X\mid \trans_{\pir})<\frac{k(n)}{2}]\\
& \quad = \Pr\bigl[\trans_{\pir} \la \Com\colon\max_{x\in \zn} \set{\Pr[X=x \mid \trans_{\pir}]} > 2^{\frac{k(n)}{2}}\bigl]\\
& \quad \leq \Pr\bigl[\trans_{\pir} \la \Com\colon\Pr[\size{{\rm Col}(X,U_r)} < 2^{\frac{k(n)}{2}+1} \mid \trans_{\pir}]\bigl]\\
& \quad =\Pr[\size{{\rm Col}(X,U_r)} < 2^{\frac{k(n)}{2}+1}] < 2^{1-k(n)/2}.
\end{align*}

Recall that $\Rc^*$'s view is the concatenation of the values of $\trans_{\pir}$, $\ext(x,t) \oplus
s$ and $t$. Using standard reduction it follows that for any two strings $s_1,s_2\in
\zo^{\floor{d(n)/2}}$, the statistical difference between $\view[s_1]$ and $\view[s_2]$ is at most
twice the statistical difference between $(\trans_{\pir},\ext(x,t),t)$ and
$(\trans_{\pir},U_{\floor{d(n)/2}},t)$, where the values of $\trans_{\pir}$, $x$ and $t$ are
induced by a random execution of $\Com$. The following concludes the proof of the lemma by showing
that the latter distance is negligible.
\begin{align*}
& \SD\biggr((\trans_{\pir},\ext(X,t),t),(\trans_{\pir}, U_{\floor{d(n)/2}},t)\biggr)\\
& \quad \leq \Pr\biggr[\Hmin(X\mid \trans_{\pir})< \frac{k(n)}{2}\biggr] \\
& \quad \quad + \SD\biggr((\trans_{\pir},\ext(X,t),t),(\trans_{\pir}, U_{\floor{d(n)/2}},t) \mid
\Hmin(X\mid \trans_{\pir})\geq \frac{k(n)}{2}\biggr)\\
& \quad \leq 2^{1-\frac{k(n)}2} + 2^{2-\frac{d(n)}3} = \operatorname{neg}(n).
\end{align*}

\begin{lemma}\label{lemma:COMisBinding}
$\Com$ is weakly binding.
\end{lemma}

\begin{proof}
We show that $\Com$ is $(1 - 1/n^2)$-binding. Given any malicious sender $\Ss$ that
violates the binding of the commitment scheme $\Com$ with probability at least $1 - 1/n^2$, we
construct a malicious server $\Servers$ that breaks the security of the single-server PIR protocol
$\pir$.

Let $\Ss$ be a polynomial-time malicious sender that violates the binding of $\Com$ with
probability at least $1 - 1/n^2$. As an intermediate step, we first construct a malicious server
that has a non-negligible advantage in predicting a uniformly chosen index held by the user in
$\pir$. More specifically, we construct a malicious server $\Servers$ and a predictor $\Dc'$ such that
\[ \pr{v \leftarrow {\sf view}_{\langle \Servers,\User(i) \rangle}(n) \colon
\Dc'(v) = i} \ge \frac{1}{n} + \frac{1}{n^2} , \]%
where the probability is taken over the uniform choice of $i \in [n]$ and over the coin tosses of
$\Servers$, $\Dc'$ and $\User$. Recall that ${\sf view}_{\langle \Servers,\User(i) \rangle}(n)$
denotes the distribution on the view of $\Servers$ when interacting with $\User(i)$ where $i
\in [n]$. This view consists of its random coins and of the sequence of messages it receives from
$\User$.

The malicious server $\Servers$ follows the malicious sender $\Ss$ in the embedded execution of
$\pir$ in $\Com$. Following the interaction, $\Servers$ proceeds the execution of $\Ss$ to obtain
a pair $(t, y)$ and two decommitments $(x_1, s_1)$ and $(x_2, s_2)$. If $x_1 = x_2$, then $\Servers$
fails. Otherwise, denote by $j \in [n]$ the minimal index such that $x_1[j] \neq x_2[j]$. Now, the
predictor $\Dc'$ outputs a uniformly distributed value $i'$ from the set $[n] \setminus \{ j \}$.

In order to analyze the success probability in predicting $i$, note that if $(x_1, s_1)$ and $(x_2,
s_2)$ are valid decommitments and $s_1 \neq s_2$ (\ie $\Ss$ broke the binding of $\Com$), then
it must hold that $x_1 \neq x_2$. In this case, let $j \in [n]$ be the minimal index such that
$x_1[j] \neq x_2[j]$, then it must be the case that $i \neq j$, as otherwise $\Rc$ will not accept
the two decommitments. Therefore, when the predictor $\Dc'$ outputs a uniformly distributed $i' \in
[n] \setminus \{ j \}$, it will output $i$ with probability $1 / (n-1)$. Thus,
\begin{align}
\pr{v \leftarrow {\sf view}_{\langle \Servers,\User(i) \rangle}(n) \colon \Dc'(v) = i} & \ge
 \left( 1 - \frac{1}{n^2} \right) \cdot \frac{1}{n-1} \\
& =  \frac{n + 1}{n^2} \nonumber\\
& =  \frac{1}{n} + \frac{1}{n^2}\nonumber.
\end{align}

In the remainder of the proof we apply a rather standard argument in order to be fully consistent
with \cref{def:PIR-security} of the security of single-server PIR. That is, we
show that there exists a pair of indices $i, j \in [n]$, a malicious server $\Servers$ and a
distinguisher $\Dc$ such that
\begin{align}
\left| \pr{v \leftarrow {\sf view}_{\langle \Servers,\User(i) \rangle}(n) \colon
\Dc(v) = 1} - \pr{v \leftarrow {\sf view}_{\langle \Servers,\User(j)
\rangle}(n)\colon \Dc(v) = 1} \right| \ge \frac{1}{p(n)} ,
\end{align}
for some polynomial $p(n)$. We prove that this holds for independently and uniformly chosen $i, j
\in [n]$ (and therefore there exist $i$ and $j$ for which this holds) where $\Servers$ is the
malicious server described above, and $\Dc = \Dc_{i,j}$ is a distinguisher that uses $\Dc'$ as
follows:
\begin{itemize}
\item If $\Dc'$ outputs $i$, then $\Dc$ outputs $1$.

\item If $\Dc'$ outputs $j$, then $\Dc$ outputs $0$.

\item Otherwise, $\Dc$ outputs a uniformly distributed $b \in \zo$.
\end{itemize}
It follows that
\begin{align}
& \pr{v \leftarrow {\sf view}_{\langle \Servers,\User(i) \rangle}(n) \colon \Dc (v) = 1} \\
& \quad \quad = \pr{v \leftarrow {\sf view}_{\langle \Servers,\User(i) \rangle}(n) \colon \Dc'(v) =
i}\nonumber \\
& \quad \quad \quad + \frac{1}{2} \cdot \pr{v \leftarrow {\sf view}_{\langle \Servers,\User(i) \rangle}(n) \colon
\Dc'(v) \notin \{i,j\}} \nonumber\\
& \quad \quad \ge \frac{1}{n} + \frac{1}{n^2} + \frac{1}{2} \cdot \pr{v \leftarrow {\sf
view}_{\langle \Servers,\User(i) \rangle}(n) \colon \Dc'(v) \notin \{i,j\}},\nonumber
\end{align}%
and
\begin{align}
& \pr{v \leftarrow {\sf view}_{\langle \Servers,\User(j) \rangle}(n) \colon \Dc (v) = 1} \\
& \quad \quad = \pr{v \leftarrow {\sf view}_{\langle \Servers,\User(j) \rangle}(n) \colon \Dc'(v) =
i} \nonumber\\
& \quad \quad \quad + \frac{1}{2} \cdot \pr{v \leftarrow {\sf view}_{\langle \Servers,\User(j) \rangle}(n) \colon
\Dc'(v) \notin \{i,j\}}\nonumber \\
& \quad \quad = \frac{1}{n} + \frac{1}{2} \cdot \pr{v \leftarrow {\sf view}_{\langle
\Servers,\User(j) \rangle}(n) \colon \Dc'(v) \notin \{i,j\}},\nonumber
\end{align}%
where the last equality holds since both $i$ and $j$ are independently chosen. Finally, note that
\[ \pr{v \leftarrow {\sf
view}_{\langle \Servers,\User(i) \rangle}(n) \colon \Dc'(v) \notin \{i,j\}} = \pr{v \leftarrow {\sf
view}_{\langle \Servers,\User(j) \rangle}(n) \colon \Dc'(v) \notin \{i,j\}} , \]%
and we conclude that
\[ \left| \pr{v \leftarrow {\sf view}_{\langle \Servers,\User(i) \rangle}(n) \colon
\Dc(v) = 1} - \pr{v \leftarrow {\sf view}_{\langle \Servers,\User(j)
\rangle}(n)\colon \Dc(v) = 1} \right| \ge \frac{1}{n^2} . \]%
\end{proof}

\end{document}